\documentclass[a4paper]{article}
\usepackage[a4paper,top=3cm,bottom=3cm,left=2.1cm,right=2.1cm,marginparwidth=1.75cm]{geometry}
\usepackage[utf8]{inputenc}
\usepackage{bm}
\usepackage{amsmath, amsfonts, amssymb, amsthm}
\usepackage{mathrsfs, dsfont}
\usepackage[dvipsnames]{xcolor}
\usepackage{graphicx}
\usepackage{float}
\usepackage{subfig}
\usepackage{bigints}
\def\twoplotswidth{0.48\linewidth}
\def\threeplotswidth{0.3\linewidth}

\usepackage[colorlinks=true,allcolors=blue]{hyperref}
\usepackage{xparse} 
\usepackage{hyperref}
\usepackage{multirow}

\usepackage{tasks}
\settasks{style=itemize}

\usepackage[numbers]{natbib}
\usepackage{pifont}

\newtheorem{theorem}{Theorem}[section]
\newtheorem{lemma}[theorem]{Lemma}
\newtheorem{definition}[theorem]{Definition}

\newtheorem{proposition}[theorem]{Proposition}

\usepackage{algorithm}
\usepackage{algpseudocode}

\theoremstyle{remark}
\newtheorem{remark}{Remark}[theorem]
\newtheorem{example}{Example}[theorem]

\numberwithin{equation}{section}

\newenvironment{sqremark}{\begin{remark}}{\hfill \tiny $\blacksquare$ \end{remark}}
\newenvironment{sqexample}{\begin{example}}{\hfill \tiny $\blacksquare$ \end{example}}

\newcommand{\R}{\mathbb{R}}
\newcommand{\N}{\mathbb{N}}
\newcommand{\C}{\mathbb{C}}
\newcommand{\E}{\mathbb{E}}
\def\d{\mathrm{d}}
\def\P{\mathbb{P}}
\newcommand{\Q}{\mathbb{Q}}
\newcommand{\F}{\mathcal{F}}
\renewcommand{\Re}{\, \mathfrak{Re}}
\renewcommand{\Im}{\, \mathfrak{Im}}

\newcommand{\alphabet}[1][d]{{A}_{#1}}
\newcommand{\TA}[1][d]{T(\R^{#1})}
\newcommand{\eTA}[1][d]{T((\R^{#1}))}
\newcommand{\tTA}[2][d]{T^{#2}(\R^{#1})}

\newcommand{\word}[1]{{\mathcolor{NavyBlue}{\mathbf{#1}}}}
\newcommand{\emptyword}{{\color{NavyBlue}\textup{\textbf{\o{}}}}}
\newcommand{\proj}[1]{|_{\word{#1}}}

\newcommand{\conpow}[1]{^{\otimes #1}}
\newcommand{\coninv}[1]{\left( #1 \right)^{-1}}

\newcommand{\shuprod}{\mathrel{\sqcup \mkern -3.2mu \sqcup}}
\newcommand{\shupow}[1]{^{\shuprod #1}}
\newcommand{\shuexp}[1]{e \shupow{#1}}

\newcommand{\A}{\mathcal{A}}
\newcommand{\I}{\mathcal{I}}
\newcommand{\norm}[1]{|| #1 ||}

\NewDocumentCommand{\sigX}{O{t} O{X}}{\mathbb{#2}_{#1}}
\NewDocumentCommand{\sig}{O{t} O{W}}{\widehat{\mathbb{#2}}_{#1}}
\NewDocumentCommand{\sigE}{O{t} O{W}}{\E[\sig[#1][#2]]}

\NewDocumentCommand{\bracketsigX}{O{t} O{X} m}{\left \langle #3, \sigX[#1][#2] \right \rangle} 
\NewDocumentCommand{\bracketsig}{O{t} O{W} m}{\left \langle #3, \sig[#1][#2] \right \rangle}   
\NewDocumentCommand{\bracketsigtrunc}{O{M} O{t} O{W} m}{\left \langle #4, \sig[#2][#3]^{\leq #1} \right \rangle}   
\NewDocumentCommand{\bracketsigE}{O{t} O{W} m}{\left \langle #3, \sigE[#1][#2] \right \rangle} 

\newcommand{\bsigma}{\bm{\sigma}}
\newcommand{\bpsi}{\bm{\psi}}
\newcommand{\bgamma}{\bm{\gamma}}
\newcommand{\bell}{\bm{\ell}}
\newcommand{\bp}{\bm{p}}
\newcommand{\bq}{\bm{q}}

\newcommand{\sign}[1]{\textnormal{sign} \left( #1 \right)}

\usepackage{mathtools}
\usepackage{authblk}
\mathtoolsset{showonlyrefs}

\usepackage{enumitem}

\usepackage[normalem]{ulem}

\title{Signature volatility models:\ pricing and hedging with Fourier}

\author[1]{Eduardo Abi Jaber \thanks{eduardo.abi-jaber@polytechnique.edu. The first author is grateful for the financial support from the Chaires FiME-FDD, Financial Risks, Deep Finance \& Statistics and Machine Learning and systematic methods in finance at École Polytechnique.}}
\author[2,3]{Louis-Amand Gérard \thanks{louis-amand.gerard@etu.univ-paris1.fr. We would like to thank Olivier Guéant for fruitful discussions, and Yuxing Huang for his precious comments.}}
\affil[1]{École Polytechnique, CMAP}
\affil[2]{Université Paris 1 Panthéon-Sorbonne, CES}
\affil[3]{Gefip}
\date{February 10, 2025}

\begin{document}
\maketitle

\begin{abstract}
    We consider a stochastic volatility model where the dynamics of the volatility are given by a possibly infinite linear combination of the elements of the time extended signature of a Brownian motion. First, we show that the model is remarkably universal, as it includes, but is not limited to, the celebrated Stein-Stein, Bergomi, and Heston models, together with some path-dependent variants. Second, we derive the joint characteristic functional of the log-price and integrated variance provided that some infinite-dimensional extended tensor algebra valued Riccati equation admits a solution. This allows us to price and (quadratically) hedge certain European and path-dependent options using Fourier inversion techniques. We highlight the efficiency and accuracy of these Fourier techniques in a comprehensive numerical study.
\end{abstract}

\begin{description} 
    \item[MSC 2020:] 60L10, 91G20, 91G60 
    \item[Keywords:]{stochastic volatility, path signature, pricing, hedging, calibration, Fourier methods}
\end{description}

\bigskip

\section{Introduction}

    An important challenge of stochastic volatility modeling is the construction of realistic models that remain tractable for option pricing, risk hedging, trading and calibration purposes. Notable realistic features encompass various aspects, such as inter-temporal and path dependencies, which are inherent phenomena in financial markets. These phenomena have been established empirically at different timescales through extensive research, either in the form of long/short range dependence since \citet{mandelbrot1968fractional} and as documented in \cite{andersen1997intraday, comte1998long, gatheral2018volatility, guyon2022volatility, cont-stylizedfacts} or in the form of self-excitation of financial markets \cite{bacry2015hawkes}. They can also be understood more strategically using the transitory nature of the decisions of market participants and their impact on prices \cite{bouchaud2003fluctuations, bouchaud2009markets}.

    ~\\
    Incorporating path dependencies and memory effects into the modeling framework gives rise to non-Markovian models, which generally pose computational challenges. Recent developments have identified interesting mathematical classes of models that address this issue. One such class comprises stochastic Volterra models, which provide a flexible framework for capturing certain inter-temporal dependencies while maintaining computational tractability under specific affine and quadratic structures \cite{abi2022characteristic, abi2019affine, cuchiero2019markovian, cuchiero2020generalized, el2019characteristic}. Another avenue worth exploring lies in the application of signature-based methods. The signature of a path, initially introduced by Chen \cite{chen1957integration} in 1957, consists of the (infinite) sequence of iterated integrals of a path. It plays a crucial role in the theory of rough paths \cite{friz_victoir_2010, lyons2014rough} and is recently gaining considerable momentum in the fields of Machine Learning \cite{chevyrev2016primer, fermanian2021embedding, arribas-bipolardisorder} and Mathematical Finance \cite{arribas2020sig, bayer2023optimal, bayer2023primal, buehler2020generating, arribas-optiexec, cuchiero2022theocalib, cuchiero2023spvix, dupire2022functional, arribas-nonparam}, namely due to its universal linearization property: any functional of the path can be approximated by a linear combination of the elements of the signature of the path, provided some regularity. In our present work, we explore the modeling and numerical aspects of these signature-based approaches for stochastic volatility modeling. We uncover their potential in addressing the computational challenges associated with path-dependent models and we highlight their versatility and universality. 
    
    ~\\
    We consider a class of stochastic volatility models for a stock price $S$ in the form 
    $$ \d S_t = S_t \Sigma_t \left( \rho \d W_t + \sqrt{1 - \rho^2} \d W_t^\perp \right), $$
    where the stochastic volatility process $\Sigma$ is of a general path-dependent form
    $$ f(t, W_{0 \leq s \leq t}), \quad t \leq T, $$
    for a certain class of measurable functions $f$ and $(W, W^\perp)$ a two-dimensional Brownian motion. The correlation $\rho \in [-1, 1]$ accounts for the leverage effect. More precisely, we assume that the volatility process $\Sigma_t$ is a (possibly infinite) linear combination of the elements of the signature process $\sig$ of the time extended Brownian motion $\widehat{W}_t := (t, W_t)$ defined by the infinite sequence of iterated \citet{stratonovich} integrals: 
    
    \begin{align} \label{eq:introsig}
        \sig &
        = \left( 1,
        \begin{pmatrix}
            t \\
            W_t
        \end{pmatrix},
        \begin{pmatrix}
            \frac{t^2}{2!} & \int_0^t s \d W_s \\
            \int_0^t W_s \d s & \frac{W_t^2}{2!}
        \end{pmatrix},
       \cdots \right).
    \end{align}
    We call such models signature volatility models. They have been introduced by \citet{arribas2020sig}, with $\Sigma_t$ a finite linear combination of elements of $\sig$, and their theoretical and empirical properties have been studied further in \citet{cuchiero2022theocalib,cuchiero2023spvix} for pricing and calibration purposes. They allow to naturally embed inter-temporal and path dependencies. The construction is very elementary, yet, from the mathematical perspective, such class of models enjoys a beautiful and powerful universality feature. So far, the universal approximation property of signatures have been invoked to argue that “the framework is universal in the sense that classical models can be approximated arbitrarily well” by signature volatility models, see for instance \cite[proposition 2.14]{cuchiero2022theocalib}. In contrast to the existing literature on signature volatility models, we allow $\Sigma_t$ to be an infinite linear combination of the elements of the signature process $\sig$. This introduces intricate theoretical challenges, particularly regarding convergence issues, but provides a more profound comprehension of the elegant universal structure inherent to these models. Our approach draws inspiration from \citet{cuchiero2023polynomial} where infinite linear combinations of signature elements are considered in the context of signature stochastic differential equations.
    
    \paragraph{Universality and flexibility of signature volatility models.} In a first step, by considering infinite linear combinations of signature elements, we go beyond the ‘approximated universality’ and we prove that the class of signature volatility models is universal in the sense of exact representations. We show that many classical and popular Markovian models, and even more advanced not necessarily Markovian models, belong to the class of signature volatility models, using novel exact representations formulas derived in the sequel and in the accompanying paper \cite{linearfbm}. This includes:
    
    \begin{enumerate}[label=(\roman*)]
        \item \label{introi}
        \textbf{Affine Markovian models:} The models of \citet{stein-stein}, \citet{schobel1999stochastic} and \citet{heston} which became popular because of the explicit knowledge of the characteristic function of the log-price, allowing for fast and accurate pricing and hedging using Fourier inversion techniques. 
        
        \item \label{introii}
        \textbf{Non-affine Markovian models:} The models of \citet{BergomiSmileII} and \citet{hull-white}, which are more flexible than their affine counterparts but less tractable. In addition, the recently introduced Quintic Ornstein-Uhlenbeck model of \citet{quintic}, which is able to jointly capture SPX and VIX smiles, also belongs to the class of signature volatility models. 
        
        \item \label{introiii} 
        \textbf{Non-Markovian models:} A large class of models based on stochastic delayed equations and Volterra processes, including the class of polynomial Gaussian volatility models \cite{abi2022joint}, Volterra Stein-Stein model \cite{abi2022characteristic}, the Volterra \cite{abi2024volatility} and rough \cite{bayer2016pricing} Bergomi models.
    \end{enumerate}
    
    The above list is far from being exhaustive and we believe that more known and important volatility models can be exactly re-written as a signature stochastic volatility model. In addition, when such exact representation cannot be found for a specific model, building an approximated signature volatility model is possible thanks to the universal approximation property of the signature. The representations are derived in Section \ref{sec:model} and illustrated on numerical examples.
    
    \paragraph{Tractability of signature volatility models.} In a second step, we develop a generic framework for pricing and quadratic hedging certain vanilla and path-dependent options on the log-price and the integrated variance using Fourier inversion technology. More specifically, we obtain in Theorem~\ref{theo:charfun} that for any signature volatility model, the joint characteristic functional of the log-price and the integrated variance is known up to the solution of an infinite dimensional Riccati equation. This result opens the door to fast and accurate Fourier pricing and hedging going beyond the standard affine classes \ref{introi} including the whole list of Markovian and non-Markovian models in \ref{introii}-\ref{introiii} above, for which pricing and hedging is an intricate task.  
    Our representation formula for the characteristic function can be directly related to the ones that have appeared in \citet{cuchiero2023polynomial}, and share similarities with the formulas in \citet{friz2022forests} and \citet{lyons2024pde}. \\
    
    Using a comprehensive numerical study we highlight the efficiency and the accuracy of the Fourier techniques for pricing and hedging in signature volatility models{ in Sections \ref{sec:pricing} and \ref{sec:hedging}}. We stress that the numerics are not straightforward, since they involve several approximations and truncations. We use ideas in the spirit of ‘control variate’ with Black-Scholes prices and deltas to stabilize the Fourier inversions and reduce the number of evaluations of the characteristic functional. We also point out that the proposed implementation is both generic and scalable, as it can be applied, or fine-tuned if needed, to any signature volatility model. It only requires as inputs the coefficients (i.e.~the parameters) of the signature volatility model, which essentially define the model itself whether it is Markovian/non-Markovian/non-semimartingale, and it generates as outputs option prices and hedging strategies by Fourier methods. \\
    
    \textbf{Outline.}
    Section \ref{sec:signature} introduces the framework of signatures with a focus on infinite linear combinations of signature elements. In Section \ref{sec:model}, we define signature volatility models and study their representation properties. Section \ref{sec:charfun} derives the characteristic functional of the log-price and integrated variance. As applications, pricing and hedging of various European and Asian options under the signature volatility model are implemented in Sections \ref{sec:pricing} and \ref{sec:hedging}, where we also include calibration examples on simulated and market data. Finally, Appendix \ref{apn:linrep} collects some proofs for the representations found in Subsection \ref{subsec:representations}.

\section{A primer on signatures} \label{sec:signature}

\subsection{Tensor algebra}

    In this section, we setup the framework for dealing with signatures of semimartingales. One can also refer to the first sections in \cite{bayer2023optimal, cuchiero2023spvix, arribas-nonparam}. \\
    
    Let $d \in \N$ and denote by $\otimes$ the tensor product over $\R^d$, e.g. $(x \otimes y \otimes z)_{ijk} = x_i y_j z_k$, for $i, j, k = 1, \dots, d$, for $x, y, z \in \R^d$. For $n\geq 1$, we denote by $(\R^d) \conpow{n}$ the space of tensors of order $n$ and by $(\R^d) \conpow{0} = \R$. In the sequel, we will consider mathematical objects, path signatures, that live on the extended tensor algebra space $\eTA $ over $\R^d$, that is the space of (infinite) sequences of tensors defined by
    
    $$ \eTA := \left\{ \bell = (\bell^n)_{n=0}^\infty : \bell^n \in (\R^d) \conpow{n} \right\}. $$
    
    Similarly, for $M \geq 0$, we define the truncated tensor algebra $\tTA{M}$ as the space of sequences of tensors of order at most $M$ defined by
    
    $$ \tTA{M} := \left\{ \bell \in \eTA : \bell^n = 0, \text{ for all } n > M \right\}, $$
    
    and the tensor algebra $\TA$ as the space of all finite sequences of tensors defined by
    
    $$ \TA := \bigcup_{M \in \N } \tTA{M}. $$
    
    We clearly have $\TA \subset \eTA$.
    For $\bell = (\bell^n)_{n \in \N}, \bp = (\bp^n)_{n \in \N} \in \eTA$ and $\lambda \in \R$, we define the following operations:
    
    \begin{align*}
        \bell + \bp :&
        = (\bell^n + \bp^n)_{n \in \N}, \quad \bell \otimes \bp :
        = \left( \sum_{k=0}^n \bell^k \otimes \bp^{n-k} \right)_{n \in \N}, \quad 
        \lambda \bell :
        = (\lambda \bell^n)_{n \in \N}.
    \end{align*}
    
    These operations induce analogous operations on $\tTA{M}$ and $\TA$. \\
    
    \textbf{Important notations.} Let $\{ e_1, \dots, e_d \} \subset \R^d$ be the canonical basis of $\mathbb{R}^d$ and $\alphabet = \{ \word{1}, \word{2}, \dots, \word{d} \}$ be the corresponding alphabet. To ease reading, for $i \in \{ 1, \dots, d \}$, we write $e_{i}$ as the blue letter $\word{i}$ and for $n \geq 1, i_1, \dots, i_n \in \{ 1, \dots, d \}$, we write $e_{i_1} \otimes \cdots \otimes e_{i_n} $ as the concatenation of letters $\word{i_1 \cdots i_n}$, that we call a word of length $n$. We note that $(e_{i_1} \otimes \cdots \otimes e_{i_n})_{(i_1, \dots, i_n) \in \{ 1, \dots, d \}^n}$ is a basis of $(\R^d) \conpow{n}$ that can be identified with the set of words of length $n$ defined by 
    
    \begin{equation} \label{eq:sig_basis}
        V_n := \{ \word{i_1 \cdots i_n}: \word{i_k} \in \alphabet \text{ for } k = 1, 2, \dots, n \}.
    \end{equation}
    
    Moreover, we denote by $\emptyword$ the empty word and by $V_0 = \{ \emptyword \}$ which serves as a basis for $(\R^d) \conpow{0} = \R$. It follows that $V := \cup_{n \geq 0} V_n$ represents the standard basis of $\eTA$. In particular, every $\bell \in \eTA$ can be decomposed as
    
    \begin{equation} \label{eq:sig_expansion}
        \bell = \sum_{n=0}^\infty \sum_{\word{v} \in V_n} \bell^{\word{v}} \word{v},
    \end{equation}
    where $\bell^\word{v}$ is the real coefficient of $\bell$ at coordinate $\word{v}$. Representation \eqref{eq:sig_expansion} will be frequently used in the paper. We stress again that in the sequel, every blue `word` $\word{v} \in V$ represents an element of the canonical basis of $\eTA$, i.e. there exists $n \geq 0$ such that $\word{v}$ is of the form $\word{v} = \word{i_1 \cdots i_n}$, which represents the element $e_{i_1} \otimes \cdots \otimes e_{i_n} $. The concatenation $\bell \word{v}$ of elements $\bell \in \eTA$ and the word $\word{v} = \word{i_1 \cdots i_n}$ means $\bell \otimes e_{i_1} \otimes \cdots \otimes e_{i_n}$. \\

    In addition to the decomposition \eqref{eq:sig_expansion} of elements $\bell \in \eTA$, we introduce the projection $\bell \proj{u} \in \eTA$ as
    \begin{equation} \label{eq:projection}
        \bell \proj{u} := \sum_{n=0}^\infty \sum_{\word{v} \in V_n} \bell^\word{vu} \word{v}
    \end{equation}
    for all $\word{u} \in V$. The projection plays an important role in the space of iterated integrals as it is closely linked to partial differentiation, in contrast with the concatenation that relates to integration. It will be used throughout the paper.
    
    \begin{sqremark} \label{rem:proj-decomposition}
        The projection allows us to decompose elements of the extended tensor algebra $\bell \in \eTA$ as 
        \begin{align}
            \bell = \bell^\emptyword \emptyword + \sum_{\word{i} \in \alphabet} \bell \proj{i} \word{i}.
        \end{align}
        This decomposition is quite natural as, when iterated, it gives back the decomposition in \eqref{eq:sig_expansion}.
    \end{sqremark}
    
    \begin{sqexample}
        Take the alphabet $\alphabet[3] = \{ \word{1}, \word{2}, \word{3} \}$ and let $\bell = 4 \cdot \emptyword + 3 \cdot \word{1} - 1 \cdot \word{12} + 2 \cdot \word{2212}$, then
        $$ \bell^{\emptyword} = 4 \cdot \emptyword, \quad \bell \proj{1} = 3 \cdot \emptyword, \quad \bell \proj{2} = - 1 \cdot \word{1} + 2 \cdot \word{221}, \quad \bell \proj{3} = 0. $$
    \end{sqexample}
    
    We now define the bracket between $\bell \in \TA$ and $\bp \in \eTA$ by
    \begin{align} \label{eq:bracket}
        \langle \bell, \bp \rangle 
        := \sum_{n=0}^{\infty} \sum_{\word{v} \in V_n} \bell^\word{v} \bp^\word{v}. 
    \end{align}
    
    Notice that it is actually well defined as $\bell$ has finitely many non-zero terms. For $\bell \in \eTA$, the series in \eqref{eq:bracket} involves infinitely many terms and requires special care, this is discussed in Subsection~\ref{S:infinite}. 
    
    We will also consider another operation on the space of words, the shuffle product. The shuffle product plays a crucial role for an integration by parts formula on the space of iterated integrals, see Proposition~\ref{prop:shufflepropertyextended} below.
    
    \begin{definition}[Shuffle product] \label{def:shuffleprod}
        The shuffle product $\shuprod: V \times V \to \TA$ is defined inductively for all words $\word{v}$ and $\word{w}$ and all letters $\word{i}$ and $\word{j}$ in $\alphabet$ by
        
        \begin{align*}
            (\word{v} \word{i}) \shuprod (\word{w} \word{j}) &
            = (\word{v} \shuprod (\word{w} \word{j})) \word{i} + ((\word{v} \word{i}) \shuprod \word{w}) \word{j},
            \\ \word{w} \shuprod \emptyword &
            = \emptyword \shuprod \word{w} = \word{w}.
        \end{align*}
        
        With some abuse of notation, the shuffle product on $\eTA$ induced by the shuffle product on $V$ will also be denoted by $\shuprod$. The shuffle product is clearly commutative. See \cite{reeshuffles} and \cite{gainesshuffle} for more information on the shuffle product.
    \end{definition}
    
    The shuffle product corresponds to all riffle shuffles of two decks of cards together, which keeps the order of each single deck, e.g. $\word{12} \shuprod \word{34} = \word{1234} + \word{1324} + \word{3124} + \word{1342} + \word{3142} + \word{3412}$.

\subsection{Resolvent and linear equation}

    For $n \in \N$ and $\bell \in \eTA$, we define the concatenation power of $\bell$ by
    
    $$ \bell \conpow{n} := \overbrace{\bell \otimes \bell \otimes \cdots \otimes \bell}^{\text{$n$ times}}, $$
    
    with the convention that $\bell \conpow{0} = \emptyword$. For $\bell \in \eTA$ such that $\bell^\emptyword=0$, we define the \textit{resolvent} of $\bell$ by 
    \begin{equation} \label{nota:coninv}
        \coninv{\emptyword - \bell} := \sum_{n=0}^\infty \bell \conpow{n}.
    \end{equation}
    The assumption $\bell^\emptyword = 0$ assures that it is well-defined. \\

    The resolvent allows us to solve linear algebraic equations, for instance:
    \begin{proposition} \label{prop:resolvent}
        Let $\bp, \bq \in \eTA$ such that $\bq^\emptyword = 0$, then the unique solution $\bell \in \eTA$ to the linear algebraic equation 
        \begin{align} \label{eq:resolvent}
            \bell = \bp + \bell \bq 
        \end{align}
        is given by 
        \begin{align}
            \bell = \bp \coninv{\emptyword - \bq},
        \end{align}
        with $\coninv{\emptyword - \cdot}$ as defined in \eqref{nota:coninv}.
    \end{proposition}
    
    Interestingly, whenever $\bell$ is a linear combination of single letters, the resolvent of $\bell$ is equal to the shuffle exponential $\shuexp{\bell}$ defined by
    \begin{equation} \label{nota:shuexp}
        \shuexp{\bell} := \sum_{n=0}^{\infty} \frac{\bell \shupow{n}}{n!},
    \end{equation}
    where 
    $$ \bell \shupow{n} := \overbrace{\bell \shuprod \bell \shuprod \cdots \shuprod \bell}^{\text{$n$ times}}, \quad n \geq 1, \quad \bell \shupow{0} = \emptyword. $$

\subsection{Signatures}

    We define the (path) signature of a semimartingale process as the sequence of iterated stochastic integrals in the sense of Stratonovich. Throughout the paper, the Itô integral is denoted by $\int_0^\cdot Y_t \d X_t$ and the Stratonovich integral by $\int_0^\cdot Y_t \circ \d X_t.$ If both $X$ and $Y$ are semimartingales then, we have the relation $\int_0^\cdot Y_t \circ \d X_t = \int_0^\cdot Y_t \d X_t + \frac{1}{2} [X,Y]_\cdot $.

    \begin{definition}[Signature] \label{def:sig}
        Fix $T > 0$. Let $(X_t)_{t \geq 0}$ be a continuous semimartingale in $\R^d$ on some filtered probability space $(\Omega, \F, (\F_t)_{t \geq 0}, \P)$. The signature of $X$ is defined by
        \begin{align*}
            \mathbb{X}: \Omega \times [0, T] &
            \to \eTA
            \\ (\omega, t) &
            \mapsto \sigX (\omega) := (1, \sigX^1(\omega), \dots, \sigX^n(\omega), \dots),
        \end{align*}
        where
        $$ \sigX^n := \int_{0 < u_1 < \cdots < u_n < t} \circ \d X_{u_1} \otimes \cdots \otimes \circ \d X_{u_n} $$
        takes value in $(\R^d)^{\otimes n}$, $n \geq 0$. Similarly, the truncated signature of order $M \in \N$ is defined by
        \begin{align} \label{def:sig-trunc}
            \mathbb{X}^{\leq M}: [0, T] &
            \to \tTA{M}
            \\ (\omega, t) &
            \mapsto \sigX^{\leq M}(\omega) := (1, \sigX^1(\omega), \dots, \sigX^M(\omega), 0, \dots, 0, \dots).
        \end{align}
    \end{definition}
    
    The signature plays a similar role to polynomials on path-space. Indeed, in dimension $d=1$, the signature of $X$ is the sequence of monomials $\left( \frac{1}{n!} (X_t - X_0)^n \right)_{n \in \N}$. In particular, any finite combination of elements of the signature $\bracketsigX{\bell}$, defined in \eqref{eq:bracket} for $\bell \in \tTA{M}$, is a polynomial of degree $M$ in $X_t$.
    
    \begin{sqremark} \label{rmk:sig_iteration_def}
        Explicitly we can write the term $\sigX^n$ as $(\sigX^\word{i_1 \cdots i_n})_{(\word{i_1 \cdots i_n}) \in V_n}$. So the definition can be written in iterated form as
        \begin{align} \label{eq:signdef2}
            \sigX^\word{i_1 \cdots i_n} = \int_0^t \sigX[s]^\word{i_1 \cdots i_{n-1}} \circ \d X_s^{\word{i_n}}.
        \end{align} 
    \end{sqremark}

    In what will follow, we are exclusively interested in the case $d=2$ and $X_t = \widehat{W}_t := (t, W_t)$ where $W$ is a 1-dimensional Brownian motion.
    Its first few signature orders are given by
    \begin{equation}
        \sig^0 = 1,
        \quad
        \sig^1 =
        \begin{pmatrix}
            t \\
            W_t
        \end{pmatrix},
        \quad
        \sig^2 =
        \begin{pmatrix}
            \frac{t^2}{2!} & \int_0^t s \d W_s \\
            \int_0^t W_s \d s & \frac{W_t^2}{2!}
        \end{pmatrix}.
    \end{equation}

\subsection{Infinite linear combinations of signature elements} \label{S:infinite}

    In this section, we recall some results on infinite linear combinations $\bracketsig{\bell}$ for certain admissible $\bell \in \eTA$ for which the infinite series will make sense. Two crucial ingredients for our paper are the shuffle product (Proposition~\ref{prop:shufflepropertyextended}) and an Itô's formula (Lemma~\ref{lem:sig-ito}). We follow the presentation in \cite[Section 2]{linearfbm}, and we refer to \cite{cuchiero2023polynomial} for more general results. \\

    We first introduce the space $\A$ of admissible elements $\bell$ below, using the associated semi-norm:
    
    $$ \norm{\bell}_t^\A := \sum_{n=0}^\infty \left| \sum_{\word{v} \in V_n} \bell^\word{v} \sig^\word{v} \right|, \quad t \geq 0, $$
    
    recall the definition of $V_n$ in \eqref{eq:sig_basis} and the decomposition \eqref{eq:sig_expansion}. Whenever, $\norm{\bell}_t^\A < \infty$ a.s., the infinite linear combination 
    
    $$ \bracketsig{\bell} = \sum_{n=0}^\infty \sum_{\word{v} \in V_n} \bell^{\word{v}} \sig^{\word{v}} $$
    
    is well-defined. This leads to the following definition for the admissible set $\A$:
    
    $$ \A := \left\{ \bell \in \eTA[2] : \norm{\bell}_t^\A < \infty \text{ for all } t \in [0, T], \text{ a.s.} \right\}.\label{eq:defA} $$ 
    We stress that the null set in $\mathcal{A}$ does not depend on $t$. Note that $\TA[2] \subset \A$ and that $\bracketsig{\bell}$ is an extension of \eqref{eq:bracket}, as the two bracket operations $\langle \cdot, \cdot \rangle$ coincide whenever $\bell \in \TA[2]$. \\
    
    The admissible set $\A$ has another very interesting property, as it allows us to linearize polynomials on infinite linear combination of the signature, see Proposition \ref{prop:shufflepropertyextended}. This is what most of the literature refers to when putting forth the linearization power of the signature.

    \begin{proposition}[Shuffle property] \label{prop:shufflepropertyextended}
        If $\bell_1, \bell_2 \in \A$, then $\bell_1 \shuprod \bell_2 \in \A$ and
        $$ \bracketsig{\bell_1} \bracketsig{\bell_2} = \bracketsig{\bell_1 \shuprod \bell_2}. $$
    \end{proposition}
    
    \begin{proof}
       This follows from a particular instance of shuffle compatible partitions, see \cite[Lemma 4.4]{cuchiero2023polynomial}. Alternatively, we refer to \cite[Proposition 3.1]{linearfbm}.
    \end{proof}


    For elements $\bell \in \A$, the process $(\langle \bell, \sig \rangle)_{t \leq T}$ is well-defined. An important question is to know whether it is a semimartingale and compute its Itô's decomposition. The answer is positive, thanks to an Itô's formula in Lemma~\ref{lem:sig-ito} below, for elements in the set 
    \begin{equation} \label{eq:I}
        \I := \left\{ \bell \in \A: \text{ for all } t \in [0, T], \: \norm{\bell}_t^\I < \infty \text{ and } \int_0^T \norm{\bell}_t^\I \d t < \infty \text{ a.s.} \right \},
    \end{equation}
    where
    $$ \norm{\bell}_t^{\I} := \norm{\bell \proj{1}}_t^\A + \norm{\tfrac{1}{2} \bell \proj{22}}_t^\A + \left( \norm{\bell \proj{2}}_t^\A \right)^2. $$
    More generally, we state the result for time dependent linear combinations $(\langle \bell_t, \sig \rangle)_{t \leq T}$ with $\bell:[0,T] \to \A$ in the set
    
    \begin{equation} \label{eq:Iprime}
        \I' := \left\{ \bell: [0, T] \to \I:
        \begin{matrix}
            \text{ for all } t \in [0, T], \: \bell_t^\word{v} \in C^1([0, T]) \text{ for all } \word{v} \in V,
            \\ \text{ and } \norm{\dot{\bell}_t}_t^\A < \infty \text{ and } \int_0^T \norm{\dot{\bell}_t}_t^\A \d t < \infty \text{ a.s.} \hfill
        \end{matrix}
        \right \},
    \end{equation}
    where $\dot{\bell_t} := \sum_{\word{v} \in V} \frac{\d}{\d t} \bell_t^\word{v} \word{v}$ for all $t \in [0, T]$.
    
    \begin{lemma}[Itô's decomposition] \label{lem:sig-ito}
        Let $\bell \in \I$, then
        \begin{align} \label{eq:ItoW}
            \d \bracketsig{\bell} = \bracketsig{\bell \proj{1} + \tfrac{1}{2} \bell \proj{22}} \d t + \bracketsig{\bell \proj{2}} \d W_t. 
        \end{align} 
        Let $\bell \in \I'$, then
        \begin{align} \label{eq:ItohatW}
            \d \bracketsig{\bell_t} = \bracketsig{\dot{\bell}_t + \bell_t \proj{1} + \tfrac{1}{2} \bell_t \proj{22}} \d t + \bracketsig{\bell_t \proj{2}} \d W_t. 
        \end{align} 
    \end{lemma}

\section{The signature volatility model} \label{sec:model}

    Let $(\Omega, \F, \Q)$ be a probability space supporting a two-dimensional Brownian motion $(W, W^{\perp})$. We denote by $(\F_t)_{t \geq 0}$ the filtration generated by $(W, W^{\perp})$. We set 
    \begin{align} \label{eq:BM}
        B = \rho W + \sqrt{1 - \rho^2} W^{\perp},
    \end{align}
    for some $\rho \in [-1, 1]$. We define the time-augmented process $\widehat{W}_t = (t, W_t)$ and we consider that the dynamics of the risky asset $S$, under the risk neutral probability measure $\Q$, are given by a stochastic volatility model where the volatility process $\Sigma$ is a (possibly infinite) linear combination of the signature of $\widehat{W}$: 
    \begin{align}
        \frac{\d S_t}{S_t} &= \Sigma_t \d B_t, \label{eq:sigmodel1} \\
        \Sigma_t &= \bracketsig{\bsigma_t}, \label{eq:sigmodel2}
    \end{align}
    where $\bsigma: [0, T] \to \A$ corresponds to the parameters of the volatility process and is such that 
    \begin{align} \label{eq:Sigma2}
        \int_0^T \E \left[ \Sigma_t^2 \right] \d t < \infty.
    \end{align}
    We recall the definition of the set $\A$ in \eqref{eq:defA}. 
    The condition \eqref{eq:Sigma2} ensures that the stochastic integral
    $$ \int_0^\cdot \Sigma_s \d B_s $$
    is well defined as an Itô integral, so that there exists a unique solution to \eqref{eq:sigmodel1} given by
    $$ S_t = S_0 \exp \left( -\frac{1}{2} \int_0^t \Sigma_s^2 \d s + \int_0^t \Sigma_s \d B_s \right), \quad t \geq 0. $$
    
    Condition \eqref{eq:Sigma2} can be made more explicit by observing that the instantaneous variance $\Sigma^2_t$ is also linear in the signature, i.e. 
    \begin{align} \label{eq:shufflesigma2}
        \Sigma_t^2 = \bracketsig{\bsigma_t \shupow{2}},
    \end{align}
    thanks to the shuffle product in Proposition \ref{prop:shufflepropertyextended}. We note that the time-independent case $\bsigma_t = \bsigma$ for all $t \in [0, T]$ and some $\bsigma \in \A$, leads to $\int_0^T \E \left[ \Sigma_t^2 \right] \d t = \langle \bsigma \shupow{2} \word{1}, \sigE[T] \rangle$ where the quantity $\sigE[T]$ can be computed explicitly using Fawcett's formula \cite{fawcett} extended to time-augmented Brownian motions in \cite[Proposition 4.10]{lyonsvictoir}:
    \begin{align} \label{eq:fawcet}
        \sigE = \sum_{n \geq 0} \frac{t^n}{n!} \left( \word{1} + \frac{1}{2} \word{22} \right) \conpow{n}.
    \end{align}
    
    In practice, we will usually be interested in truncated elements $\bsigma: [0, T] \to \tTA[2]{M}$ for some $M \in \N$, which automatically satisfy the condition \eqref{eq:Sigma2} since in this case $\bsigma$ has only a finite number of non-zero terms and all the integrated moments of the Brownian motion are finite.  \\

    Notice that for general $\bsigma: [0, T] \to \A$ the process $\Sigma$ is not necessarily Markovian nor a semimartingale. It is a semimartingale if in addition $\bsigma \in \I'$ thanks to Itô's formula in Lemma~\ref{lem:sig-ito}. For truncated elements $\bsigma: [0, T] \to \tTA[2]{M}$ the process $\Sigma$ is a semimartingale. So far, truncated elements have been considered in the related literature \cite{arribas2020sig,cuchiero2023spvix}. \\
    
    In the next subsection, we highlight the flexibility introduced by infinite linear combinations of signature elements in terms of exact representations and provide numerical implementations of their truncated form, i.e.
    $$ \Sigma_t^{\leq M} := \bracketsig{\bsigma^{\leq M}}, \quad M \geq 0, $$
    where $\bsigma^{\leq M}: [0, T] \to \tTA[2]{M}$ is the truncated form of $\bsigma$ at order $M$, i.e.~its $M$ first levels coincide with $\bsigma$ and everything else is set to $0$. Note that $\bracketsig{\bsigma^{\leq M}} = \bracketsigtrunc[M]{\bsigma}$ where $\sig^{\leq M}$ is defined in \eqref{def:sig-trunc}.

\subsection{Examples of exact representations} \label{subsec:representations}

    We first highlight the flexibility of the signature volatility model \eqref{eq:sigmodel1}-\eqref{eq:sigmodel2} by showing that it subsumes several known and useful Markovian and non-Markovian models based on Ornstein-Uhlenbeck processes, mean-reverting geometric Brownian motions, square-root processes\footnote{The theoretical justification for the representation of the square-root process is still open, it is validated numerically in Section~\ref{S:CIR} below.} and processes with path-dependent dynamics including stochastic Volterra processes. Our approach is closely related to the stochastic Taylor expansion, see for instance \cite[Chapter 5]{taylorsto_kloeden} and \cite[Theorem 6.3.6]{taylorsto_foster_cir} for expansions on Ornstein-Uhlenbeck and Cox-Ingersoll-Ross processes respectively, with two main advantages: first, it is not limited to Markovian processes, and second, it gives an algebraic formulation that is easy to derive and work with.

\subsubsection{Models based on the Ornstein-Uhlenbeck process}
    
    The Ornstein-Uhlenbeck (OU) process $X$ given by
    \begin{align} \label{sde:OU}
        \d X_t = \kappa (\theta - X_t) \d t + \eta \d W_t, \quad X_0 = x \in \R,
    \end{align}
    for $\kappa, \theta, \eta \in \R$, can be represented as an infinite linear combination of the signature elements of the time-extended Brownian motion, either in a time-independent or time-dependent way as shown in the next Lemma.
    
    \begin{proposition} \label{prop:rep-OU}
        The unique solution $X$ to \eqref{sde:OU} is given by
        \begin{align} \label{eq:linear-OU}
            X_t = \bracketsig{\bell^\textnormal{OU}},
            \quad
            \bell^\textnormal{OU} = (x \emptyword + \kappa \theta \word{1} + \eta \word{2}) \shuexp{-\kappa \word{1}},
        \end{align}
        such that $\bell^\textnormal{OU} \in \A$, with $\shuexp{~}$ the exponential shuffle defined in \eqref{nota:shuexp}.
        Furthermore, $X$ can also be written in terms of time-dependent coefficients:
        \begin{align} \label{eq:linear-OUtime}
           X_t = \bracketsig{\tilde{\bell}_t^\textnormal{OU}},
           \quad
           \tilde{\bell}_t^\textnormal{OU} = \theta \emptyword + e^{-\kappa t} \left( (x - \theta) \emptyword + \eta \shuexp{\kappa \word{1}} \word{2} \right),
        \end{align}
    \end{proposition}

    \begin{proof}
        The proof is detailed in Appendix \ref{apn:OU}. 
    \end{proof}
    
    \begin{sqexample} \label{ex:elements_OU}
        To be more explicit, up to order 2, the linear form of the Ornstein-Uhlenbeck process in \eqref{eq:linear-OU} reads
        $$ \bell^\textnormal{OU} = \left( x,
        \begin{pmatrix}
            -\kappa (x - \theta) \\
            \eta
        \end{pmatrix},
        \begin{pmatrix}
            \kappa^2 (x - \theta) & 0 \\
            -\kappa \eta & 0
        \end{pmatrix},
        \cdots \right). $$
    \end{sqexample}

    Although the representations \eqref{eq:linear-OU} and \eqref{eq:linear-OUtime} are equivalent, from a numerical perspective \eqref{eq:linear-OUtime} is more advantageous. In the time-independent representation \eqref{eq:linear-OU}, $e^{-\kappa t}$ is approximated by the first terms of its Taylor expansion $\sum_{n \geq 0} \frac{(-\kappa t)^n}{n!}$, which becomes numerically unstable for lower orders of truncation in the region $t > 1 / |\kappa|$, see Figure \ref{fig:ex-traj-OU}. However, in the time-dependent representation \eqref{eq:linear-OUtime}, $e^{-\kappa t}$ is exact and even though the approximate solution deteriorates with time, it stays stable and converges towards $\theta$, see Figure \ref{fig:ex-traj-OUtime}.
    
    \begin{figure}[H]
        \centering
        \subfloat[\centering $\kappa=1, \theta=0.25, \eta=1.2$]{{\includegraphics[width=\twoplotswidth]{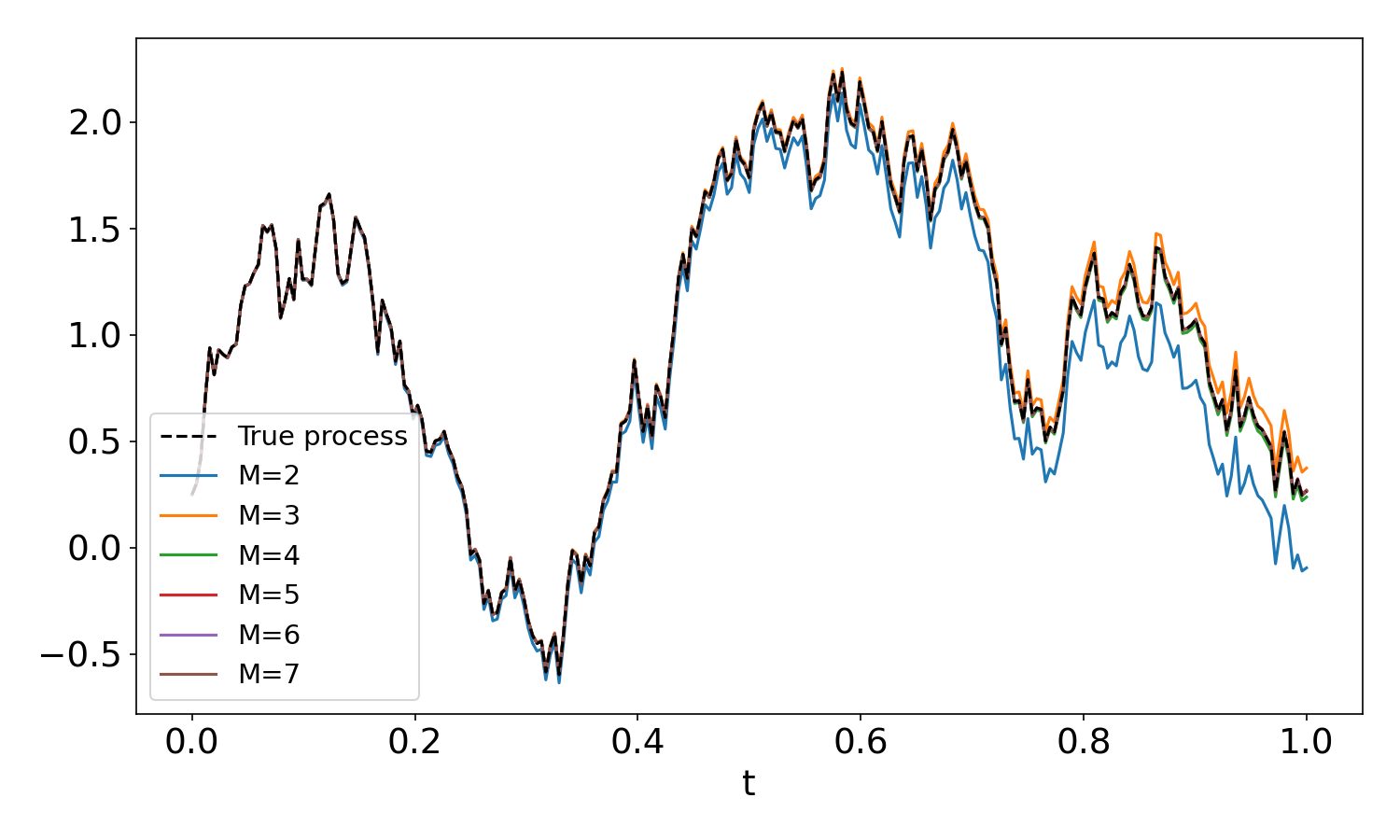}}} 
        \quad
        \subfloat[\centering $\kappa=4, \theta=0.25, \eta=2$]{{\includegraphics[width=\twoplotswidth]{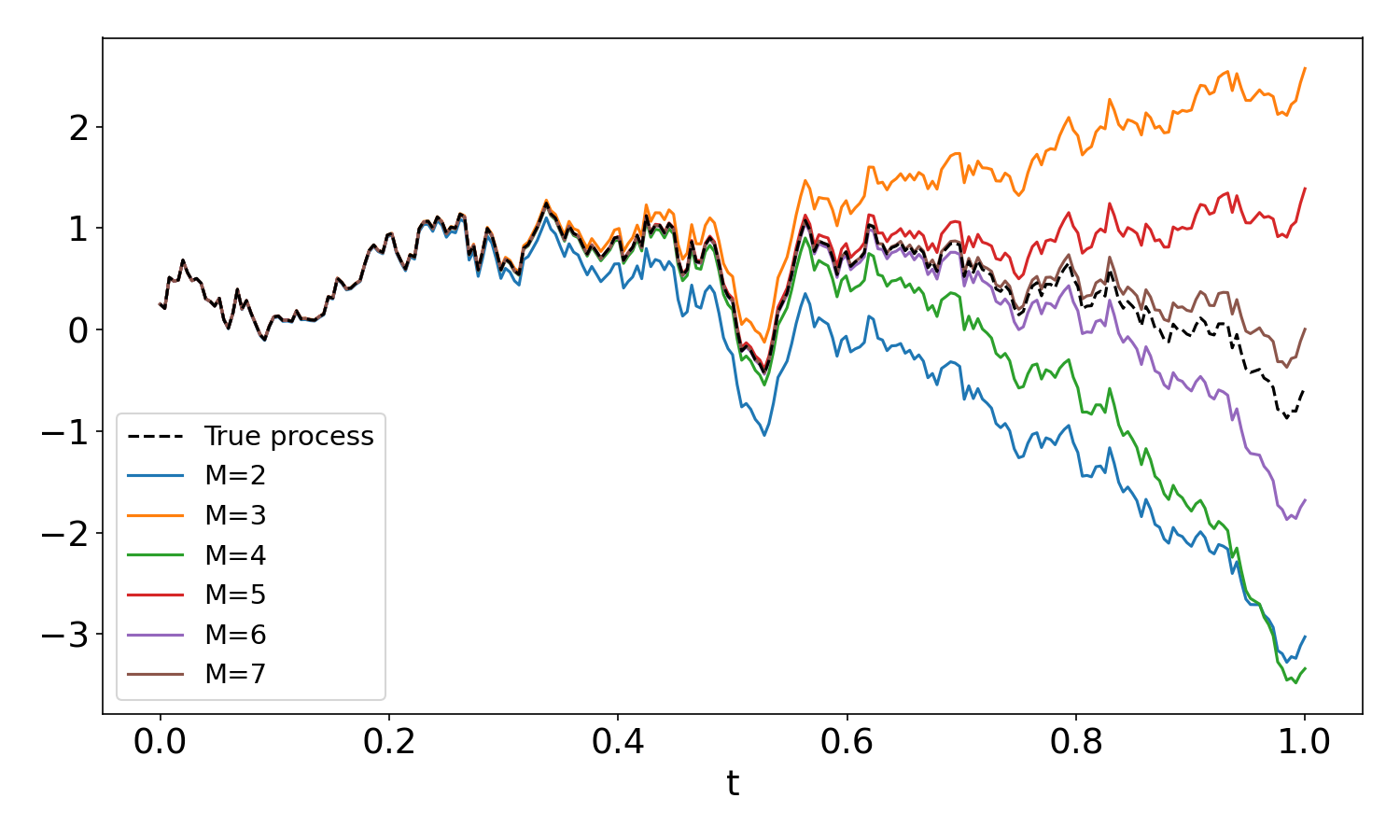}}}  
        \caption{Trajectories of an Ornstein-Uhlenbeck process against their truncated time-independent linear representation \eqref{eq:linear-OU}, i.e. $\bracketsigtrunc[M]{\bell^\textnormal{OU}}$, for several truncation orders $M$.}
        \label{fig:ex-traj-OU}
    \end{figure}

    \begin{figure}[H]
        \centering
        \subfloat[\centering $\kappa=1, \theta=0.25, \eta=1.2$]{{\includegraphics[width=\twoplotswidth]{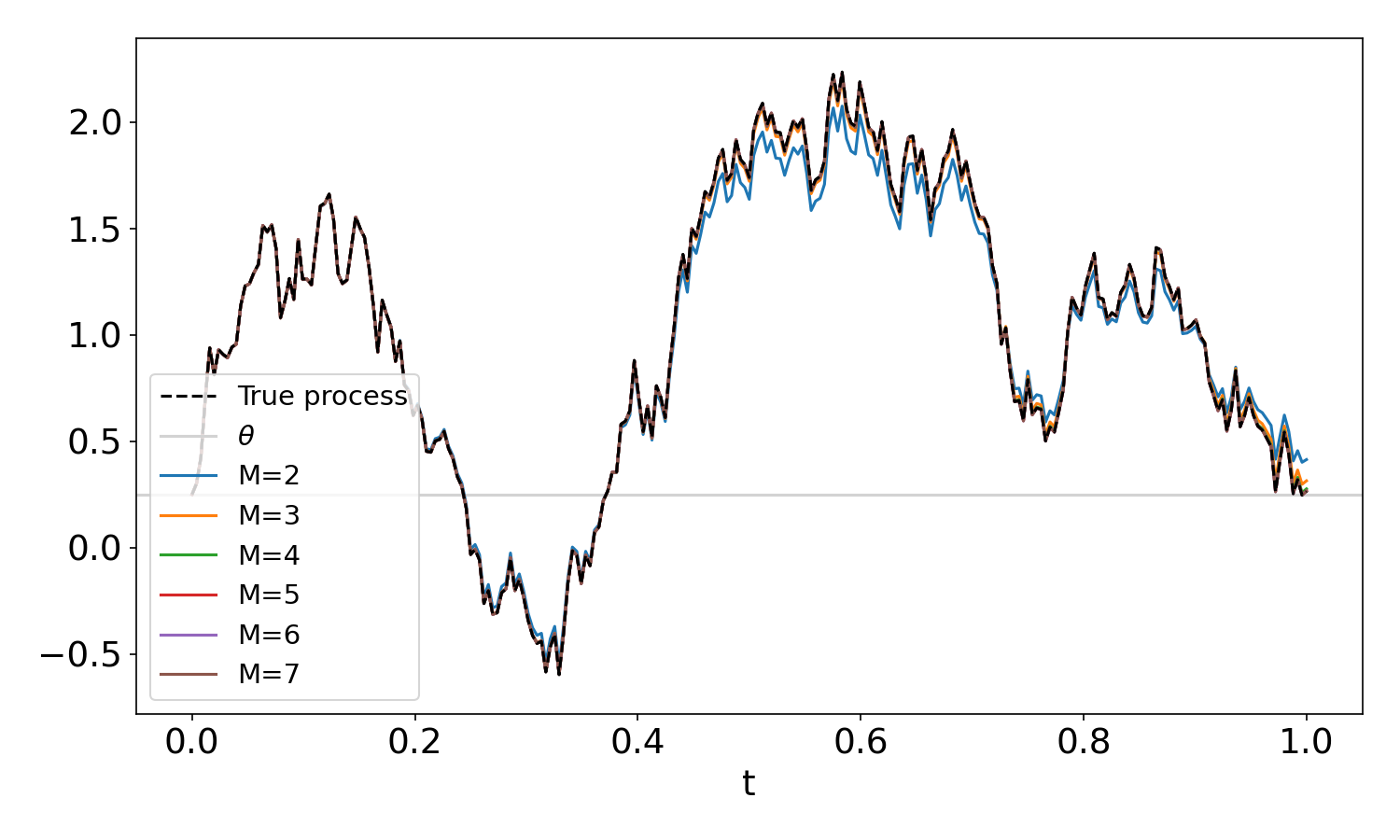}}} 
        \quad
        \subfloat[\centering $\kappa=4, \theta=0.25, \eta=2$]{{\includegraphics[width=\twoplotswidth]{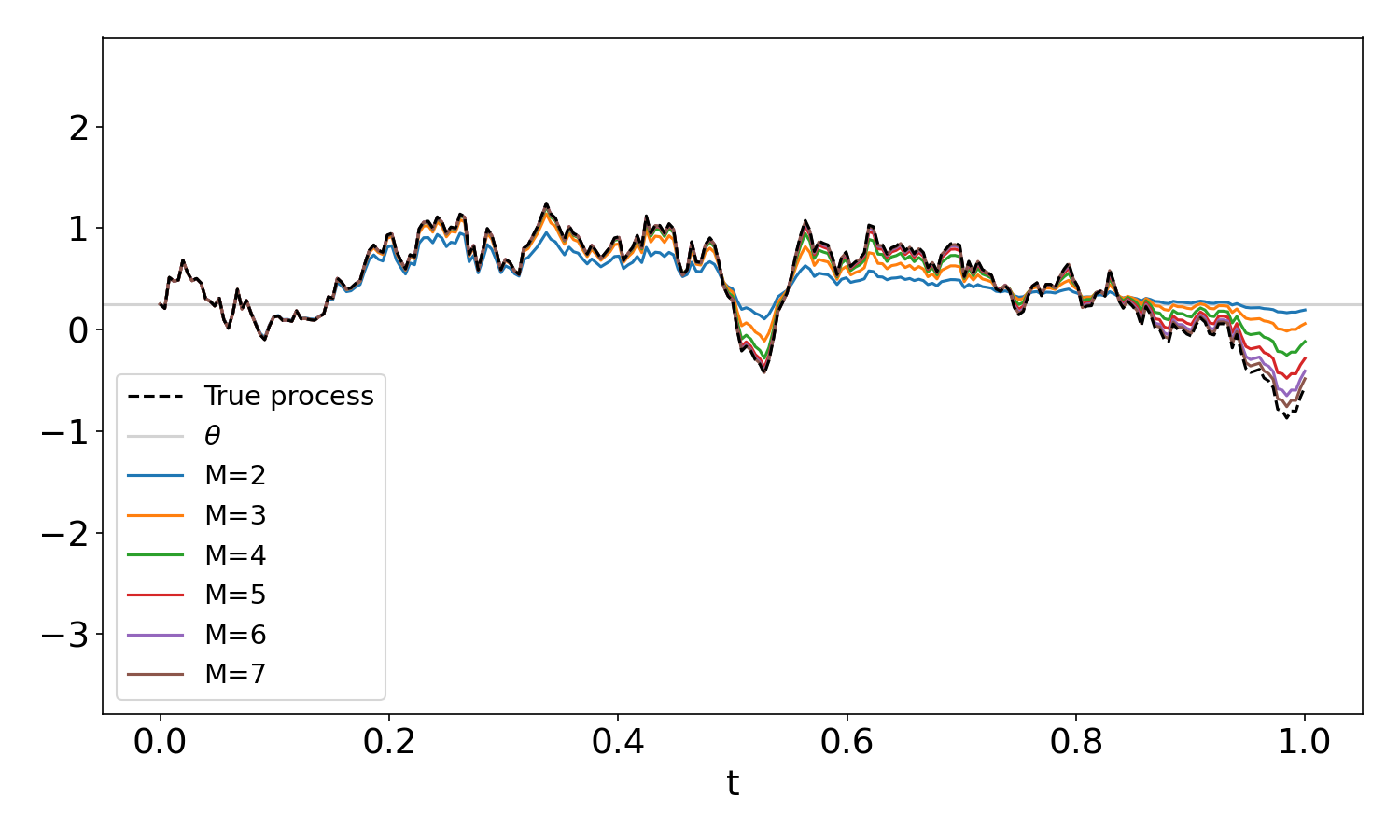}}} 
        \caption{Trajectories of an Ornstein-Uhlenbeck process against their truncated time-dependent linear representation \eqref{eq:linear-OUtime}, i.e. $\bracketsigtrunc[M]{\tilde{\bell}_t^\textnormal{OU}}$, for several truncation orders $M$.}
        \label{fig:ex-traj-OUtime}
    \end{figure}

    If $|\kappa t| < 1$, we can see in the left-hand side of Figure \ref{fig:ex-traj-OU} and Figure \ref{fig:ex-traj-OUtime} that the truncated linear representations seems to converge quite quickly to the explicit solution of the Ornstein-Uhlenbeck and a truncation order $M=4$ is sufficient to get a relatively close fit. \\

    Going back to our signature volatility model, the representations of the Ornstein-Uhlenbeck process in Proposition~\ref{prop:rep-OU} combined with the shuffle property of Proposition~\ref{prop:shufflepropertyextended} show that the model \eqref{eq:sigmodel1}-\eqref{eq:sigmodel2} nests any stochastic volatility model based on an Ornstein-Uhlenbeck process of the form
    
    $$ \frac{\d S_t}{S_t} = f(t, X_t) \d B_t, \quad f(t, x) = \sum_{k \geq 0} \alpha_k(t) x^k, $$
    
    for some coefficients $\alpha_k: [0, T] \to \R$ for which the series has an infinite radius of absolute convergence for all $t \in [0, T]$. More precisely, using the representation \eqref{eq:linear-OU} or \eqref{eq:linear-OUtime} and the shuffle property in Lemma~\ref{prop:shufflepropertyextended}, we can write 
    \begin{align} \label{eq:fanalytic}
        f(t, X_t) = \sum_{k \geq 0} \alpha_k(t) \bracketsig{\bell^\textnormal{OU}}^k = \bracketsig{\bsigma_t} \quad \text{with } \bsigma_t := \sum_{k \geq 0} \alpha_k(t) \left( \bell^\textnormal{OU} \right) \shupow{k}.
    \end{align}
    It is easily shown that $\bsigma_t \in \A$ for all $t \in [0, T]$ since $\norm{\bsigma_t}_s^\A \leq \sum_k |\alpha_k(t)| \cdot \norm{(\bell^\textnormal{OU}) \shupow{k}}_s^\A$ for all $s \in [0, T]$ and $\norm{\bell \shupow{k}} \leq \norm{\bell}^k$ for all $k$, see \cite[Appendix C]{linearfbm}, implying $(\alpha_k(t))_k$ with infinite radius of absolute convergence for all $t \in [0, T]$ is enough to imply $\bsigma_t \in \A$.
    
    This clearly includes:
    \begin{itemize}
        \item
        The Stein-Stein model \cite{stein-stein} for $f(t, x) = x$,
        
        \item 
        The Bergomi model \cite{dupire1993,BergomiSmileII} for 
        $$ f(t, x) = \xi_0(t) e^{\eta x} = \xi_0(t) \sum_{k \geq 0} \frac{(\eta x)^k}{k!}, $$
        for some $\eta \in \R$ and some deterministic input curve $\xi_0$,

        \item
        The Quintic OU model \cite{quintic} for 
        $$ f(t, x) = \xi_0(t) (\alpha_0 + \alpha_1 x + \alpha_3 x^3 + \alpha_5 x^5), $$
        for some $\alpha_i \geq 0$ and some deterministic input curve $\xi_0$, and any other finite polynomial of the Ornstein-Uhlenbeck process.
    \end{itemize}

\subsubsection{Models based on the mean-reverting geometric Brownian motion}

    More generally, the mean-reverting geometric Brownian motion (mGBM) $Y$, given by 
    \begin{align} \label{sde:mGBM}
        \d Y_t = \kappa (\theta - Y_t) \d t + (\eta + \alpha Y_t) \d W_t, \quad Y_0 = y \in \R,
    \end{align}
    for $\kappa, \theta, \eta, \alpha \in \R$ can be represented as an infinite  linear combination of the signature of the time-extended Brownian motion, either in a time-independent or time-dependent way.
    
    \begin{proposition} \label{prop:rep-mGBM}
        The unique solution $Y$ to \eqref{sde:mGBM} is given by
        \begin{align} \label{eq:linear-mGBM}
            Y_t = \bracketsig{\bell^\textnormal{mGBM}},
            \quad
            \bell^\textnormal{mGBM} = \left( y \emptyword + \left( \kappa \theta - \frac{\alpha \eta}{2} \right) \word{1} + \eta \word{2} \right) \shuexp{\left(-\left( \kappa + \frac{\alpha^2}{2} \right) \word{1} + \alpha \word{2} \right)},
        \end{align}
        such that $\bell^\textnormal{mGBM} \in \A$, with $\shuexp{~}$ as defined in \eqref{nota:shuexp}.
        Equivalently, $Y$ can also be written in terms of time-dependent coefficients:
        \begin{align} \label{eq:linear-mGBMtime}
           Y_t = \bracketsig{\tilde{\bell}_t^\textnormal{mGBM}},
           \quad
           \tilde{\bell}_t^\textnormal{mGBM} = \theta \emptyword + e^{-\lambda t} \left( (\bell^\textnormal{mGBM} - \theta \emptyword) \shuprod \shuexp{\lambda \word{1}} \right),
        \end{align}
        for some $\lambda \in \R$.
    \end{proposition}

    \begin{proof}
        The proof is detailed in Appendix \ref{apn:mGBM}. 
    \end{proof}

    \begin{sqexample}
        Up to order 2, the linear form of a mean-reverting geometric Brownian motion reads
        $$ \bell^\textnormal{mGBM} = \left( y,
        \begin{pmatrix}
            \beta \\
            \gamma
        \end{pmatrix},
        \begin{pmatrix}
            \beta \mu & \beta \alpha \\
            \gamma \mu & \gamma \alpha
        \end{pmatrix},
        \begin{pmatrix}
            \beta \mu^2 & \beta \mu \alpha & \\
            \gamma \mu^2 & \gamma \mu \alpha & \\
            & \beta \mu \alpha & \beta \alpha^2 \\
            & \gamma \mu \alpha & \gamma \alpha^2
        \end{pmatrix},
        \cdots \right), $$
        where $\mu = -\left( \kappa + \frac{\alpha^2}{2} \right), \beta = \mu y + \left( \kappa \theta - \frac{\alpha \eta}{2} \right)$ and $\gamma = \alpha y + \eta$.
    \end{sqexample}
    
   The behavior and numerical limitations of the linear representations of the mean-reverting geometric Brownian motion are similar to those of the Ornstein-Uhlenbeck, as shown in Figure \ref{fig:ex-traj-mGBM} and Figure \ref{fig:ex-traj-mGBMtime}. \\

    \begin{figure}[H]
        \centering
        \subfloat[\centering $\kappa=1, \theta=0.25, \eta=0.5, \alpha=1.6$]{{\includegraphics[width=\twoplotswidth]{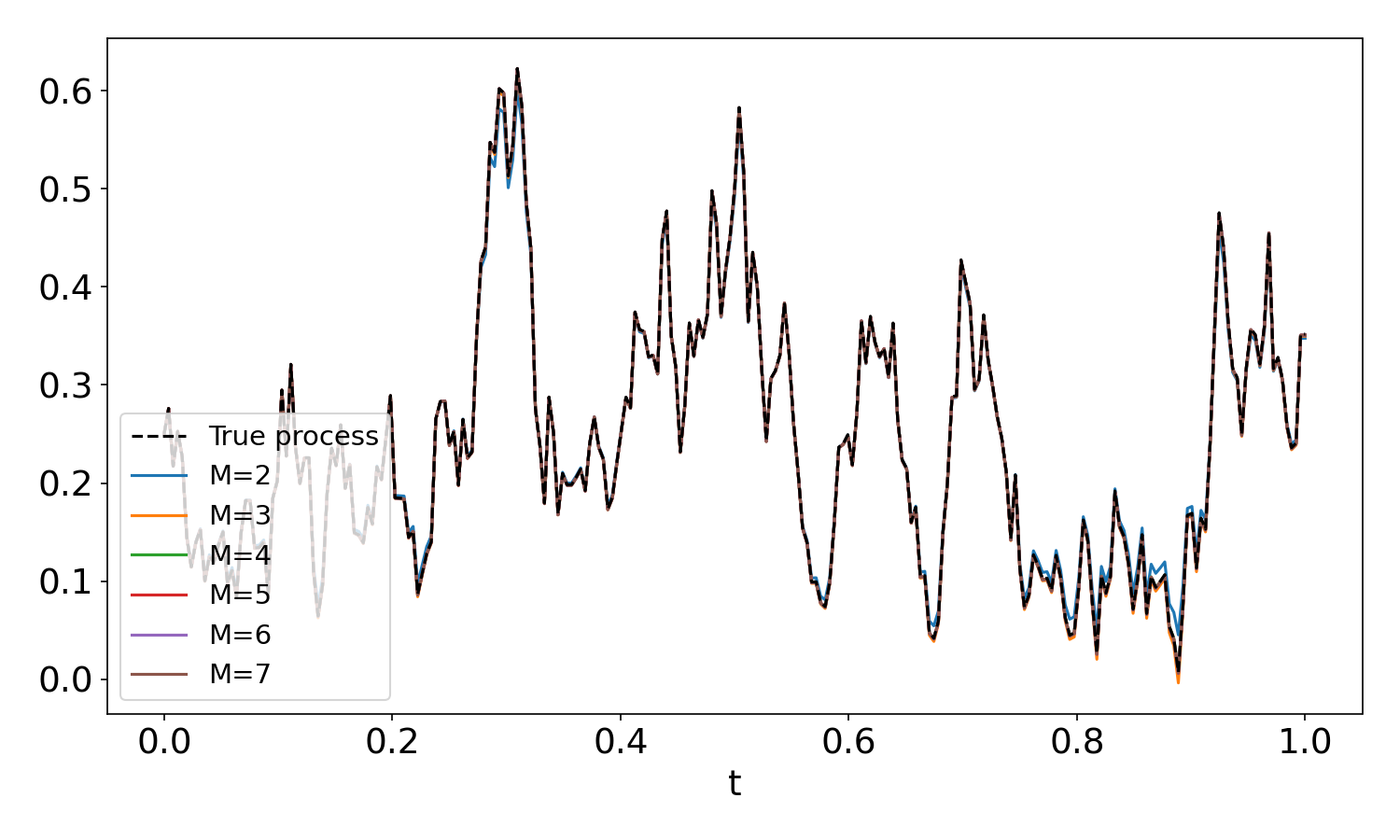}}} 
        \quad
        \subfloat[\centering $\kappa=4, \theta=0.25, \eta=0.5, \alpha=2$]{{\includegraphics[width=\twoplotswidth]{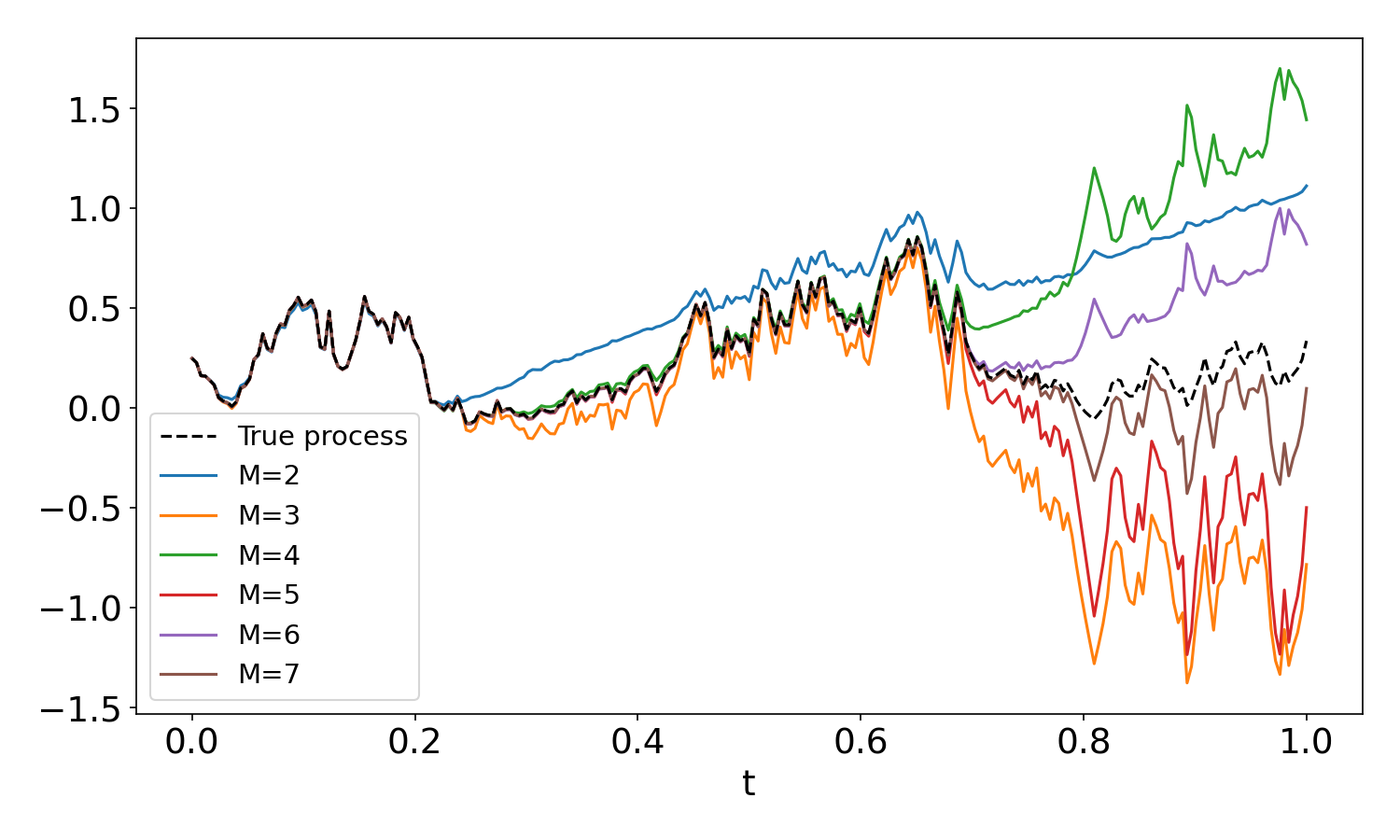}}} 
        \caption{Trajectories of a mean-reverting geometric Brownian motion against their truncated time-independent linear representation \eqref{eq:linear-mGBM}, i.e. $\bracketsigtrunc[M]{\bell^\textnormal{mGBM}}$, for several truncation orders $M$.}
        \label{fig:ex-traj-mGBM}
    \end{figure}

    \begin{figure}[H]
        \centering
        \subfloat[\centering $\kappa=1, \theta=0.25, \eta=0.5, \alpha=1.6, \lambda = \kappa + \frac{\alpha^2}{2}$]{{\includegraphics[width=\twoplotswidth]{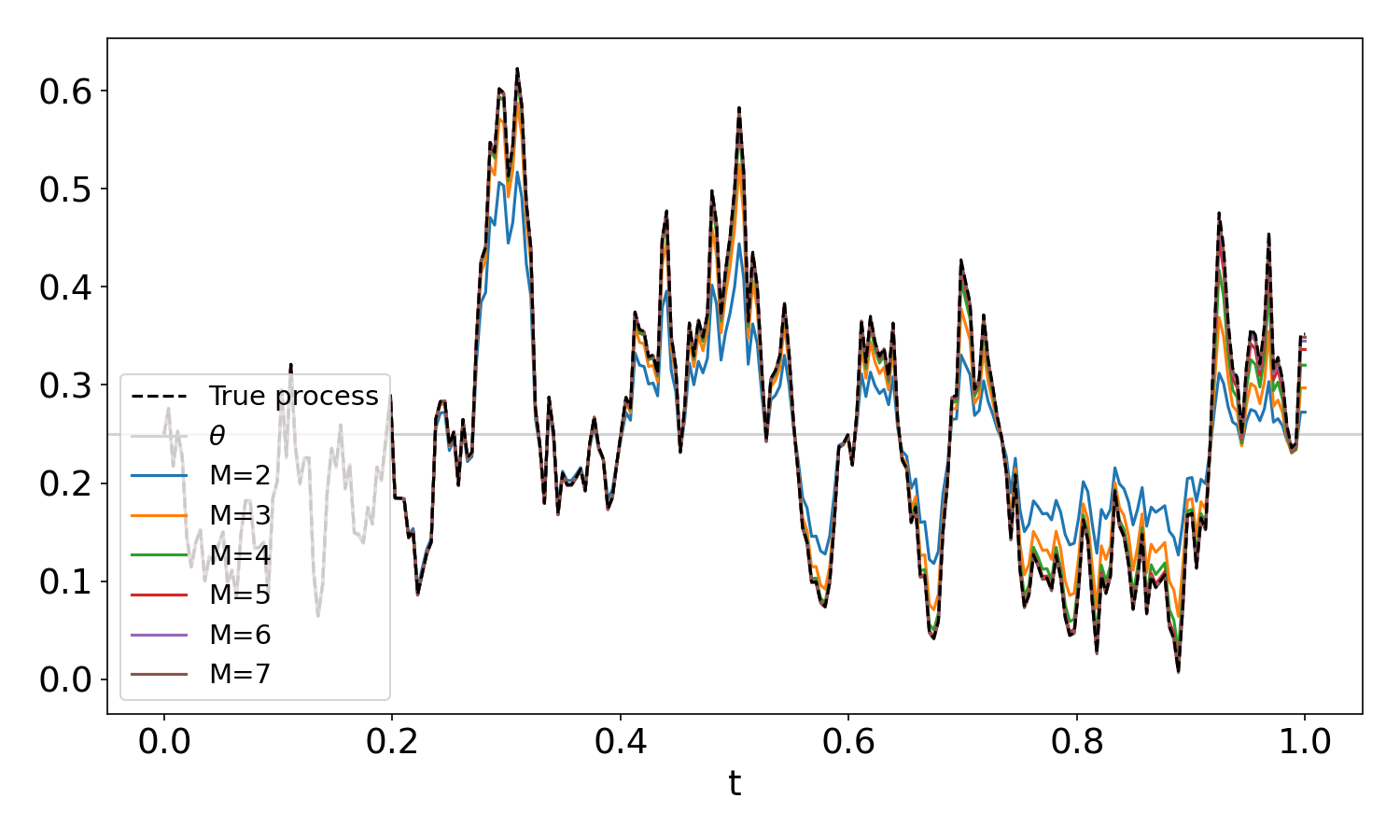}}} 
        \quad
        \subfloat[\centering $\kappa=4, \theta=0.25, \eta=0.5, \alpha=2, \lambda = \kappa + \frac{\alpha^2}{2}$]{{\includegraphics[width=\twoplotswidth]{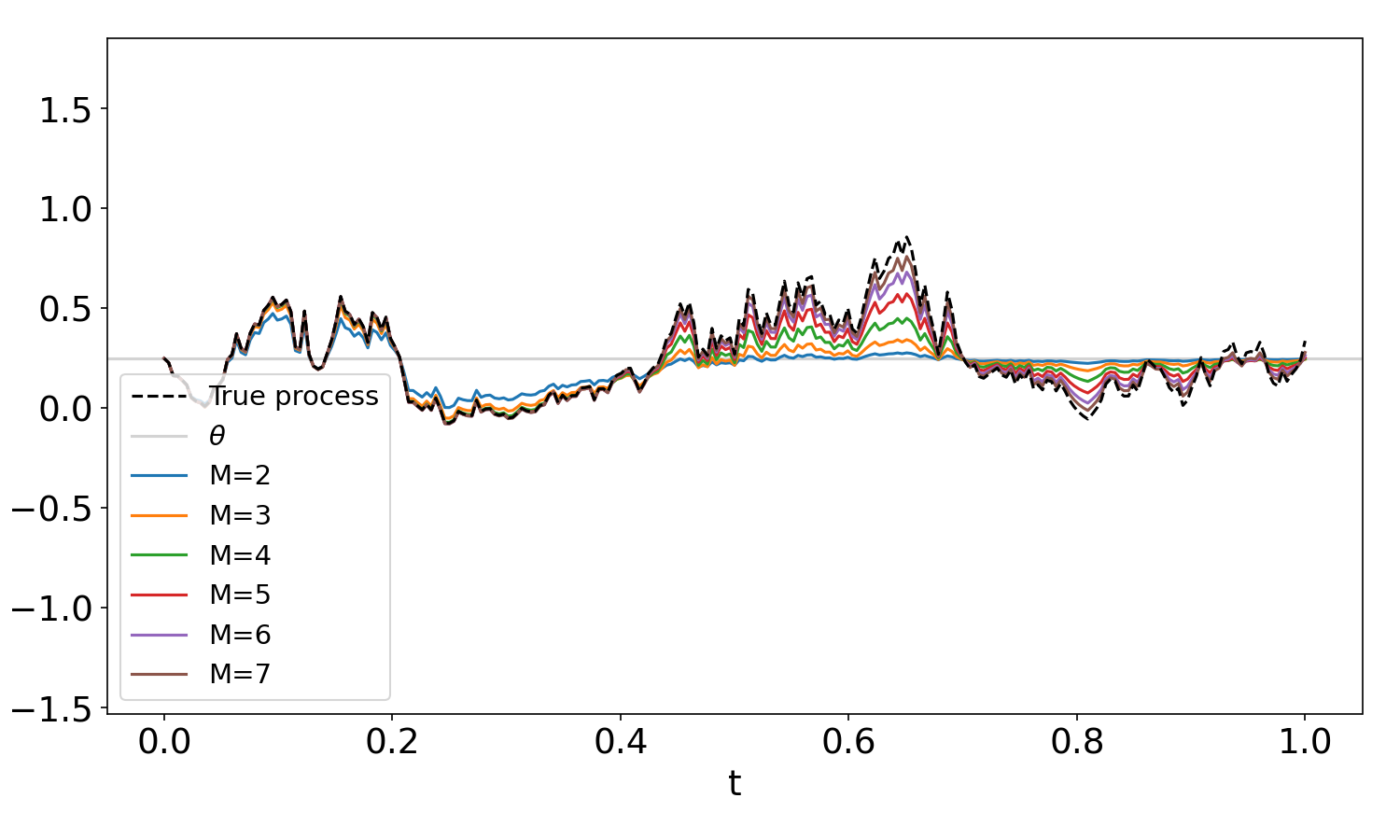}}} 
        \caption{Trajectories of a mean-reverting geometric Brownian motion against their truncated time-dependent linear representation \eqref{eq:linear-mGBMtime}, i.e. $\bracketsigtrunc[M]{\tilde{\bell}_t^\textnormal{mGBM}}$, for several truncation orders $M$.}
        \label{fig:ex-traj-mGBMtime}
    \end{figure}

    \begin{sqremark}
        For Monte Carlo simulations, it is better to have explicit processes. When $\alpha \neq 0$, the mGBM solution to \eqref{sde:mGBM} can also be formulated explicitly with
        \begin{align} \label{eq:mGBM-montecarlo}
            Y_t := \left( y + \frac{\eta}{\alpha} + \kappa \left( \theta + \frac{\eta}{\alpha} \right) \int_0^t e^{\left( \kappa + \frac{\alpha^2}{2} \right) s - \alpha W_s} \d s \right) e^{-\left( \kappa + \frac{\alpha^2}{2} \right) t + \alpha W_t} - \frac{\eta}{\alpha}.
        \end{align}
        When $\alpha = 0$, the mGBM is an Ornstein-Uhlenbeck process.
    \end{sqremark}
    
    Going back to the signature volatility model \eqref{eq:sigmodel2}. The mGBM representation encompasses the following volatility processes 
    \begin{itemize}
        \item The volatility in the Hull-White model \cite{hull-white}, i.e.
        \begin{align} \label{sde:hullwhite}
            \frac{\d \sigma_t}{\sigma_t} = \left( \mu - \tfrac{1}{2} \xi^2 \right) \d t + \xi \d W_t, \quad \sigma_0 \in \R.
        \end{align}
        
        \item The volatility in \citet{dupire1993}, i.e.
        \begin{align} \label{sde:dupireforwardvariance}
            \frac{\d \sigma_t}{\sigma_t} = \frac{1}{2} \left(\frac{\partial \log V_t(0)}{\partial t} - \frac{b^2}{4} \right) \d t + \frac{b}{2} \d W_t, \quad \sigma_0 \in \R^+,
        \end{align}
        for some deterministic forward variance curve $V_t(0)$. Note that the linear representation of $\sigma$ would be time-dependent.
    \end{itemize}

\subsubsection{Models based on the square-root process} \label{S:CIR}

    A square-root or Cox-Ingersoll-Ross \cite{cox-ingersoll-ross} (CIR) process $V$, driven by
    \begin{align} \label{sde:CIR}
        \d V_t = \kappa (\theta - V_t) \d t + \eta \sqrt{V_t} \d W_t, \quad V_0 = v > 0,
    \end{align}
    seems to admit a conjectured linear representation
    \begin{align} \label{eq:linear-CIR}
        V_t = \bracketsig{\bell^\textnormal{CIR}} = \left( \bracketsig{\bsigma^\textnormal{CIR}} \right)^2,
    \end{align}
    where $\bell^\textnormal{CIR} := \left( \bsigma^\textnormal{CIR} \right) \shupow{2}$ with $\bsigma^\textnormal{CIR}$ satisfying  the non-linear algebraic equation
    $$ \left( \bsigma^\textnormal{CIR} \right) \shupow{2} = v \emptyword + \left( \left( \kappa \theta - \frac{\eta^2}{4} \right) \emptyword - \kappa \left( \bsigma^\textnormal{CIR} \right) \shupow{2} \right) \word{1} + \eta \bsigma^\textnormal{CIR} \word{2}. $$

    \begin{sqremark}
        The theoretical convergence, i.e.~proving that $\bsigma^\textnormal{CIR} \in \A$, seems intricate to obtain and is still an open problem. From the numerical perspective, the truncated version of \eqref{eq:linear-CIR} is well-defined and Figure~\ref{fig:ex-traj-CIR} suggests a certain type of convergence under the \citet{feller} condition $2 \kappa \theta > \eta$, with $V$ in \eqref{eq:linear-CIR} being the strong solution. Moreover, Figure~\ref{fig:sig-pricing}~(b) below also provides a numerical convergence of pricing of our truncated representation \eqref{eq:linear-CIR} in the context of the Heston model, where the Feller condition is not satisfied. In this case, we expect the representation in \eqref{eq:linear-CIR} to yield a weak solution.
    \end{sqremark}

    \begin{sqexample} \label{ex:elements_CIR}
        Up to order 3, the linear representation of the Cox-Ingersoll-Ross process reads
        \begin{align*}
            \bell^{\textnormal{CIR}} &
            = \left( v,
            \begin{pmatrix}
                \beta \\
                \gamma
            \end{pmatrix},
            \begin{pmatrix}
                \beta \mu & \beta \alpha \\
                \gamma \mu & \gamma \alpha
            \end{pmatrix},
            \begin{pmatrix}
                \beta \mu^2 & \beta \left( \mu - \frac{\beta}{2 v} \right) \alpha & \\
                \gamma \mu^2 & \gamma \left( \mu - \frac{\beta}{2 v} \right) \alpha & \\
                & \beta \mu \alpha & 0 \\
                & \gamma \mu \alpha & 0
            \end{pmatrix},
            \cdots \right),
        \end{align*}
        where $\mu = -\kappa, \beta = \mu v + \left( \kappa \theta - \frac{\eta^2}{4} \right), \gamma = \eta \sqrt{v}$ and $\alpha = \frac{\eta}{2 \sqrt{v}}$.
    \end{sqexample}

    \begin{figure}[H]
        \centering
        \subfloat[\centering $\kappa=1, \theta=0.25, \eta=0.7$]{{\includegraphics[width=\twoplotswidth]{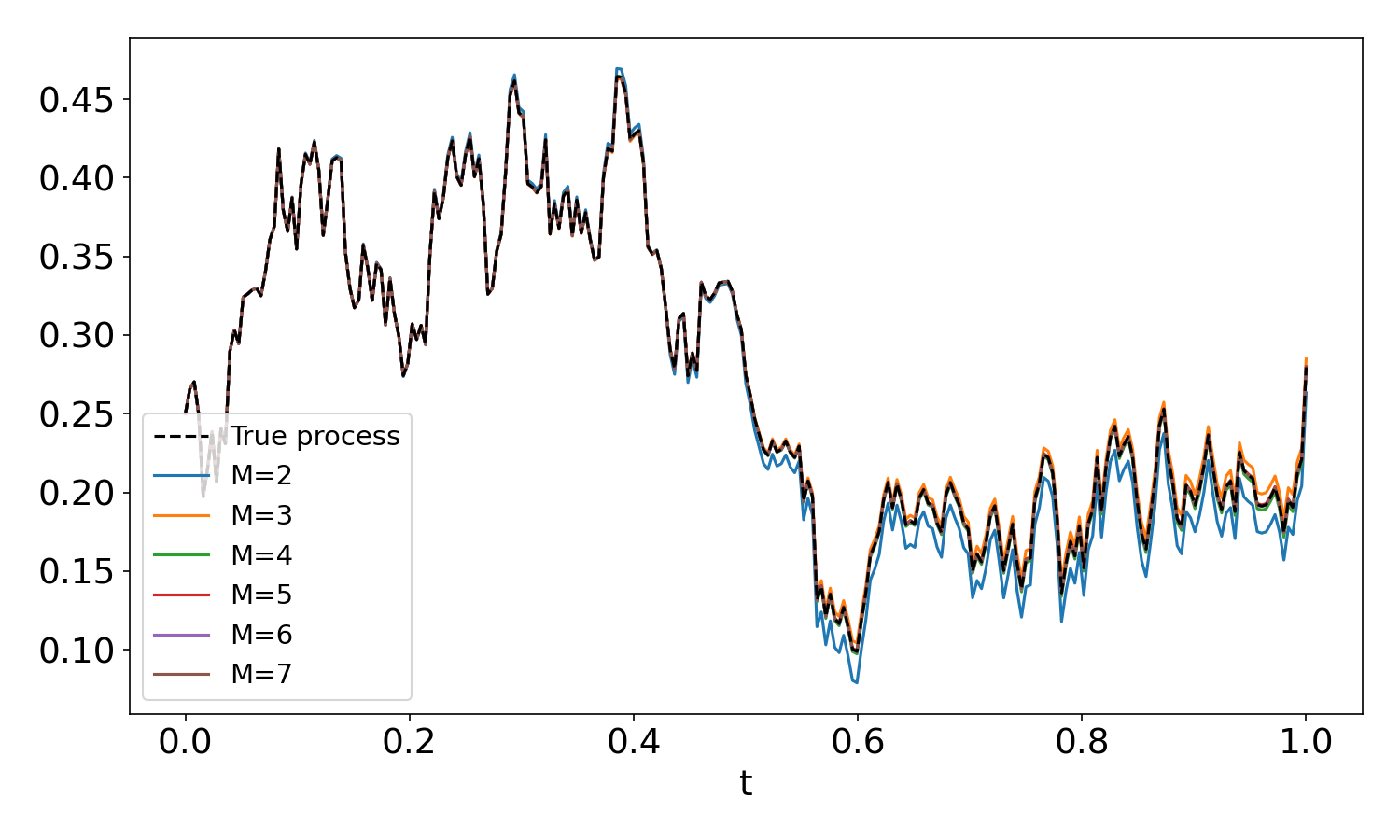}}} 
        \quad
        \subfloat[\centering $\kappa=4, \theta=0.25, \eta=1.4$]{{\includegraphics[width=\twoplotswidth]{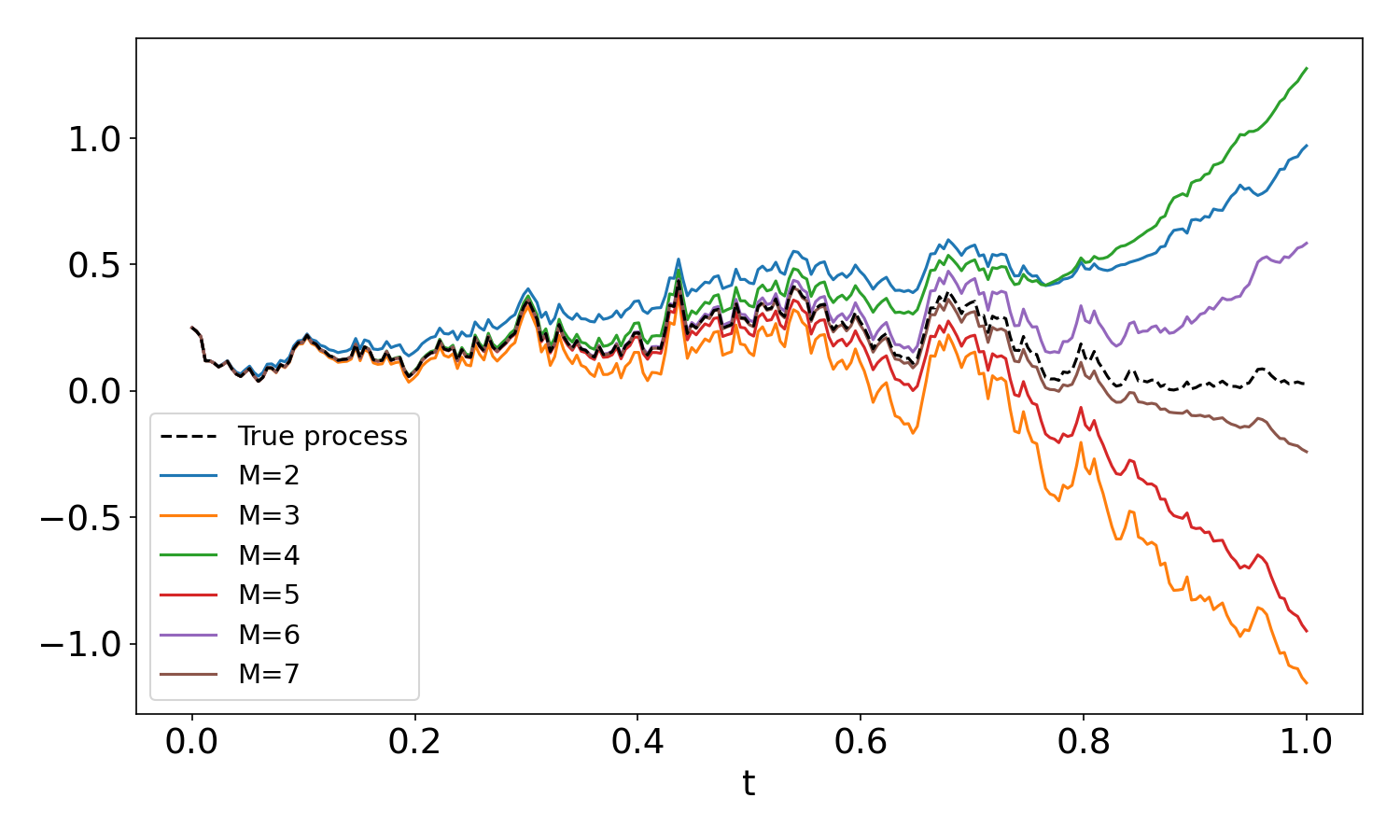}}} 
        \caption{Trajectories of a Cox-Ingersoll-Ross process against their truncated time-independent linear representation \eqref{eq:linear-CIR}, i.e. $\bracketsigtrunc[M]{\bell^{\textnormal{CIR}}}$, for several truncation orders $M$.}
        \label{fig:ex-traj-CIR}
    \end{figure}
    
   Going back to our signature volatility model, the representation \eqref{eq:linear-CIR} allows us to include in our framework the volatility of the Heston model \cite{heston}.

\subsubsection{Models based on path-dependent processes}

    As shown in \cite{linearfbm}, several path-dependent stochastic Volterra processes:
    
    $$ Z_t = Z_0 + \int_0^t K(t-s) (a_0 + a_1 Z_s) \d s + \int_0^t K(t-s) (b_0 + b_1 Z_s) \d W_s, $$
    
    with $Z_0,a_0,a_1,b_0,b_1 \in \R$, for certain locally square-integrable kernels $K$ also admit linear representations in the form $Z_t = \bracketsig{\bell_t}$. This includes non-semimartingale processes such as the Riemann-Liouville fractional Brownian motion
    
    $$ W_t^H = \int_0^t (t-s)^{H-1/2} \d W_s, \quad H \in (0,1). $$
    
    For instance, the (time-dependent) representation of $W^H$ reads
    
    $$ W_t^H = \bracketsig{\bell_t^{\textnormal{RL}}}, \quad \bell_t^{\textnormal{RL}} = t^{H - \frac{1}{2}} \sum_{n=0}^\infty t^{-n} \left( H - \tfrac{1}{2} \right)^{\bar{n}} \word{1} \conpow{n} \word{2}, $$
    
    where $(\cdot)^{\bar{n}}$ is the falling factorial. Please refer to \cite[Theorem 4.5]{linearfbm} for more details and illustrations on such representations. Again any volatility process that is an analytic function of such processes falls into the framework of signature volatility models \eqref{eq:sigmodel1}-\eqref{eq:sigmodel2} thanks to the shuffle property, recall \eqref{eq:fanalytic}. This includes for instance the class of Volterra polynomial models \cite{abi2022joint}, in particular Volterra \cite{abi2024volatility} and rough Bergomi models \cite{bayer2016pricing}. \\

    As as final example, the delayed equation (DE) process $U$, given by 
    \begin{align} \label{eq:sde-DE-exp}
        \d U_t &
        = \left( a_1 + b_1 U_t + c_1 \int_0^t e^{\lambda_1 (t - s)} U_s \d s \right) \d t + \left( a_2 + c_2 \int_0^t e^{\lambda_2 (t - s)} U_s \d s \right) \d W_t, \quad U_0 = u \in \R,
    \end{align}
    
    for some $a_i, b_i, c_i, \alpha_i \in \R$, can be represented as a linear combination of the signature of the time-extended Brownian motion with
    \begin{align} \label{eq:repdelay}
        U_t = \bracketsig{\bell^{\textnormal{DE}}},
        \quad
        \bell^{\textnormal{DE}} = (u \emptyword + a_1 \word{1} + a_2 \word{2}) \coninv{\emptyword - b_1 \word{1} - \word{1} \left( c_1 \shuexp{\lambda_1 \word{1}} \word{1} + c_2 \shuexp{\lambda_2 \word{1}} \word{2} \right)}.
    \end{align}
    
    The reader can refer to \cite[Theorem 4.4]{linearfbm} for more details and illustrations on the linear delayed equation process.
    Similar types of volatility dynamics have been studied in \citet{delayed_drift} in the framework of option pricing when $c_2=0$.
    
    \begin{sqexample} \label{ex:elements_DE}
        Up to order 2, the linear form of a delayed equation process reads
        $$ \bell^{\textnormal{DE}} = \left( u,
        \begin{pmatrix}
            a_1 + b_1 u \\
            a_2
        \end{pmatrix},
        \begin{pmatrix}
            c_1 u + b_1 (a_1 + b_1 u) & c_2 u \\
            a_2 b_1 & 0
        \end{pmatrix},
        \cdots \right), $$
        see Figure \ref{fig:ex-traj-DE} for a numerical illustration of the signature representation.
    \end{sqexample}
    
    \begin{figure}[H]
        \centering
        \subfloat[\centering $a_1=1, b_1=1, c_1=-2, \alpha_1=1, a_2=-0.5, c_2=-2, \alpha_2=1.5, z=0.25$]{{\includegraphics[width=\twoplotswidth]{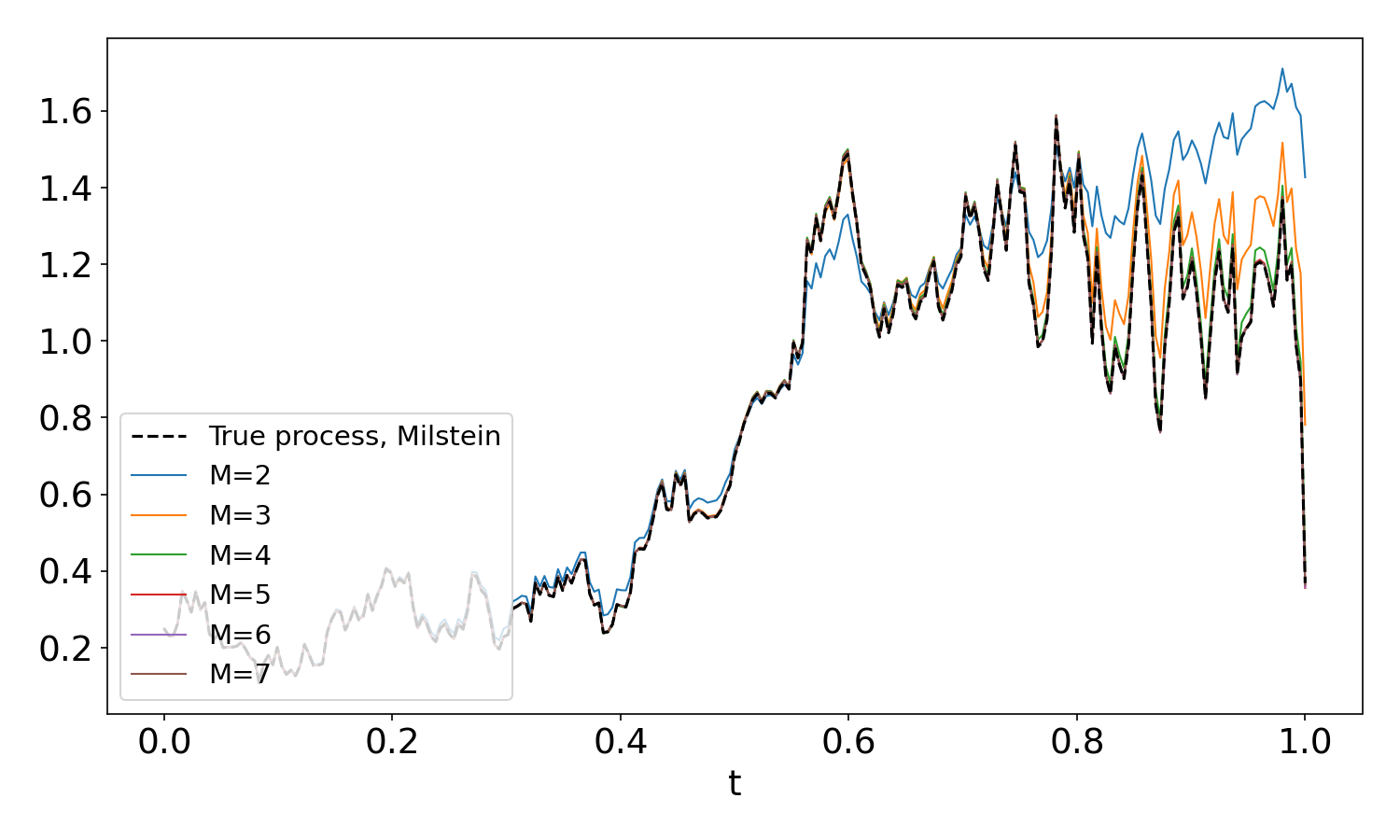}}} 
        \quad
        \subfloat[\centering $a_1=2, b_1=2, c_1=-2, \alpha_1=4, a_2=1, c_2=-1, \alpha_2=-4, z=0.25$]{{\includegraphics[width=\twoplotswidth]{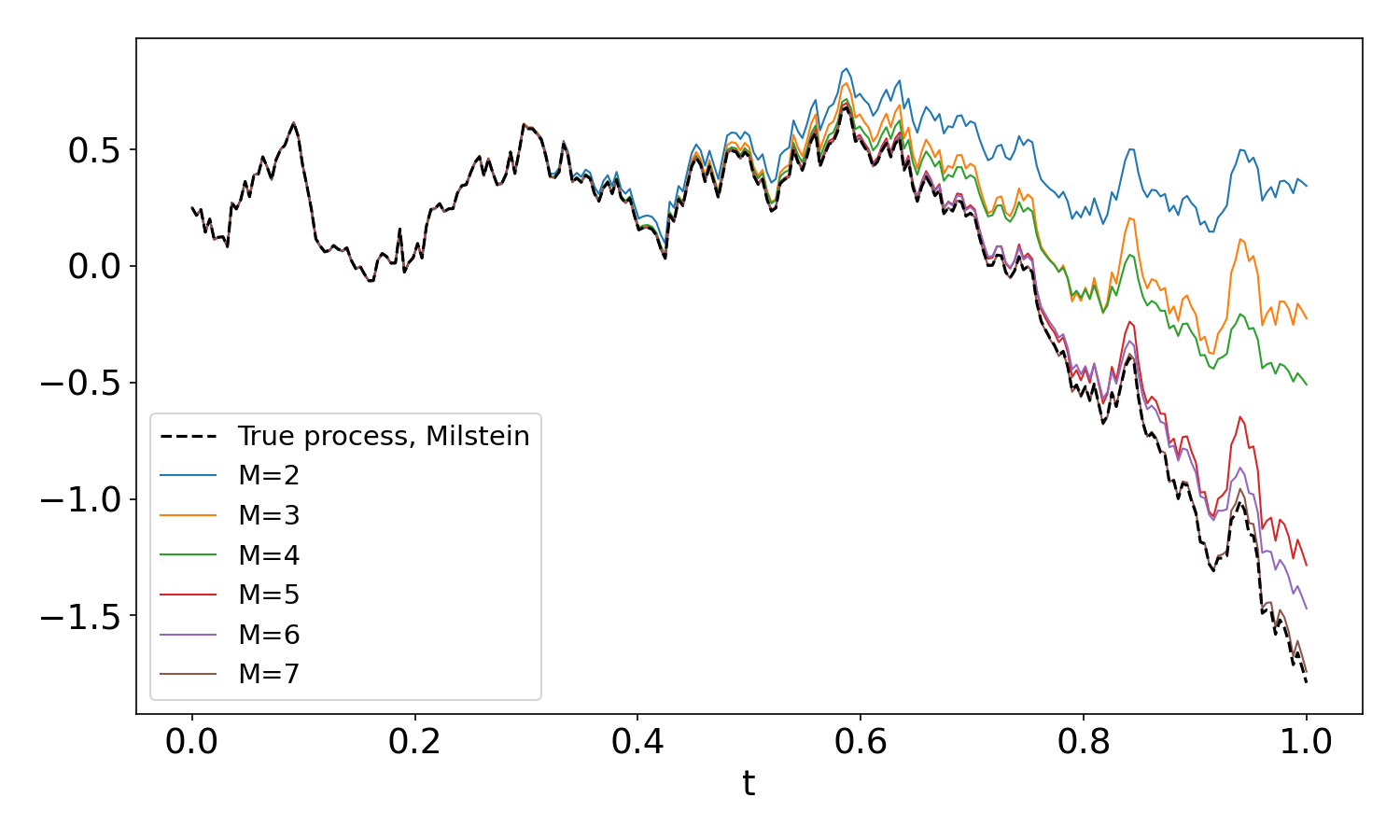}}} 
        \caption{Trajectories of a delayed equation process against their truncated time-independent linear representation \eqref{eq:repdelay}, i.e. $\bracketsigtrunc[M]{\bell^{\textnormal{DE}}}$, for several truncation orders $M$.}
        \label{fig:ex-traj-DE}
    \end{figure}

    \begin{remark}
        The signature framework also allows for multi-factor volatilities, by considering a multi-dimensional Brownian motion, see for instance \cite[Remarks 4.6 and 4.8]{linearfbm}. However the inclusion of a multi-dimensional underlying Brownian motion induces a significant numerical challenge as the number of elements in truncated linear functionals increases exponentially with the dimension.
    \end{remark}

\subsection{Leverage effect}

    When $\Sigma$ is a semimartingale, the \textit{leverage coefficient} in the model \eqref{eq:sigmodel1}-\eqref{eq:sigmodel2} is defined as the instantaneous correlation between the log-price $S$ and its instantaneous volatility $|\Sigma|$, with 
    \begin{align*}
        \frac{\d [\log S, |\Sigma|]_t}{\sqrt{\d [\log S]_t} \sqrt{\d [|\Sigma|]_t}},
    \end{align*}
    
    In practice, the leverage coefficient is negative on equity markets and one usually would like to control its sign and intensity via the correlation parameter $\rho$ between $B$ and $W$ in \eqref{eq:BM}. Since $\Sigma$ is not necessarily positive at all time, the leverage coefficient can flip sign in the model \eqref{eq:sigmodel1}-\eqref{eq:sigmodel2}, which is not realistic. The next Lemma provides a necessary and sufficient condition on the coefficients $\bsigma$ to control the sign of the leverage coefficient and its driving intensity $\rho$ at all time.
    
    \begin{lemma} \label{lem:leveff}
        Assume $\bsigma \in \I'$, then the covariation between the log-price and the instantaneous volatility is proportional to $\rho$, and of the same sign if $\bracketsig{\bsigma_t \proj{2}} \geq 0$, i.e.
        \begin{align} \label{eq:leveff}
            \frac{\d}{\d t} [\log S, |\Sigma|]_t = \rho \, \bracketsig{\bsigma_t \proj{2}} |\Sigma_t|.
        \end{align} 
    \end{lemma}
    
    \begin{proof} 
        The dynamics of $\log S$ are in the form
        \begin{align} \label{eq:shufflelogS}
            \d \log S_t = -\frac{1}{2} \Sigma_t^2 \d t + \Sigma_t \d B_t = -\frac{1}{2} \bracketsig{\bsigma_t \shupow{2}} \d t + \bracketsig{\bsigma_t} \d B_t.
        \end{align}
        The instantaneous volatility of the model is $|\Sigma|$. By an application of Itô-Tanaka's formula on $|\Sigma|$ and of Lemma~\ref{lem:sig-ito}, we have
        \begin{align*}
            \d |\Sigma|_t &= \sign{\Sigma_t} \d \Sigma_t + \d L^0(t)
            \\ &
            = \sign{\Sigma_t} \bracketsig{\bsigma_t \proj{2}} \d W_t + \sign{\Sigma_t} \bracketsig{\bsigma_t \proj{1} + \tfrac{1}{2} \bsigma_t \proj{22} + \dot{\bsigma}_t} \d t + \d L^0(t),
        \end{align*}
        where $L^0$ is the local time of $\Sigma$ at 0. \eqref{eq:leveff} then trivially follows.
    \end{proof}

    Moreover controlling the sign of $\bracketsig{\bsigma_t \proj{2}}$ can be made explicit in the Stein-Stein, Quintic and Bergomi models as shown in the next example.  
    \begin{sqexample}
        ~
        \begin{itemize}
            \item
            For the Stein-Stein model driven by a Brownian motion, $\bsigma \proj{2} = \eta \emptyword$, see \eqref{eq:linear-OU}, and $\bracketsig{\bsigma \proj{2}} = \eta \geq 0$, 
            
            \item
            For the Quintic  model of \cite{quintic} constructed on a Brownian motion, $\bsigma = \alpha_0 \emptyword + \alpha_1 \word{2} + \alpha_3 \word{222} + \alpha_5 \word{22222}$ so that $\bsigma \proj{2} = \alpha_1 \emptyword + \alpha_3 \word{22} + \alpha_5 \word{2222}$, which gives
            $$ \bracketsig{\bsigma \proj{2}} = \alpha_1 + \frac{\alpha_3}{2} W^2_t + \frac{\alpha_5}{4!} W^4_t \geq 0, \quad t \geq 0, $$
            as long as $\alpha_1, \alpha_3, \alpha_5$ are non negative,
            
            \item 
            For the Bergomi model, see \cite{dupire1993, BergomiSmileII}, $\bsigma_t = \xi_0(t) \shuexp{\eta \word{2}}$ so that $\bsigma_t \proj{2} = \eta \xi_0(t) \shuexp{\eta \word{2}}$, which gives
            $$ \bracketsig{\bsigma_t \proj{2}} = \eta \xi_0(t) e^{\eta W_t} \geq 0, $$
            as long as $\xi_0(t), \eta$ are non negative, for all $ t \leq T$.
        \end{itemize}
    \end{sqexample}

\subsection{Approximated representations and comparison with exact representations} \label{sssec:regress}

    More generally, if exact linear representations are not available for certain processes, approximate representations can be obtained thanks to the universal approximation property \cite{fermanian2021embedding} of path-signatures. Informally any continuous functional $F$ can be approximated by a linear combination of signature elements
    $$ F(t, (W_s)_{0 \leq s \leq t}) \approx \bracketsig{\bm{f}},$$
    for $\bm{f} \in \tTA{M}$ and a given truncation order $M$, see \cite[Theorem 1]{kiraly2019kernels} for more details. \\

    In practice, one would perform a linear regression on trajectories of a given process $X$ against a finite linear combination of the signature elements. Such approach was used in \cite[Section 4.1.2]{cuchiero2022theocalib} to regress trajectories of the price process in the Heston model and the volatility process in a SABR-type model. In the same spirit, \citet[Section 4.3]{arribas-nonparam} regressed certain option payoffs on the signature of price trajectories to get an approximation. \\
    
    We are going to illustrate this regression approach on trajectories of an inverse CIR process and fractional Brownian motions and show some of its limitations. We will also compare it with a couple of exact representations from the previous Section.
    
    Let us now explicit the algorithm used to get the regress linear representation:
    
    \begin{algorithm}[H] \label{alg:regression}
        \caption{Regression against truncated signature}
        Assume the spot volatility is of the following form
        $$ X_t = F(t, (W_s)_{s \leq t}). $$
        
        \textbf{Input}:
        \begin{itemize}
            \item Fix $J > 0$ the number of points in the discretization of $[0, T]$,
            \item Fix $N > 0$ the number of realisations of a Brownian motion trajectory,
            \item Fix $M \geq 0$ the truncation order of $\bm{f}$,
            \item Fix $\beta_1$ and $\beta_2$ the $L^1$ and $L^2$ regularization parameters.
        \end{itemize}
        
        \textbf{Online}:
        \begin{enumerate}
            \item Generate $N$ realizations of the Brownian motion $W$, denoted by $W^{(1)}, \ldots, W^{(N)}$,
            \item For each realization $n = 1, \ldots, N$, compute $X^{(n)}$ and the truncated signature $\sig^{(n), \leq M}$ up to order $M$ for $t \in [0, T]$,
            \item
            Regress $(X^{(n)})_{1 \leq n \leq N}$ against $(\sig[~]^{(n), \leq M})_{1 \leq n \leq N}$ to learn the coefficients of $\bm{f} \in \tTA{M}$ that minimize
            $$ \mathcal{L} = \frac{1}{N} \frac{1}{J} \sum_{n=1}^N \sum_{j=1}^J \left| X^{(n)}_{t_j} - \left \langle \bm{f}, \sig[t_j]^{(n), \leq M} \right \rangle \right|^2 + \sum_{k=0}^M \sum_{\word{v} \in V_k} \left( \beta_1 \left| \bell^{\word{i_1 \cdots i_k}} \right| + \beta_2 \left( \bell^{\word{i_1 \cdots i_k}} \right)^2 \right). $$
        \end{enumerate}
    \end{algorithm}
    Note that the $L^1$ regularization serves the purpose of sparsifying the representation, which might be of interest. For example most elements of $\bell^\textnormal{OU}$ and some of $\bell^\textnormal{CIR}$ and $\bell^\textnormal{DE}$ are 0, see Examples~\ref{ex:elements_OU}, \ref{ex:elements_CIR} and \ref{ex:elements_DE}.
    
    \begin{figure}[H]
        \centering
        \subfloat[\centering $\kappa=1, \theta=0.5, \eta=1$.]{{\includegraphics[width=\twoplotswidth]{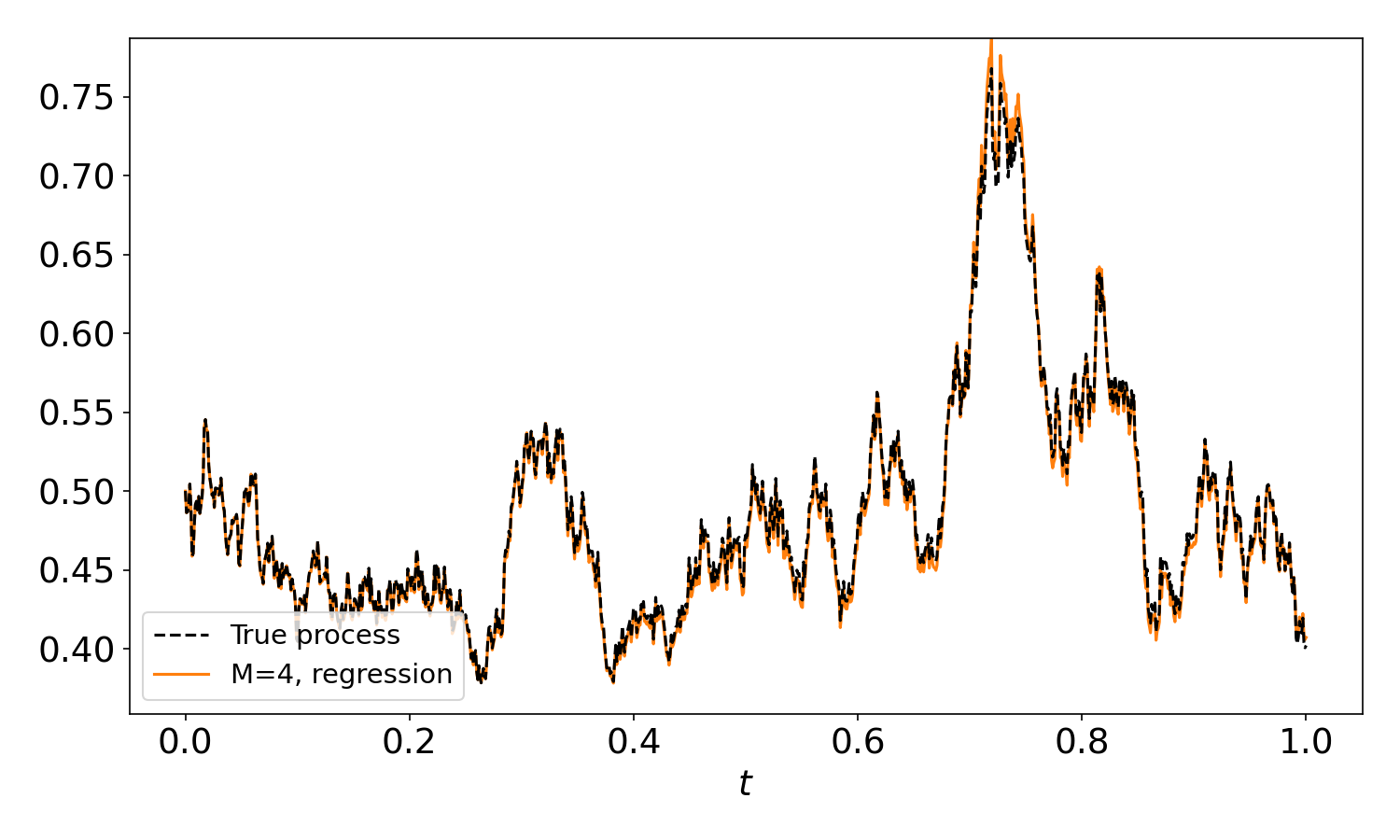}}}
        \quad
        \subfloat[\centering $\kappa=4, \theta=0.5, \eta=2$.]{{\includegraphics[width=\twoplotswidth]{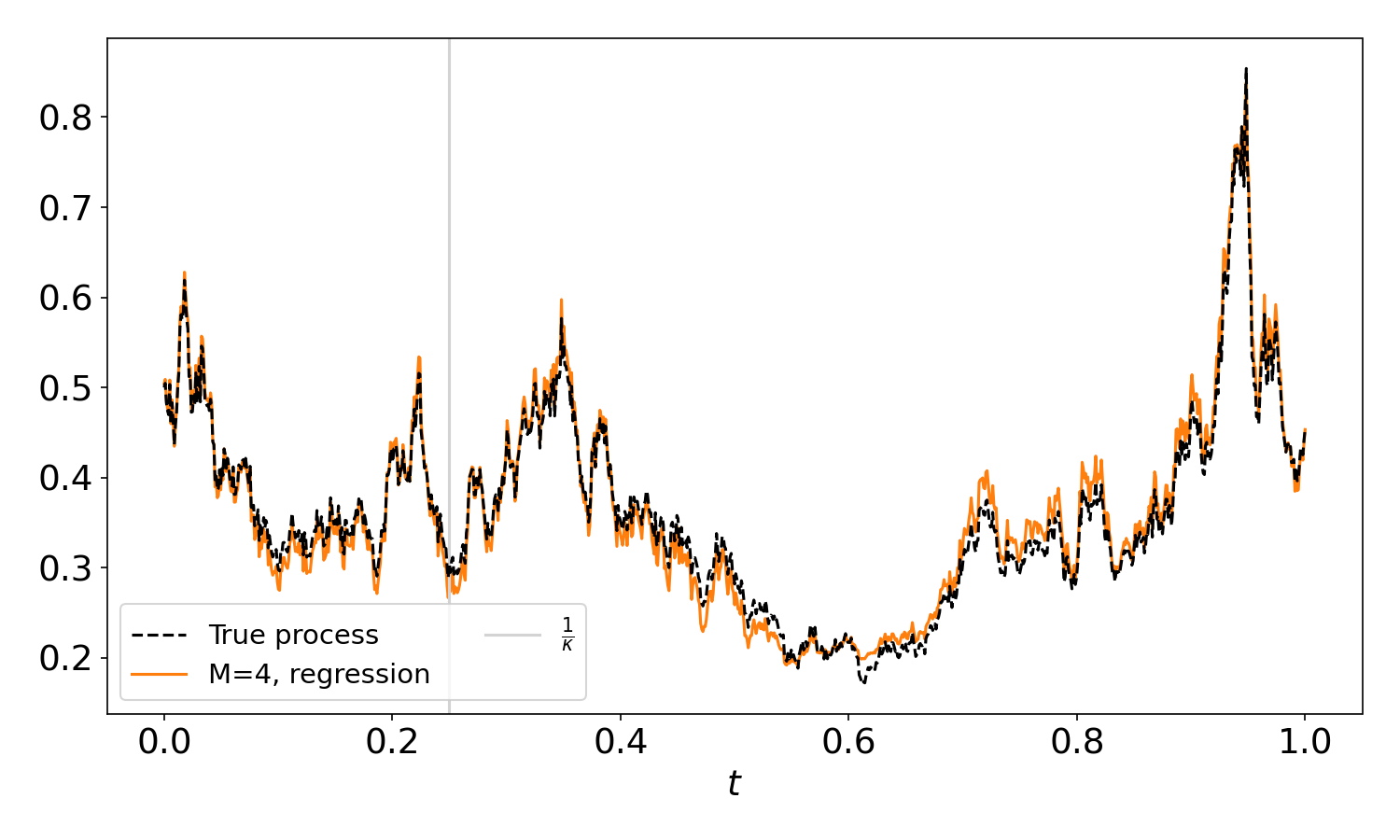}}}
        \caption{Trajectories of an inverse CIR process against their linear regression for $M=4$.}
        \label{fig:calib-32}
    \end{figure}
    
    We recall the dynamics of the inverse CIR process, being the variance in the 3/2 model \cite{32model}, driven by
    $$ d V_t = \kappa V_t (\theta - V_t) \d t + \eta V_t^{\frac{3}{2}} \d W_t. $$
    
    We see in Figure \ref{fig:calib-32} that a process without known linear representation can still be approximated quite well through a linear regression.

    \begin{figure}[H]
        \centering
        \subfloat[\centering $H=0.1$.]{{\includegraphics[width=\twoplotswidth]{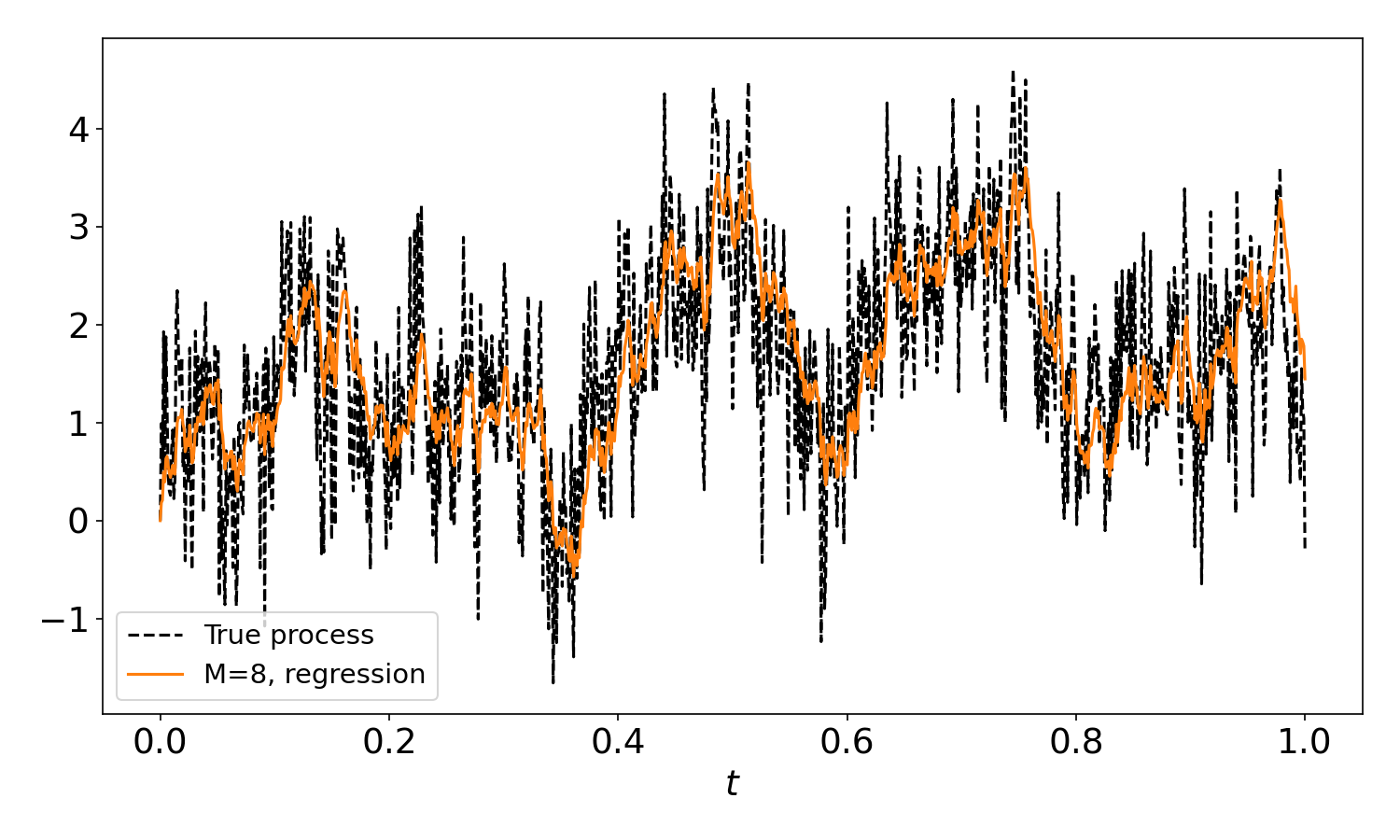}}}
        \quad
        \subfloat[\centering $H=0.3$.]{{\includegraphics[width=\twoplotswidth]{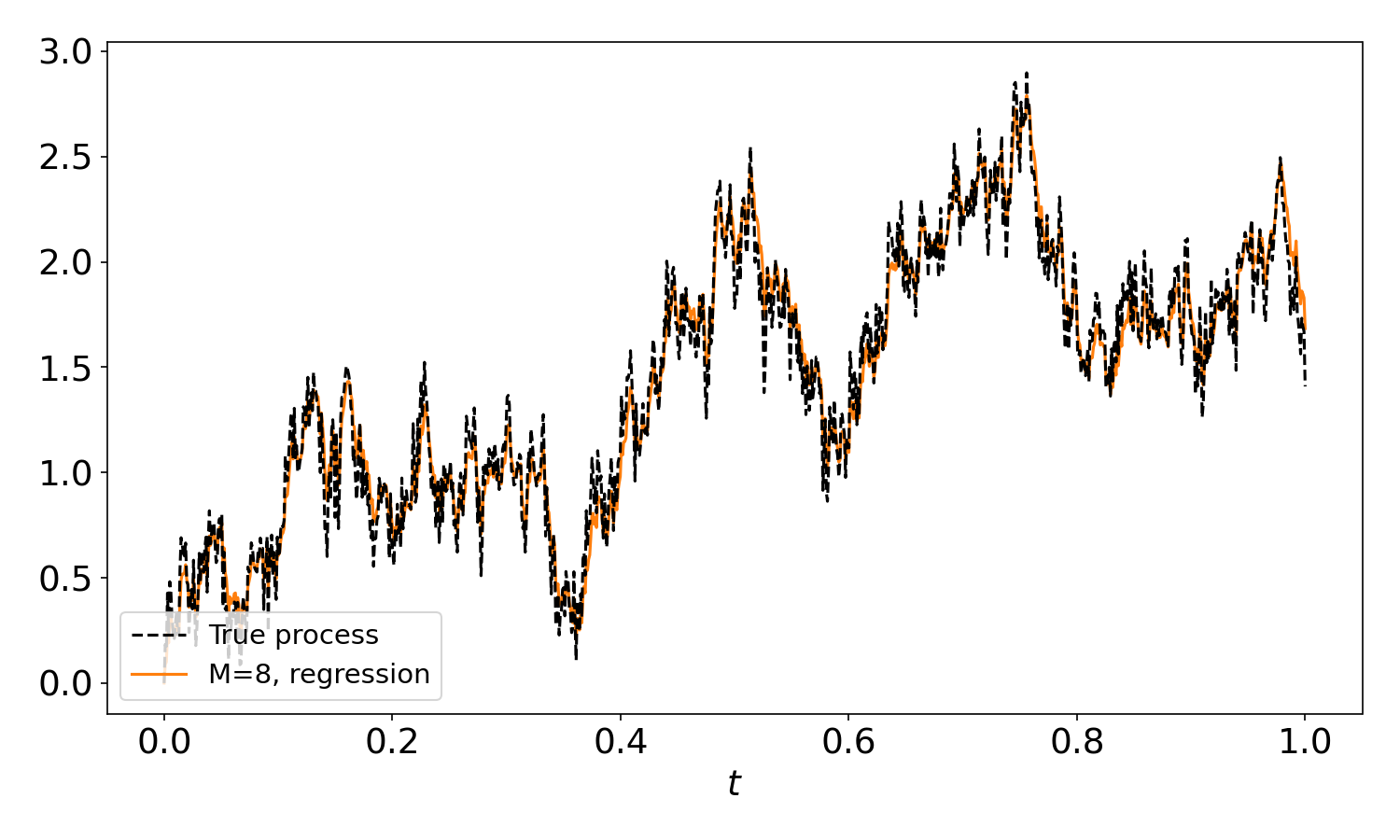}}}
        \caption{Trajectories of a fractional Brownian motion against their linear regression for $M=8$.}
        \label{fig:calib-rough}
    \end{figure}
    
    However, as shown on Figure \ref{fig:calib-rough}, the regression approach does not always work with small truncation levels. For instance, the ability to capture high roughness in a linear functional is not trivial to achieve whilst keeping the natural embedding $\widehat{W}_t = (t, W_t)$. \\

    We end this section with a brief comparison of approximate vs. exact representations. 
    As seen in the previous subsections, the convergence of the linear representation is quick for short horizons. However, when the horizon gets too large, relatively to the parameters of the represented models, e.g. the mean reversion rate, the truncated representations drastically deteriorates, see Figures \ref{fig:ex-traj-OU} to \ref{fig:ex-traj-DE}. Yet, as seen previously, linear regressions make quite stable representations over their training horizon and can thus be made over targeted horizons to control the stability. In Figure~\ref{fig:calib-OU}~(a) we can remark that the linear regression doesn't fit as well as the linear representation for short horizons (relative to $\kappa$), but does get a better fit when the linear representation starts loosing its stability. However, the stability of the regression up to the training horizon could be interpreted as overfitting, as illustrated in Figure~\ref{fig:calib-OU}~(b), where the regression too looses its stability when tested beyond its training horizon.
    
    \begin{figure}[H]
        \centering
        \subfloat[\centering $T=0.5$.]{{\includegraphics[width=\twoplotswidth]{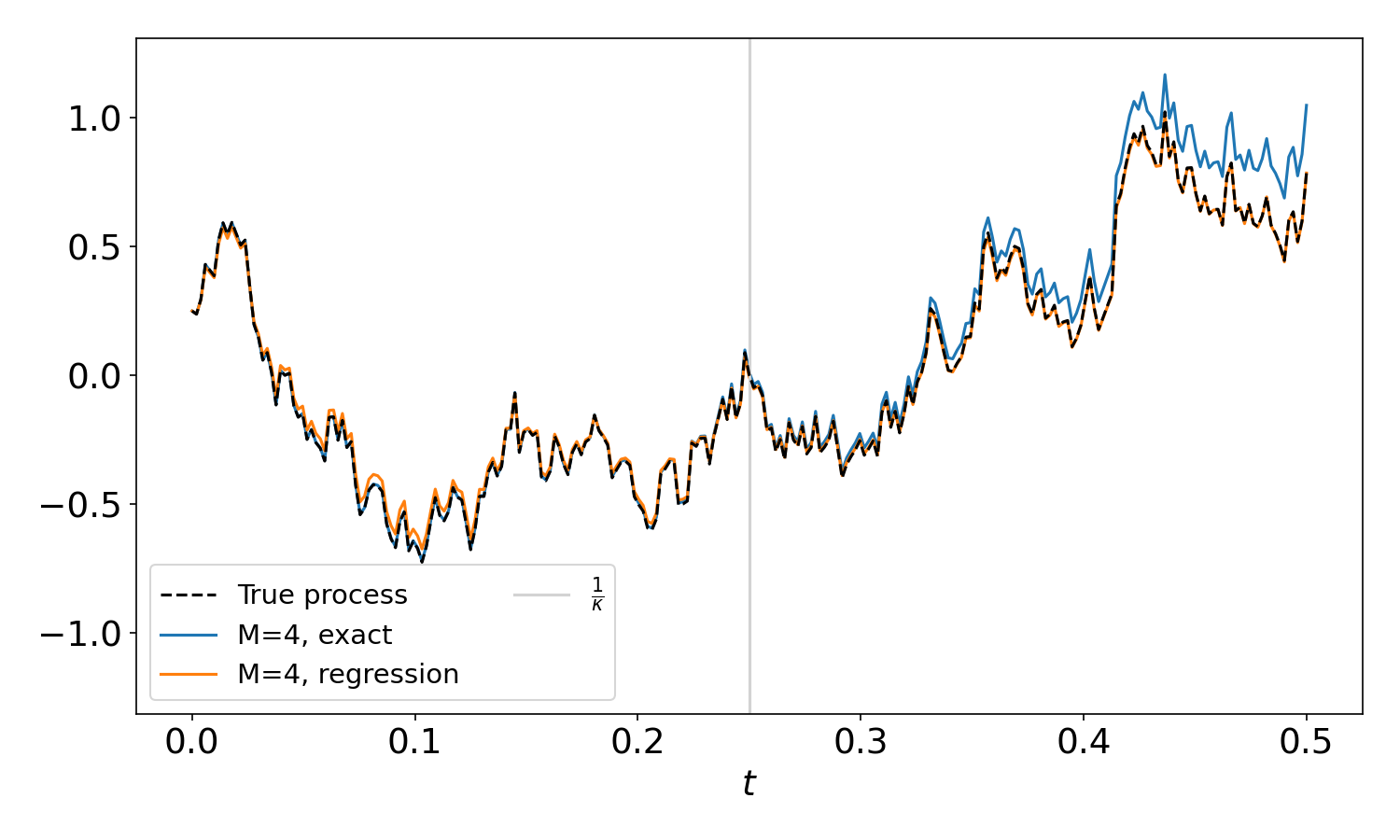}}}
        \quad
        \subfloat[\centering $T=2$.]{{\includegraphics[width=\twoplotswidth]{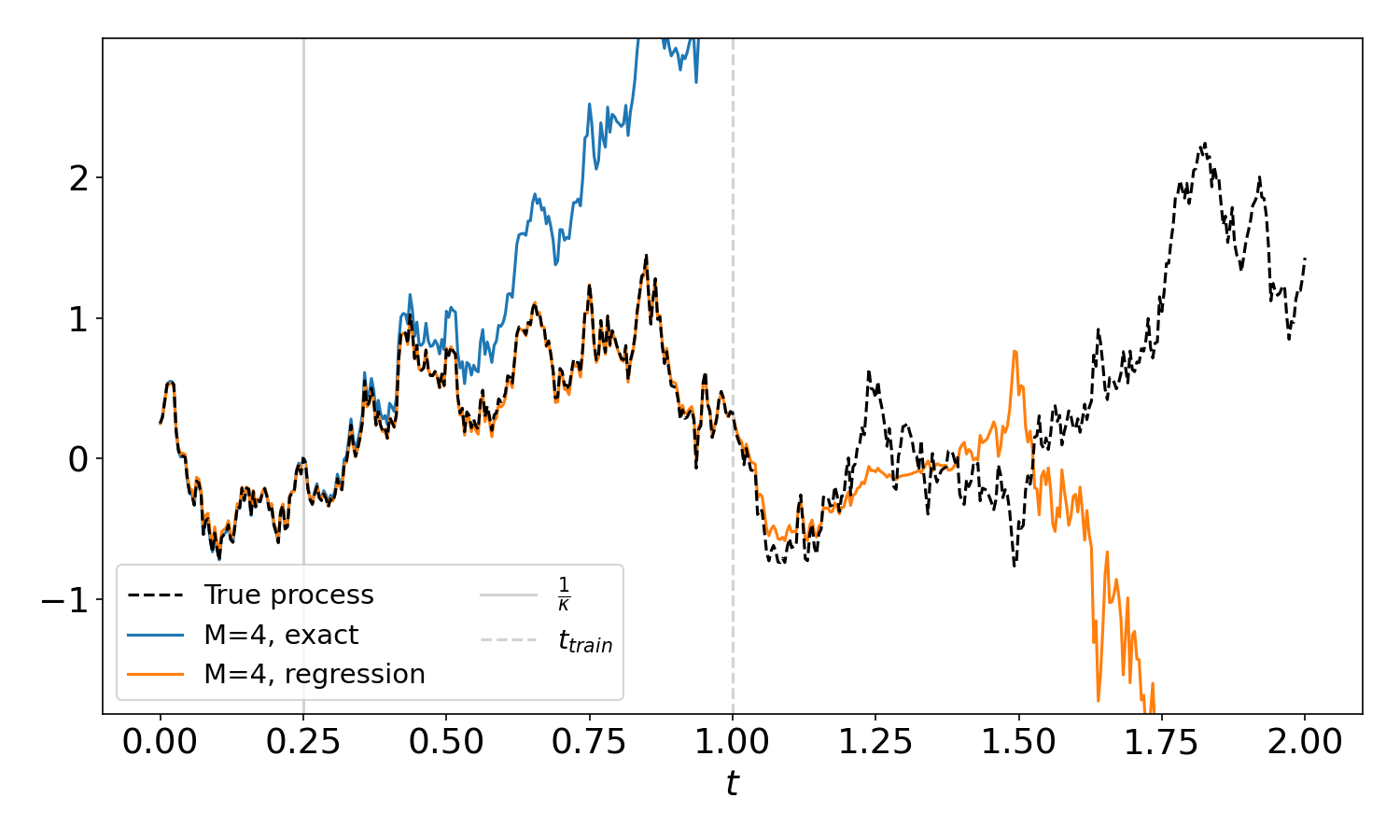}}}
        \caption{Trajectories of an Ornstein-Uhlenbeck process against their linear representation and linear regression. $\kappa=4, \theta=0.25, \eta=2$}
        \label{fig:calib-OU}
    \end{figure}
    
    Table~\ref{tab:calib-mse-OU} below displays a more complete comparison of both regression and exact representation. It displays the mean squared error between explicit simulations of an Ornstein-Uhlenbeck process and their linear representations (exact truncated and regressed against truncated signature) where $\kappa = 4, \theta = 0.25, \eta = 2$ and $v = 0.25$. The number of simulations and training horizon have been set to 100,000 and 1 year with 252 time steps respectively.
    
    \begin{table}[!ht]
        \begin{center}
            \begin{tabular}{|c|c|c|c|c|c|c|}
                \hline
                \multicolumn{2}{|c|}{\multirow{2}{*}{MSE}}
                & \multicolumn{5}{|c|}{Test horizon} \\
                \cline{3-7}
                \multicolumn{2}{|c|}{}
                & 3 months & 6 months & 1 year & 2 years & 4 years \\
                \hline
                \multirow{2}{*}{$M=2$}
                & Exact      & \textbf{3.065e$-$05} & \textbf{5.370e$-$03} & 1.969e$+$00 & 2.536e$+$01 & 2.625e$+$02 \\ 
                & Regression & 6.149e$-$03 & 8.564e$-$03 & \textbf{2.405e$-$02} & \textbf{6.239e$-$01} & \textbf{1.189e$+$01} \\
                \hline
                \multirow{2}{*}{$M=4$}
                & Exact      & \textbf{6.541e$-$08} & \textbf{1.736e$-$05} & 1.718e$+$00 & 4.243e$+$02 & 8.177e$+$04 \\
                & Regression & 8.002e$-$06 & 3.801e$-$05 & \textbf{1.079e$-$04} & \textbf{1.516e$+$00} & \textbf{1.213e$+$03} \\
                \hline
                \multirow{2}{*}{$M=6$}
                & Exact      & \textbf{5.908e$-$08} & \textbf{1.167e$-$06} & 3.173e$-$01 & 1.410e$+$03 & 4.872e$+$06 \\
                & Regression & 5.735e$-$06 & 8.426e$-$06 & \textbf{1.596e$-$07} & \textbf{1.066e$+$00} & \textbf{3.188e$+$04} \\
                \hline
            \end{tabular}
        \end{center}
        \caption{Mean-squared-error between simulations of an Ornstein-Uhlenbeck process against its linear representation and linear regression when trained over only 1 year. $\kappa = 4, \theta = 0.25, \eta = 2$.}
        \label{tab:calib-mse-OU}
    \end{table}
    
    We can see that the regression is quite good for horizons up to 1 year, but as soon as the test horizon grows too large compared to the train horizon(s), the regression doesn't fit well anymore and becomes explosive, just similar to the truncated exact representation.

\section{The joint characteristic functional} \label{sec:charfun}

    The following theorem provides the joint conditional characteristic functional of the log-price $\log S$ and the integrated variance $\bar{V} := \int_0^{\cdot} \Sigma^2_s \d s$ in the model \eqref{eq:sigmodel1}--\eqref{eq:sigmodel2} in terms of a solution to an infinite-dimensional system of time-dependent $\eTA[2]$-valued Riccati differential equations. 
    
    \begin{theorem} \label{theo:charfun}
        Let $f, g: [0, T] \to \C$ be measurable and bounded functions and $\bsigma: [0, T] \to \A$. Assume that there exists $\bpsi \in \I'$, solution to the following infinite-dimensional system of time-dependent Riccati equations
        \begin{align} \label{eq:Ric}
            -\dot{\bpsi}_t &
            = \frac{1}{2} (\bpsi_t \proj{2}) \shupow{2} + \rho f(t) (\bsigma_t \shuprod \bpsi_t \proj{2}) + \frac{1}{2} \bpsi_t \proj{22} + \bpsi_t \proj{1} + \left( \frac{f(t)^2 - f(t)}{2} + g(t) \right) \bsigma_t \shupow{2},
            \\ \bpsi_T &
            = 0.
        \end{align}
    
        Define the processes
        \begin{align} \label{eq:U}
            U_t &
            = \bracketsig{\bpsi_t} + \int_0^t f(s) \d \log S_s + \int_0^t g(s) \d \bar{V}_s,\\
            M_t &
            = e^{U_t}. \label{eq:M}
        \end{align}
        Then $M$ is a local martingale. If in addition it is a true martingale, then the following expression holds for the joint characteristic functional of $(\log S, \bar{V})$:
        \begin{equation} \label{eq:charfun}
             \E \left[ \left. \exp \left( \int_t^T f(s) \d \log S_s + \int_t^T g(s) \d \bar{V}_s \right) \right| \mathcal{F}_t \right]
             = \exp \left( \bracketsig{\bpsi_t} \right), \quad t \leq T.
        \end{equation}
    \end{theorem}
    
    \begin{proof}
        To show that $M$ is a local martingale we show that its part in $\d t$ vanishes using Itô's formula. The dynamics of $M$ read
        \begin{align} \label{eq:dM}
            \d M_t = M_t \left( \d U_t + \frac{1}{2} \d \left[ U \right]_t \right).
        \end{align}
        We start by deriving the dynamics of $U$ in \eqref{eq:U}. By using Lemma \ref{lem:sig-ito}, since $\bpsi\in \I'$ by assumption, we have that $\left( \bracketsig{\bpsi_t} \right)_{t \geq 0}$ is a semimartingale with dynamics
        \begin{align*}
            \d \bracketsig{\bpsi_t} &
            = \bracketsig{\dot{\bpsi}_t + \bpsi_t \proj{1} + \tfrac{1}{2} \bpsi_t \proj{22}} \d t + \bracketsig{\bpsi_t \proj{2}} \d W_t.
        \end{align*}
        Combining the previous identity with $d \bar{V}_t = \langle \bsigma_t \shupow{2}, \sig \rangle \d t$, recall \eqref{eq:shufflesigma2}, and the dynamics of $\log S$ in \eqref{eq:shufflelogS}, we obtain that
        \begin{align*}
             \d U_t &
             = \d \bracketsig{\bpsi_t} + f(t) \d \log S_t + g(t) \d \bar{V}_t
             \\ &
             = \bracketsig{\dot{\bpsi}_t + \bpsi_t \proj{1} + \tfrac{1}{2} \bpsi_t \proj{22}} \d t + \bracketsig{\bpsi_t \proj{2}} \d W_t
             + \left( g(t) - \tfrac{1}{2} f(t) \right) \bracketsig{\bsigma_t \shupow{2}} \d t + f(t) \bracketsig{\bsigma_t} \d B_t
             \\ &
             = \bracketsig{\dot{\bpsi}_t + \bpsi_t \proj{1} + \tfrac{1}{2} \bpsi_t \proj{22} + \left( g(t) - \tfrac{1}{2} f(t) \right) \bsigma_t \shupow{2}} \d t 
             + \bracketsig{\bpsi_t \proj{2}} \d W_t + f(t) \bracketsig{\bsigma_t} \d B_t. 
        \end{align*}
        Using the shuffle product of Proposition~\ref{prop:shufflepropertyextended} and the fact that $B$ and $W$ are correlated, recall \eqref{eq:BM}, we get that the quadratic variation of $U$ is
        \begin{align*}
            \d \left[ U \right]_t &
            = \bracketsig{(\bpsi_t \proj{2}) \shupow{2} + 2 \rho f(t) (\bpsi_t \proj{2} \shuprod \bsigma_t) + f(t)^2 \bsigma_t \shupow{2}} \d t.
        \end{align*}
        This yields that the drift of $dM_t/M_t$ in \eqref{eq:dM} is given by 
        \begin{align*}
           & \left \langle \dot{\bpsi}_t + \bpsi_t \proj{1} + \tfrac{1}{2} \bpsi_t \proj{22} + \left( g(t) - \tfrac{1}{2} f(t) \right) \bsigma_t \shupow{2} + \tfrac{1}{2} \left[ (\bpsi_t \proj{2}) \shupow{2} + 2 \rho f(t) (\bpsi_t \proj{2} \shuprod \bsigma_t) + f(t)^2 \bsigma_t \shupow{2} \right], \sig \right \rangle,
        \end{align*}
        which is equal to $0$ from the Riccati equations \eqref{eq:Ric}. This shows that $M$ is a local martingale. By assumption $M$ is even a true martingale. After observing that the terminal value of $M$, is given by
        $$ M_T = \exp \left( \int_0^T f(s) \d \log S_s + \int_0^T g(s) \d \bar{V}_s\right), $$
        recall that $\bpsi_T=0$, we obtain 
        \begin{align*}
            \E \left[ \exp \left( \int_0^T f(s) \d \log S_s + \int_0^T g(s) \d \bar{V}_s \right) \bigg| \mathcal F_t\right] = \E \left[ M_T \big| \mathcal F_t\right] = M_t = \exp \left( U_t \right),
        \end{align*}
        which yields \eqref{eq:charfun}. 
    \end{proof}

    \begin{sqremark}
        If $\bpsi$ is such that $\Re \bracketsig{\bpsi_t} \leq 0$ for all $t \in [0, T]$ a.s., $M$ would be a martingale. We expect this relation to hold whenever $\Re f(t) \in [0, 1]$ and $\Re g(t) \leq 0$. 
    \end{sqremark}

    Theorem~\ref{theo:charfun} is a verification result to obtain the exponentially affine representation of the joint characteristic functional \eqref{eq:charfun}. It disentangles the algebraic affine structure in infinite dimension. It can be related to \cite[Theorem 4.24]{cuchiero2023polynomial} if one considers the signature of the three dimensional process $((t, W_t, B_t))_{t \geq 0}$ there. It relies on the two crucial assumptions that a well-behaved $\I$-valued solution $\bpsi$ exists to the Riccati equation \eqref{eq:Ric}, together with the true martingality of $M$. These assumptions seem very intricate to prove, even in one dimensional settings, partial results in these directions for $T((\R))$-valued Riccati equations can be found in \cite[Section 6]{cuchiero2023polynomial} and \cite{abishaun-polyou}. We note that no semimartingality assumption for $\Sigma_t = \langle \bsigma_t, \sig, \rangle$ is required in Theorem~\ref{theo:charfun}. \\

    In the following two sections, we will validate the representation \eqref{eq:charfun} numerically and we will highlight the application of Theorem~\ref{theo:charfun} to the pricing and the hedging of several contingent claims by Fourier inversion techniques. The general functions $f$ and $g$ allow for enough flexibility to cover a broad set of contingent claims. For instance:
    
    \begin{enumerate}
        \item Certain vanilla options that depend on the values $(S_T, \bar{V}_T)$, like European call and put options and volatility swaps, are recovered using constant $f$ and $g$.
        
        \item Geometric Asian options on the average $\frac{1}{T} \int_0^T \log S_s \d s$ can be recovered by setting $f(s) := iu \frac{T-s}{T}$, since 
        \begin{align}
            \int_t^T f(s) \d \log S_s &
            = iu \int_t^T \frac{T-s}{T} \d \log S_s
            = iu \frac{1}{T} \int_t^T \log S_s \d s - iu \frac{T-t}{T} \log S_t.
        \end{align}
    \end{enumerate}

\section{Pricing by Fourier methods} \label{sec:pricing}

    In this section, we show how Theorem~\ref{theo:charfun} can be applied to price European options as well as $q$-Volatility swaps using Fourier inversion techniques in our signature volatility model \eqref{eq:sigmodel1}-\eqref{eq:sigmodel2}. The study of Asian options has been included for the interested reader in Appendix~\ref{apn:asian}. All of our numerical results validate the exponentially affine representation \eqref{eq:charfun}. We also provide a calibration to market volatility surface for the S\&P 500. \\
    
    For the numerical implementation, we consider a truncated version of $S$ from \eqref{eq:sigmodel1}-\eqref{eq:sigmodel2}, denoted by $S^{\leq M}$ and defined by
    $$ \d S_t^{\leq M} = S_t^{\leq M} \bracketsig{\bsigma_t^{\leq M}} \d B_t, $$
    where $\bsigma_t^{\leq M} \in \tTA[2]{M}$ has its first $M$ levels coincide with $\bsigma_t$ and is $0$ elsewhere. Finally, in order to ease notations in the sequel we assume $S_0=1$ (recall that the short rate here is assumed to be $0$). Error estimates for truncated representations are discussed in Proposition 3.6 (\textbf{iii}) of \cite{linearfbm}.

\subsection{European options}

    Let us consider a European call option on $S$ with maturity $T>0$ and strike $K>0$. Its price at time $t \leq T$ is given by $C_t(T, K) = \E [(S_T - K)^+ | \F_t]$. From \citet*{fourierlewis}, one can price this option using the Fourier inversion formula
    \begin{equation} \label{eq:lewis}
        C_t(S_t; T, K) = S_t - \frac{K}{\pi} \int_0^\infty \Re \left[ e^{i (u - \frac{i}{2}) k_t} \phi_t \left( u - \tfrac{i}{2} \right) \right] \frac{\d u}{\left( u^2 + \tfrac{1}{4} \right)},
    \end{equation}
    where $\phi$ is the conditional characteristic function $\phi_t(u) = \E \left[ e^{iu \log \frac{S_T}{S_t}} | \F_t \right]$ and $k_t = \log \frac{S_t}{K}$.

    Moreover, we add a `control variate' to quicken the convergence, as in \citet{controlvariate-anderander}.
    Given $\sigma_{\textnormal{BS}} < \infty$, one has
    \begin{equation} \label{eq:control-variate}
        C_t(S_t; T, K) = C_t^{\textnormal{BS}}(S_t; T, K) - \frac{K}{\pi} \int_0^\infty \Re \left[ e^{i (u - \frac{i}{2}) k_t} \left( \phi_t \left( u - \tfrac{i}{2} \right) - \phi_t^{\textnormal{BS}} \left( u - \tfrac{i}{2} \right) \right) \right] \frac{\d u}{\left( u^2 + \tfrac{1}{4} \right)},
    \end{equation}
    where
    \begin{align}
        C_t^{\textnormal{BS}}(S_t; K, T) &
        = \mathcal{N}(d_1) S_t - \mathcal{N}(d_2) K,
    \end{align}
    with
    \begin{align}
        d_1 = \frac{1}{\sigma_{\textnormal{BS}} \sqrt{T-t}} \left( \log \frac{S_t}{K} + \frac{\sigma_{\textnormal{BS}}^2}{2} (T-t) \right),
        \quad
        d_2 = d_1 - \sigma_{\textnormal{BS}} \sqrt{T-t}
    \end{align}
    and
    \begin{align} \label{cf-black-scholes}
        \phi_t^{\textnormal{BS}}(u) &
        = \exp \left[ -\frac{\sigma_{\textnormal{BS}}^2}{2} \left( u^2 + iu \right) (T-t) \right].
    \end{align}
    This Black-Scholes control variate can also be applied to other products and, as will be seen in Section \ref{sec:hedging}, to Fourier hedging. The quantity $\sigma_{\textnormal{BS}}$ can be determined to ensure approximate moment matching for instance, and one can use the characteristic function $\phi_t$ to approximate (by finite differences) the second order cumulant of the distribution of the log price $\log S$.

    ~\\
    Numerically, for the discretization of the Fourier integral, our numerical experiments show that Gauss-Laguerre quadrature outperform other quadrature rules in our class of models, which is in line with the empirical findings that appeared in \cite{schmelzle}. In particular, the higher the maturity is, the lower the degree of the quadrature needs to be for the same set of parameters. In addition, the use of a control variate makes computations much quicker yielding fewer calls of the characteristic function to make the quadrature converge, see Figure \ref{fig:error_tilde_M}. Moreover, we can see that the added value of the control variate, in terms of speed of convergence, grows with the time horizon. See \cite[Section 6]{schmelzle} for more details on numerical refinements. \\

    For a given signature volatility model \eqref{eq:sigmodel1}-\eqref{eq:sigmodel2}, i.e.~a set of parameters $\bsigma: [0,T] \to \A$, an application of Theorem~\ref{theo:charfun} with $f(t)=iu$ and $g=0$ provides the characteristic function $\phi_t(u)$ modulo the solution $\bpsi$ to the $\eTA[2]$-valued Riccati equation \eqref{eq:Ric}, which allows us to compute the price of call options in the signature volatility model using \eqref{eq:control-variate} together with a suitable truncation and discretization of the Riccati equation \eqref{eq:Ric}. The truncation rule for $\bpsi$ is not trivial and requires a little bit of care, because the shuffle products cannot be exact, as each step of the discretized ODE would double the truncation order of $\bpsi$. For the numerical implementation, we decided to fix the order $\tilde{M}$ of $\bpsi$ for each step and hence only have a shuffle product projected on $\tTA[2]{\tilde{M}}$, i.e.~$\widetilde{\shuprod}: (\tTA[2]{\tilde{M}})^2 \to \tTA[2]{\tilde{M}}$. Obviously, choosing $\tilde{M}$ lower than $2M$, where $M$ is the truncation order of $\bsigma$, also induces an approximation of shuffle product in $\bsigma \shupow{2}$, which greatly deteriorated the quality of the convergence. Moreover, choosing $\tilde{M}$ greater than $2M$ does not seem to improve the quality of the convergence. Figure~\ref{fig:error_tilde_M} illustrates the error induced by the choice of $\tilde{M}$. We thus fixed $\tilde{M}=2M$ throughout our experiments. Said differently for a given signature volatility $\bsigma \in \tTA[2]{M}$ the truncated solution $\bpsi$ to the Riccati equation \eqref{eq:Ric} is an element of $\tTA[2]{2M}$.
    
    \begin{figure}[H]
        \centering
        \subfloat[\centering $M=3$.]{{\includegraphics[width=\twoplotswidth]{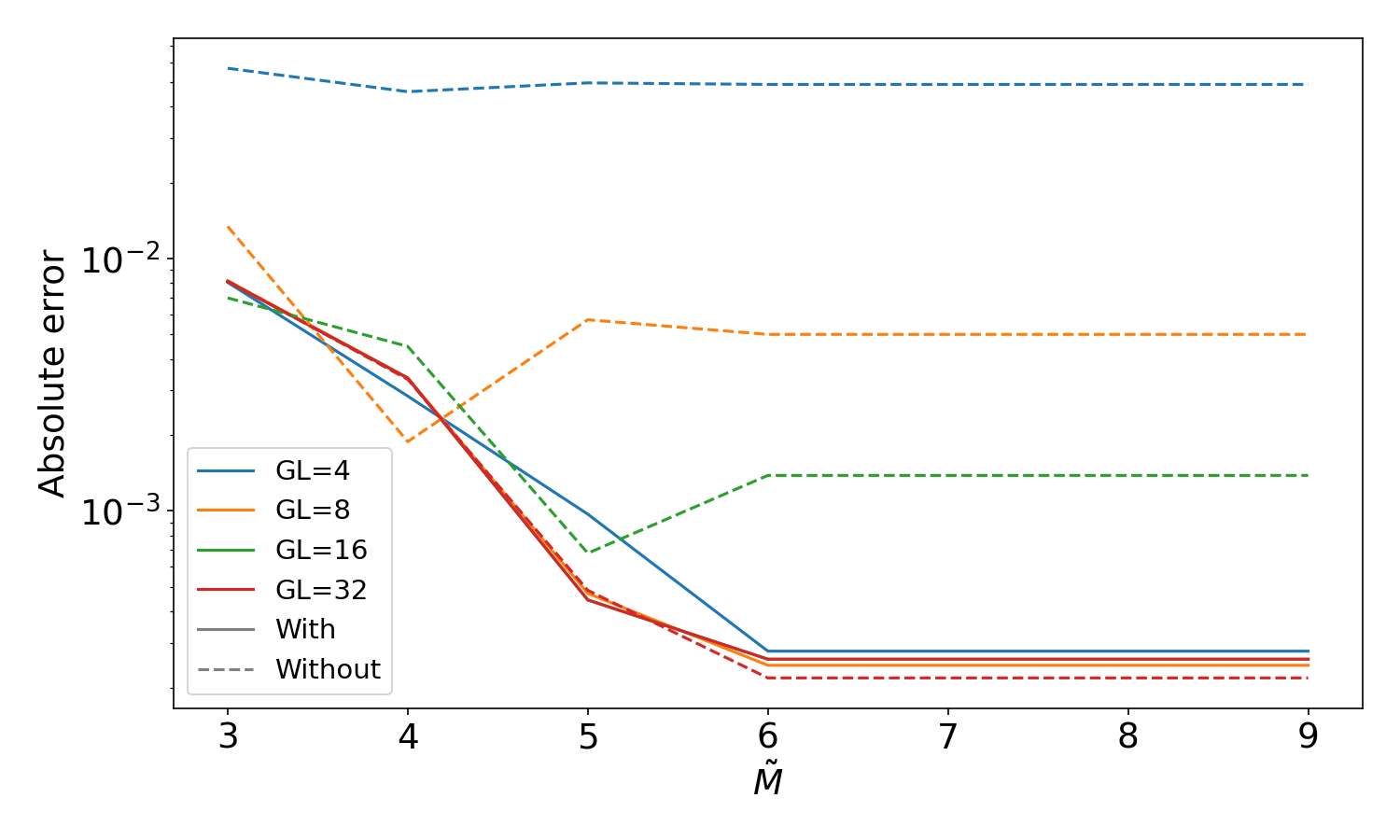}}}
        \quad
        \subfloat[\centering $M=5$.]{{\includegraphics[width=\twoplotswidth]{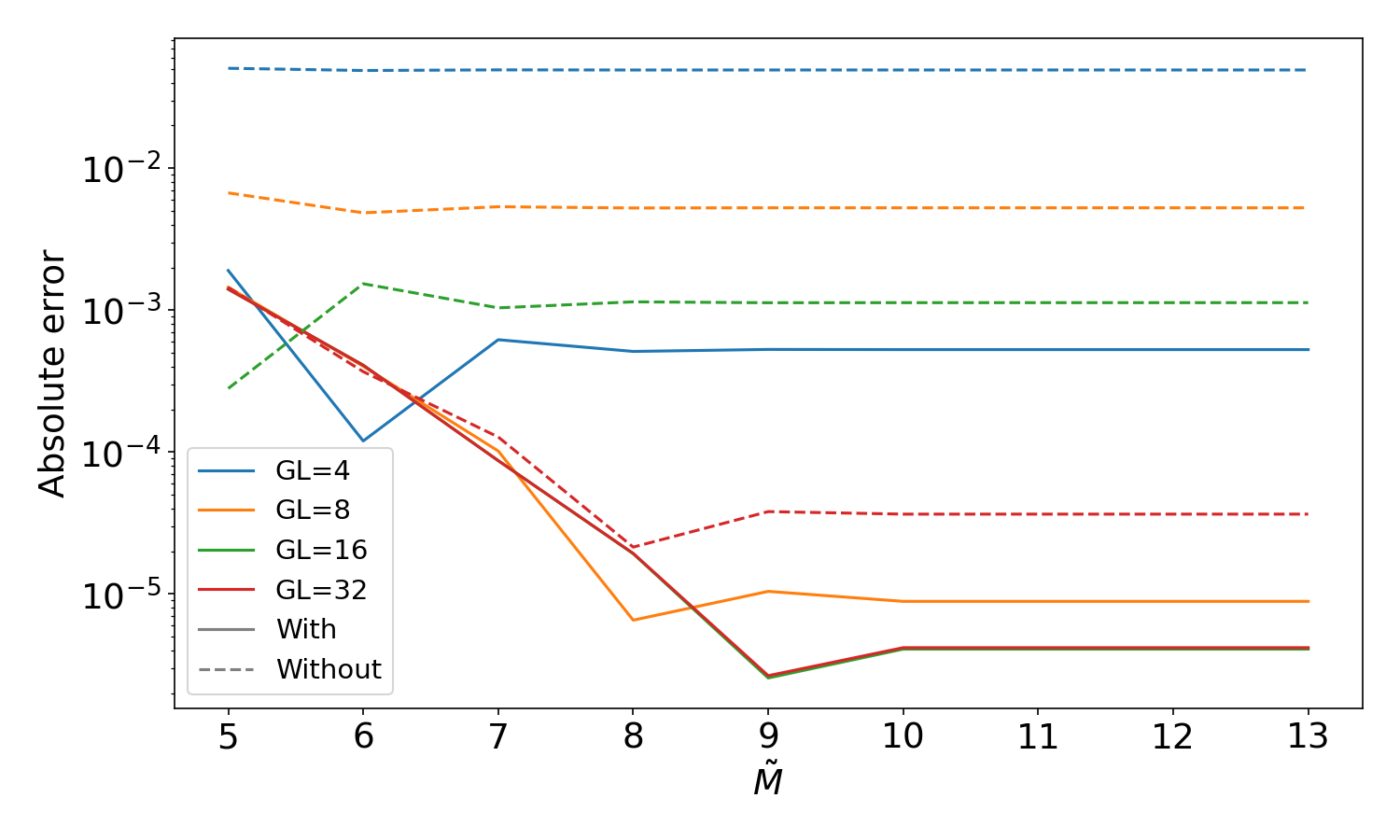}}}
        \caption{Error in terms of the truncations order $\tilde{M}$ of $\bpsi$ with (plain) \eqref{eq:lewis} and without (dashed) control-variate \eqref{eq:control-variate} for several degrees of Gauss-Laguerre quadrature, for a signature Ornstein-Uhlenbeck volatility model. $\kappa=2, \theta=0.25, \eta=0.6$ and $\rho=-0.7$, European at-the-money put options with maturity 6 months.}
        \label{fig:error_tilde_M}
    \end{figure}
    
    Regarding the numerical discretization and in order for the Riccati to converge in a realistic amount of time, we use Runge-Kutta to the 4th order to solve the ODE, which computes 4 times as many points as the Euler direct algorithm, but converges more than 4 times as fast. We also compute the characteristic function both JIT and in parallel so that it is drastically faster. Table~\ref{tab:compute_times_sig} shows the run-time, on a laptop CPU, of calling the characteristic function for several models. The left three columns show for reference the run-times of several well studied models where the characteristic function is explicitly known. The right five columns show the run-time while solving the (truncated) Riccati in \eqref{eq:Ric} to get the characteristic functional of the signature volatility model for several truncation orders, where \eqref{eq:Ric} has been solved using the Runge-Kutta algorithm of order 4 and 100 time-steps, i.e. the rhs of \eqref{eq:Ric} is called 400 times. The reader should keep in mind that the point of the signature volatility model is not to be faster than explicit models, but to be much more general and include non-affine Markovian and non-Markovian models, e.g. mean-reverting geometric Brownian motions and delayed process volatilities.
    
    \begin{table}[H]
        \centering
        \begin{tabular}{|c|c|c|c|c|c|c|c|c|}
            \hline
            \multicolumn{3}{|c|}{Explicit models} & \multicolumn{5}{|c|}{Signature volatility model} \\
            \hline
            Black-Scholes & Stein-Stein & Heston & $M=1$ & $M=2$ & $M=3$ & $M=4$ & $M=5$ \\
            \hline
            4.29e$-$7 & 5.29e$-$7 & 1.02e$-$7 & 1.47e$-$5 & 5.40e$-$5 & 3.48e$-$4 & 3.81e$-$3 & 4.23e$-$2 \\
            \hline
        \end{tabular}
        \caption{Run-time (seconds) of solving \eqref{eq:Ric} to get the characteristic function of the signature volatility model. Solved with Runge-Kutta of order 4 and 100 time-steps.}
        \label{tab:compute_times_sig}
    \end{table}

    In Figure~\ref{fig:sig-pricing}~(a)-(b), we compare our signature volatility pricing using the linear representations of the Ornstein-Uhlenbeck \eqref{eq:linear-OU} and Cox-Ingersoll-Ross \eqref{eq:linear-CIR} representations, truncated at order $M=4$, to the explicit Fourier pricing of the Stein-Stein \cite{stein-stein} and Heston \cite{heston} models. In Figure~\ref{fig:sig-pricing}~(c)-(d) we compare our signature volatility pricing using a linear mGBM and a linear LER drawn at random, see \eqref{param:LER}, to Monte Carlo simulations to demonstrate that our formulas are applicable beyond the classical affine case, where Fourier pricing is not typically employed, and beyond the Markovian setting.
    
    \begin{figure}[H]
        \centering
        \subfloat[\centering Stein-Stein model, $\kappa=1, \theta=0.25, \eta=1.2$ and $\rho=-0.5$.]{{\includegraphics[width=\twoplotswidth]{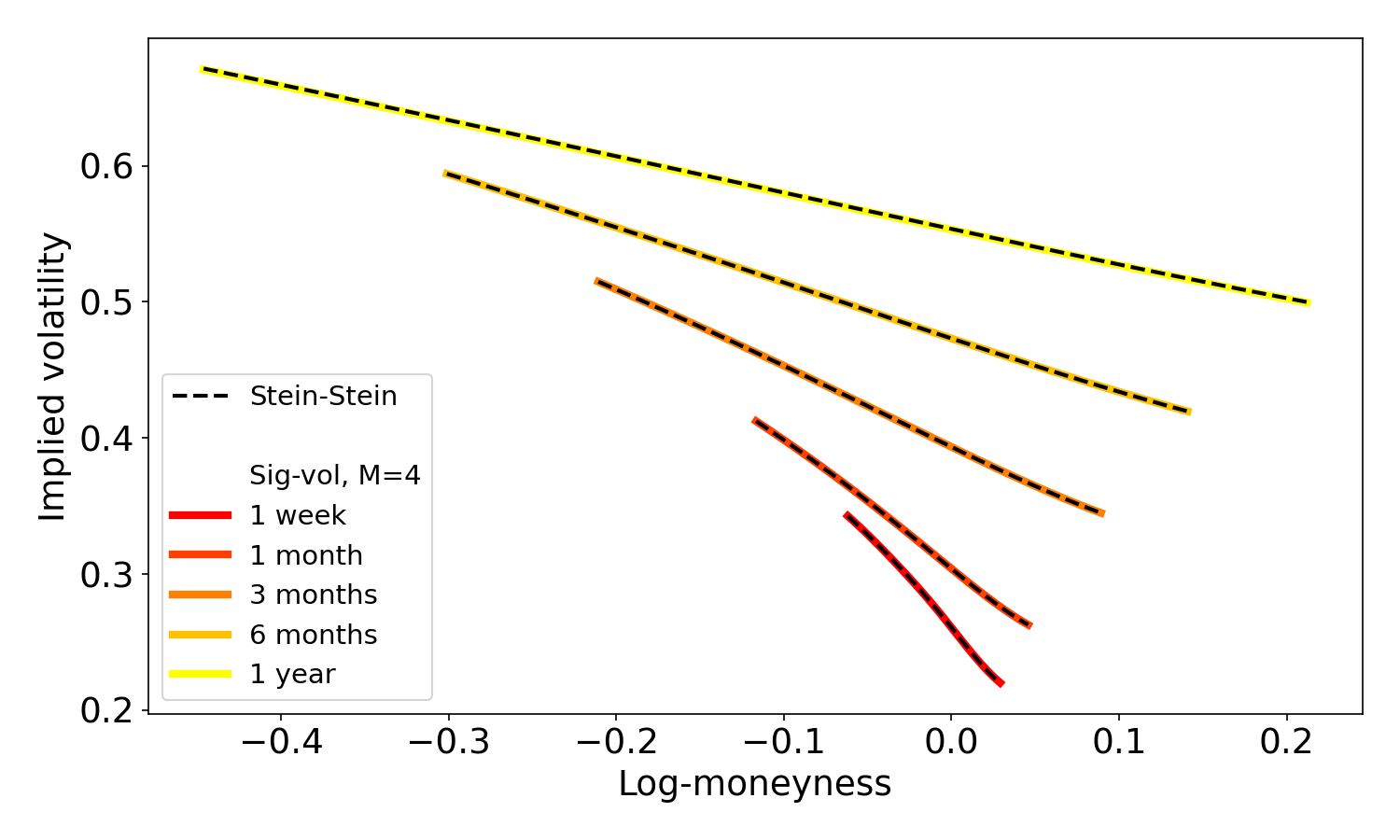}}}
        \quad
        \subfloat[\centering Heston model, $\kappa=2, \theta=0.0625, \eta=0.7$ and $\rho=-0.7$.]{{\includegraphics[width=\twoplotswidth]{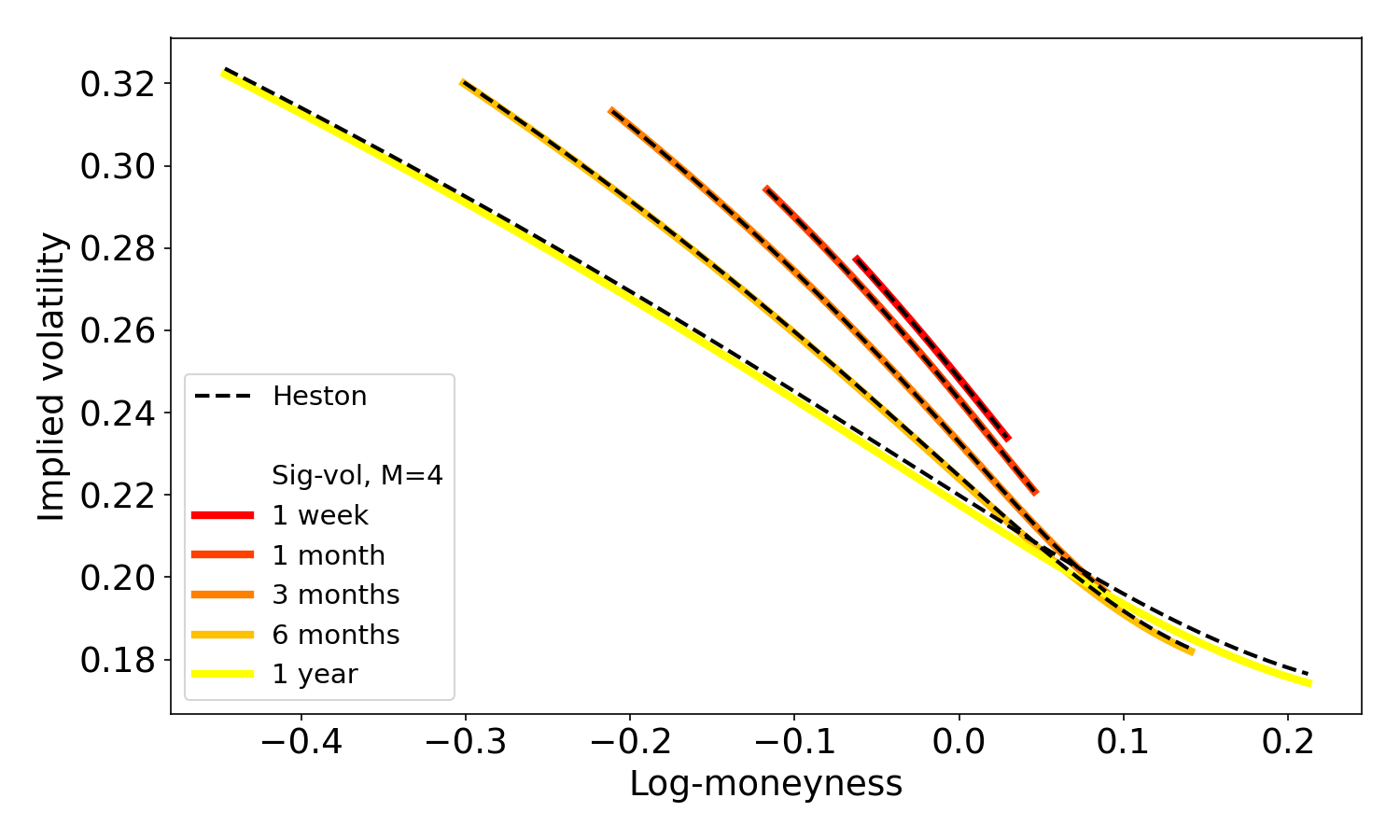}}}
        \quad
        \subfloat[\centering Hull-White model, $\kappa=1, \theta=0.25, \eta=0, \alpha=0.4$ and $\rho=-0.5711$.]{{\includegraphics[width=\twoplotswidth]{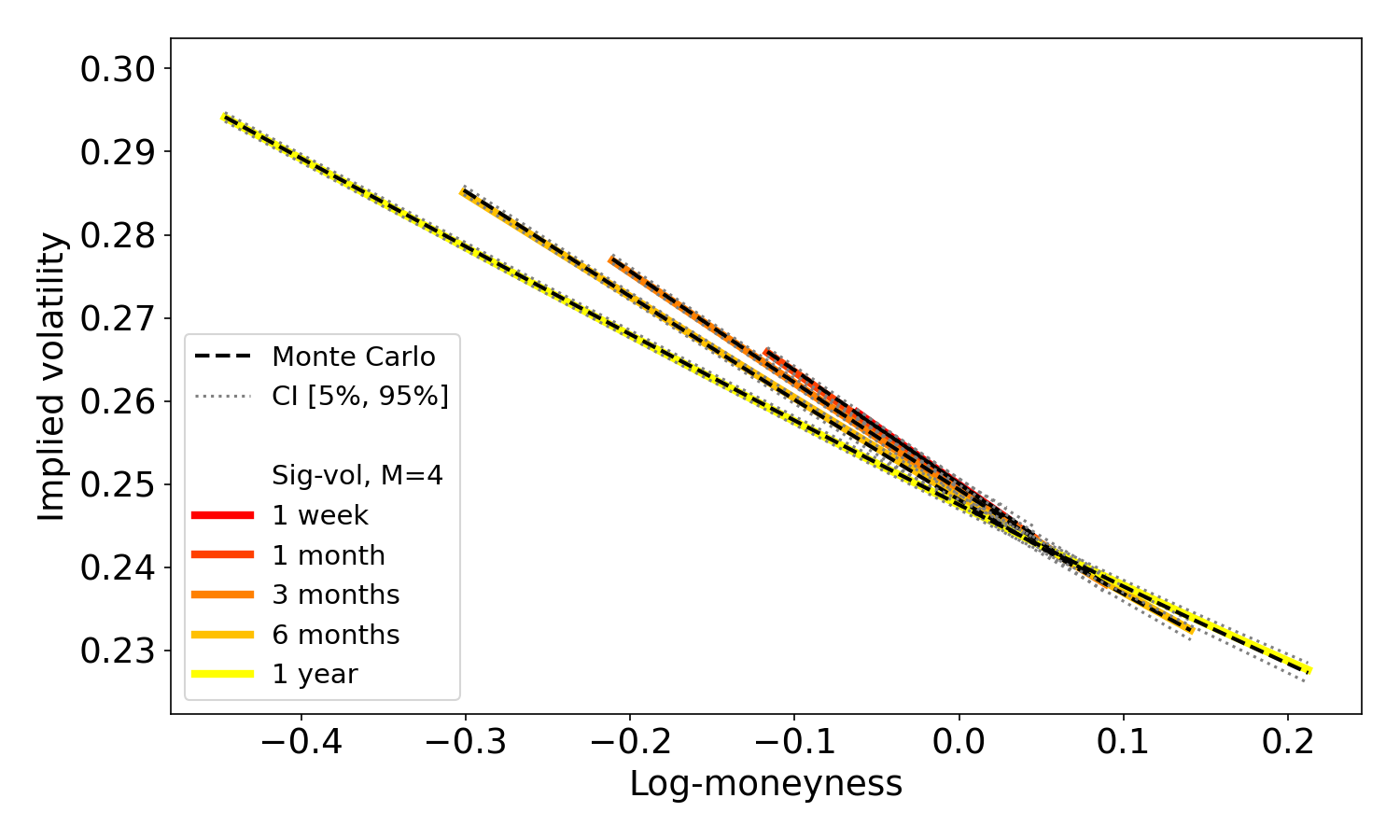}}}
        \quad
        \subfloat[\centering Coefficients drawn at random such that $\bracketsig{\bsigma^\textnormal{LER} \proj{2}} \geq 0$ and $\rho=-0.6$.]{{\includegraphics[width=\twoplotswidth]{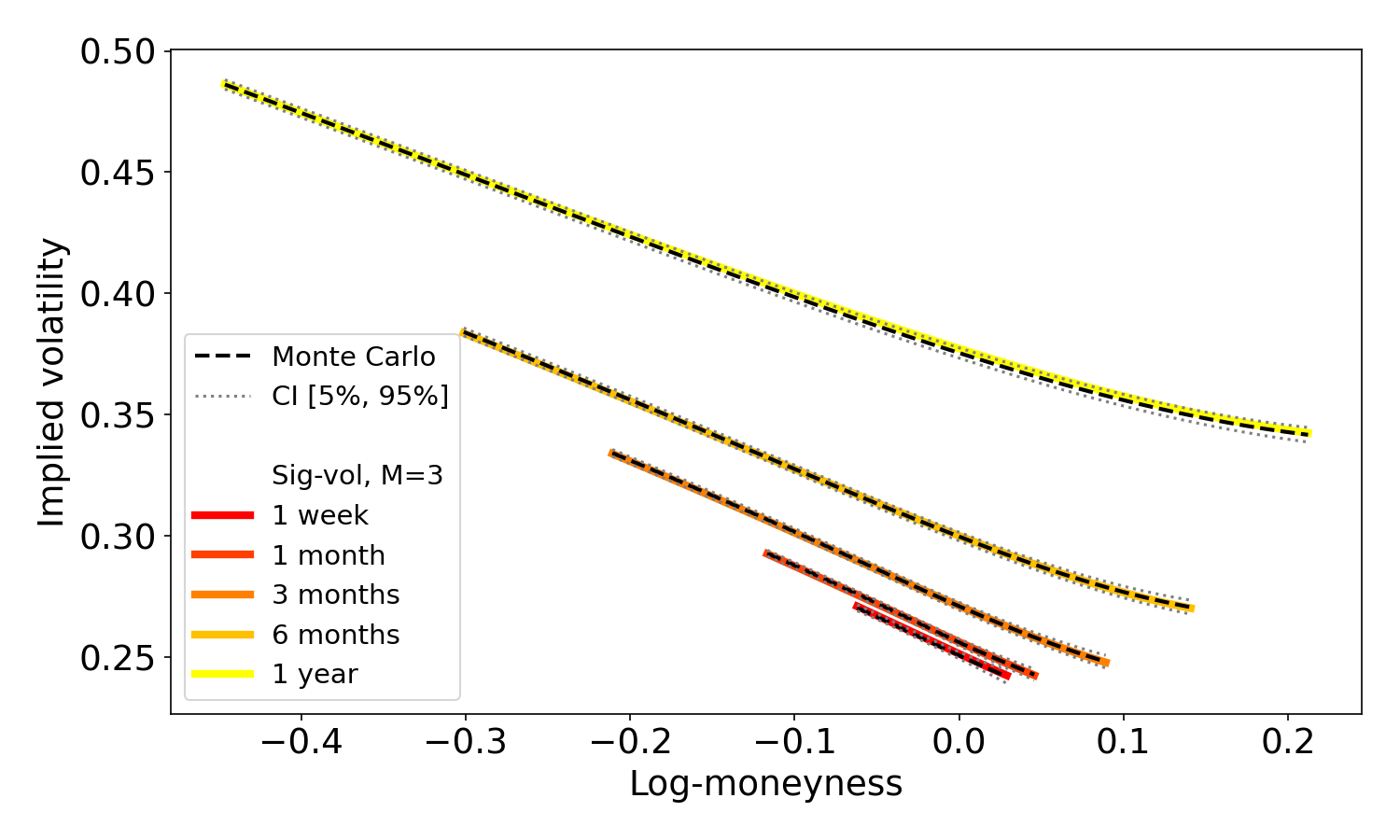}}}
        \caption{Implied volatility for several models of European put options with maturity 1 week (red), 1 month, 3 months, 6 months and 1 year (yellow).}
        \label{fig:sig-pricing}
    \end{figure}
    
    Lewis' approach together with Black-Scholes control variate, see \eqref{eq:control-variate}, was used. This serves as a numerical validation for Theorem~\ref{theo:charfun} as well as for our conjectured representation for the square-root process \eqref{eq:linear-CIR}. \\
    
    In Figure \ref{fig:sig-pricing}~(d) we work on a linear functional drawn at random under the leverage effect condition (LER), see Lemma~\ref{lem:leveff}. In this example, the coefficients $\bsigma^{\textnormal{LER}}$ are drawn such that $\bsigma^{\textnormal{LER} \emptyword} = x$, the coefficients that must remain positive have been drawn from a $\mathcal{U}_{[0, 0.5]}$, and the unconstrained coefficients have been drawn from a $\mathcal{U}_{[-0.5, 0.5]}$. Bellow is the draw used in Figure \ref{fig:sig-pricing}~(d):
    \begin{equation} \label{param:LER}
        \bsigma^{\textnormal{LER}} = \left( 0.25, 
        \begin{pmatrix}
            0.102763 \\
            0.274407
        \end{pmatrix},
        \begin{pmatrix}
             0.044883 & 0 \\
            -0.076345 & 0
        \end{pmatrix}, 
        \begin{pmatrix}
            0.145894 & 0 & \\
            0.391773 & 0 & \\
            & -0.062413 & 0 \\
            & 0.463663 & 0.357595
        \end{pmatrix},
        \bm{0} \right).
    \end{equation}

\subsection{\textit{q}-Volatility swaps}
    
    The payoff of a $q$-volatility swap is $R_T^q - K^q$ where $R_T^q = \left( \frac{1}{T} \bar{V}_T \right)^q$ is the realized $q$-volatility and $K^q$ the $q$-volatility strike. The aim in pricing $q$-volatility swaps is to find the fair strike price, i.e. $K^q = \E[R_T^q]$. This is made possible by Laplace inversion, see \citet[Theorem 1.2]{schurger},

    \begin{align} \label{eq:laplace_qvol}
        \E[X^q] = \frac{q}{\Gamma(1-q)} \int_0^\infty \frac{1 - \E \left[ e^{-uX} \right]}{u^{q+1}} \d u,
    \end{align}
    for $q \in (0, 1)$ and $X \geq 0$. It is straightforward to see that setting $f=0$ and $g(t)=-\frac{u}{T}$ in \eqref{eq:U} gives $\tilde{M}_0(u) = \E [e^{-\frac{u}{T} \bar{V}_T}]$, which allows us to compute analytically $q$-volatility swaps in the framework of the signature volatility models with
    $$ \E \left[ \left( \frac{1}{T} \bar{V}_T \right)^q \right] = \frac{q}{\Gamma(1-q)} \int_0^\infty \frac{1 - \tilde{M}_0(u)}{u^{q+1}} \d u. $$
    
    Specifically, the fair strike of the volatility swap, i.e. $q=\frac{1}{2}$, is thus of the form
    $$ \E \left[ \sqrt{\frac{1}{T} \bar{V}_T} \right] = \frac{1}{2 \sqrt{\pi}} \int_0^\infty \frac{1 - \tilde{M}_0(u)}{u^{3/2}} \d u. $$

    Moreover, the fair strike of the variance swap, i.e. $q=1$, can be written in closed form thanks to Fawcett's formula \eqref{eq:fawcet} for time-independent representations, i.e.
    $$ K_T^1 = \frac{1}{T} \bracketsigE[T]{\bsigma \shupow{2} \word{1}}. $$
    
    \begin{figure}[H]
        \centering
        \subfloat[\centering $q=\frac{1}{2}$ (volatility swap).]{{\includegraphics[width=\twoplotswidth]{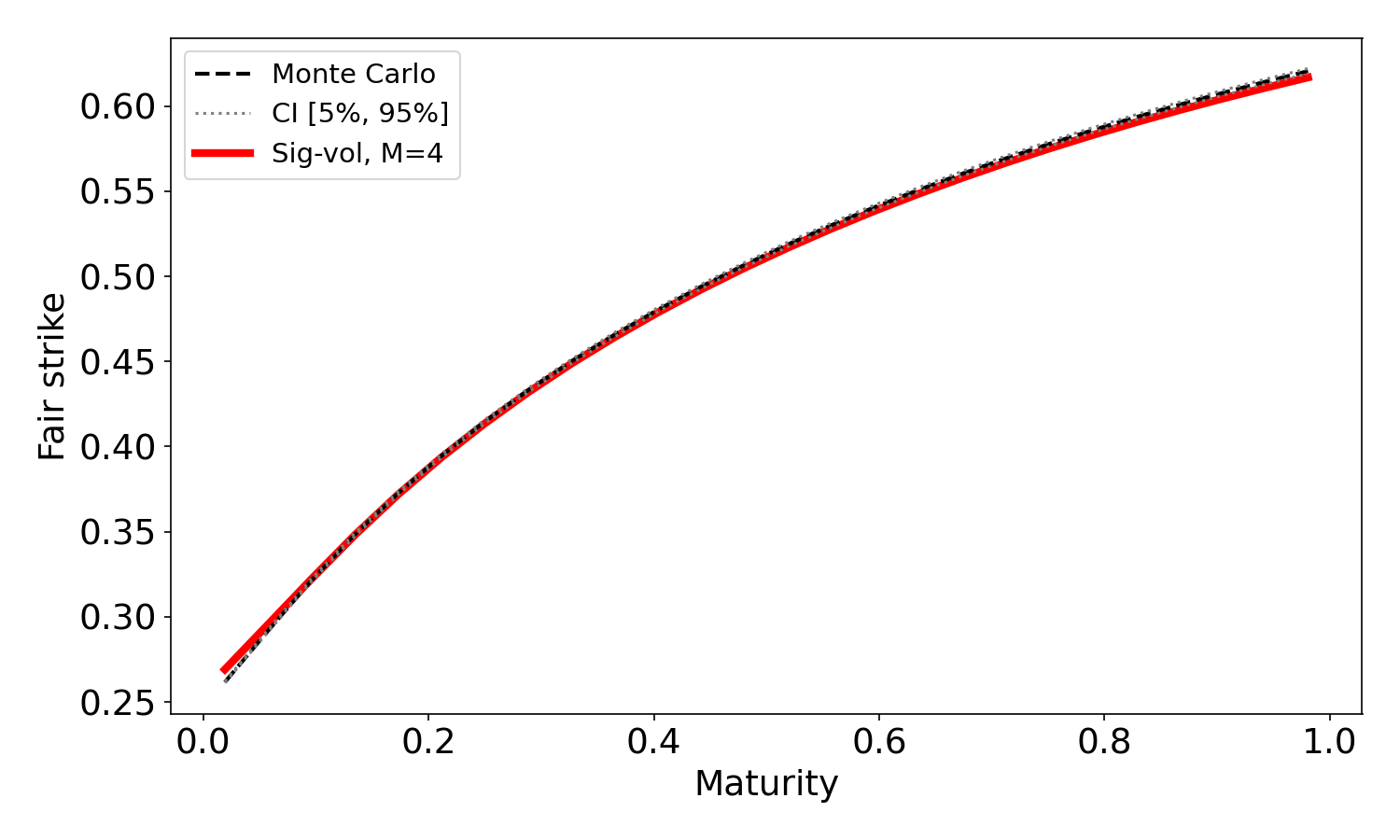}}}
        \quad
        \subfloat[\centering $q=1$ (variance swap).]{{\includegraphics[width=\twoplotswidth]{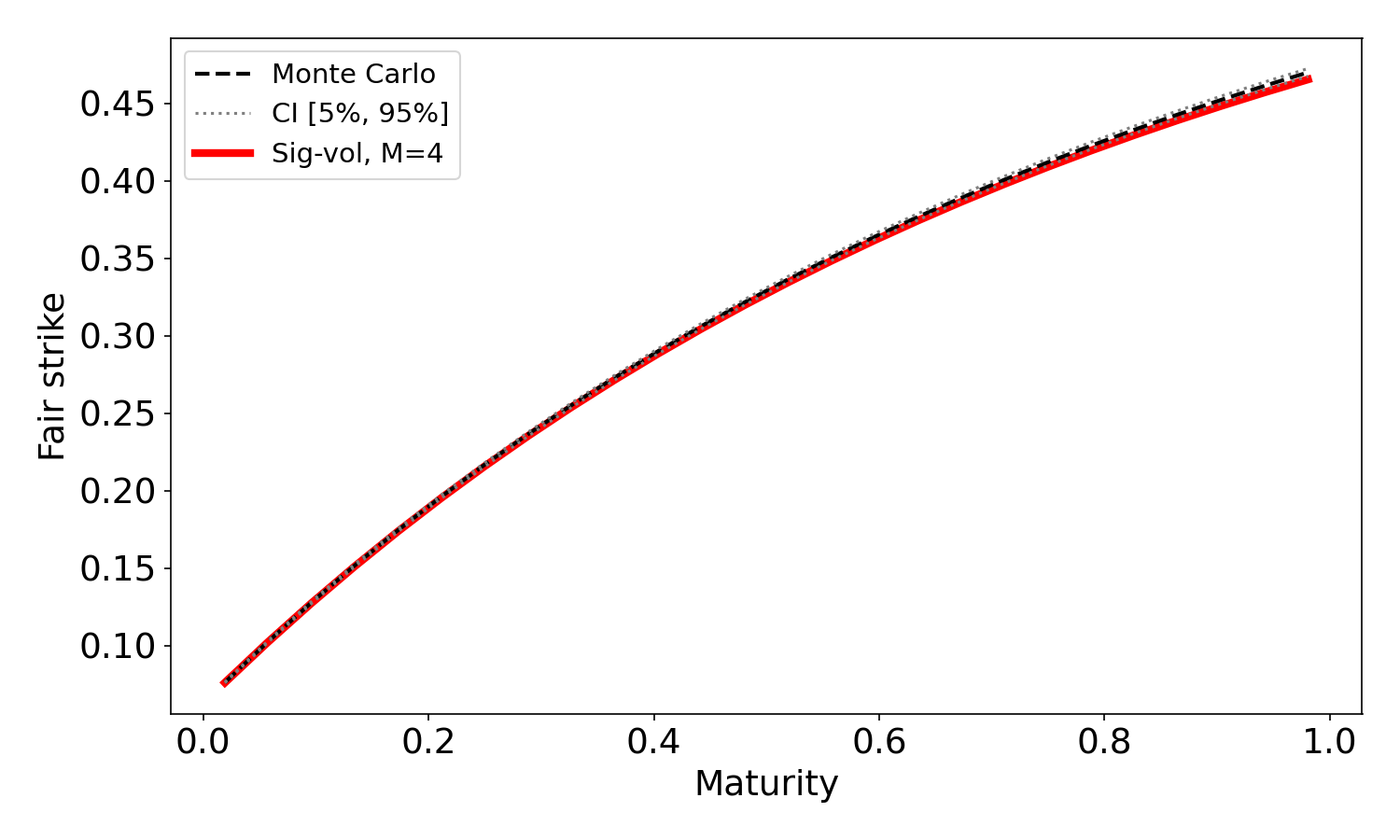}}}
        \caption{Fair strikes of $q$-volatility swaps as a function of the maturity, for a signature Ornstein-Uhlenbeck volatility model. $\kappa=1, \theta=0.25, \eta=1.2$ and $\rho=-0.7$.}
        \label{fig:vol-swaps}
    \end{figure}

\subsection{Market calibration} \label{sec:calib}

    Another way to learn the dynamics $\bsigma$ of the signature volatility model is to calibrate it against implied volatility surfaces. In this section we will focus on SPX implied volatility surface data from 2017-05-19, purchased from the CBOE website \url{https://datashop.cboe.com/}, which consists of $J$ maturities $T_j$ and $N_j$ strikes $K_{j, n}$, for each of those maturities. \\

    We calibrate one $\bsigma_j^\textnormal{SPX}$ against each maturity $T_j$, i.e. minimizing for each $j$
    
    $$ \mathcal{L}_j(\bsigma, \rho) = \frac{1}{N_j} \sum_{n=1}^{N_j} \left( \textnormal{IV}^\textnormal{market}(K_{j, n}, T_j) - \textnormal{IV}^\textnormal{sig-vol}(K_{j, n}, T_j, \bsigma, \rho) \right)^2, $$
    
    where $\textnormal{IV}^\textnormal{sig-vol}(K, T, \bsigma, \rho)$ is the implied volatility of the price~\eqref{eq:control-variate} computed after solving the Riccati~\eqref{eq:Ric}. Moreover, the constant term $(\bsigma_j^\textnormal{SPX})^\emptyword$ was fixed at $0.1204$ the closing value of the VIX on 2017-05-19 for all $j$. A differential evolution algorithm, see \cite{diff-evol}, from the \texttt{SciPy} library was used to minimize the MSE as it produced the best fit on synthetic implied surfaces. \\
    
    Figure~\ref{fig:calib-spx} shows the implied volatility of such fitted $(\bsigma_j^\textnormal{SPX}, \rho_j^\textnormal{SPX})$ for maturities $T_j \in \{7, 14, 35, 56\}$ in number of days.
    We recall that the number of non-zero terms in an object of $\tTA{M}$ is $\frac{d^{M+1}-1}{d-1}$, which is 15 in the case of $d=2$ and $M=3$.\\
  
    Finally, we stress that the primary purpose of this exercise is to illustrate the stability of our numerical scheme, across several parameters, using Fourier techniques for calibration with realistic implied volatilities, rather than to justify the choice of the volatility model itself. No specific structure is imposed on the coefficients, showing that the method is effective even beyond the affine and Markovian cases.
    
    \begin{figure}[H]
        \centering
        \subfloat[\centering Maturity 7 days.]{{\includegraphics[width=\twoplotswidth]{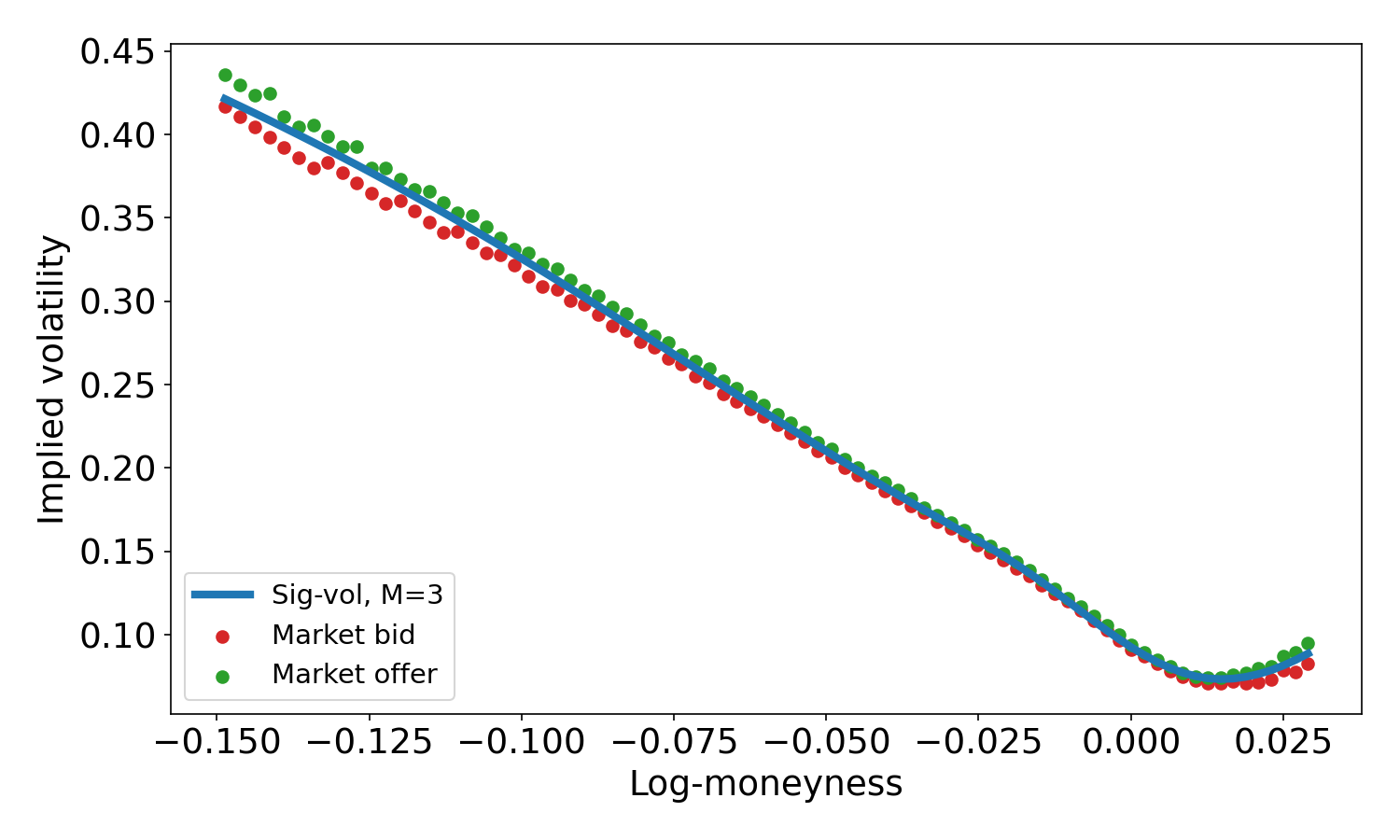}}}
        \quad
        \subfloat[\centering Maturity 14 days.]{{\includegraphics[width=\twoplotswidth]{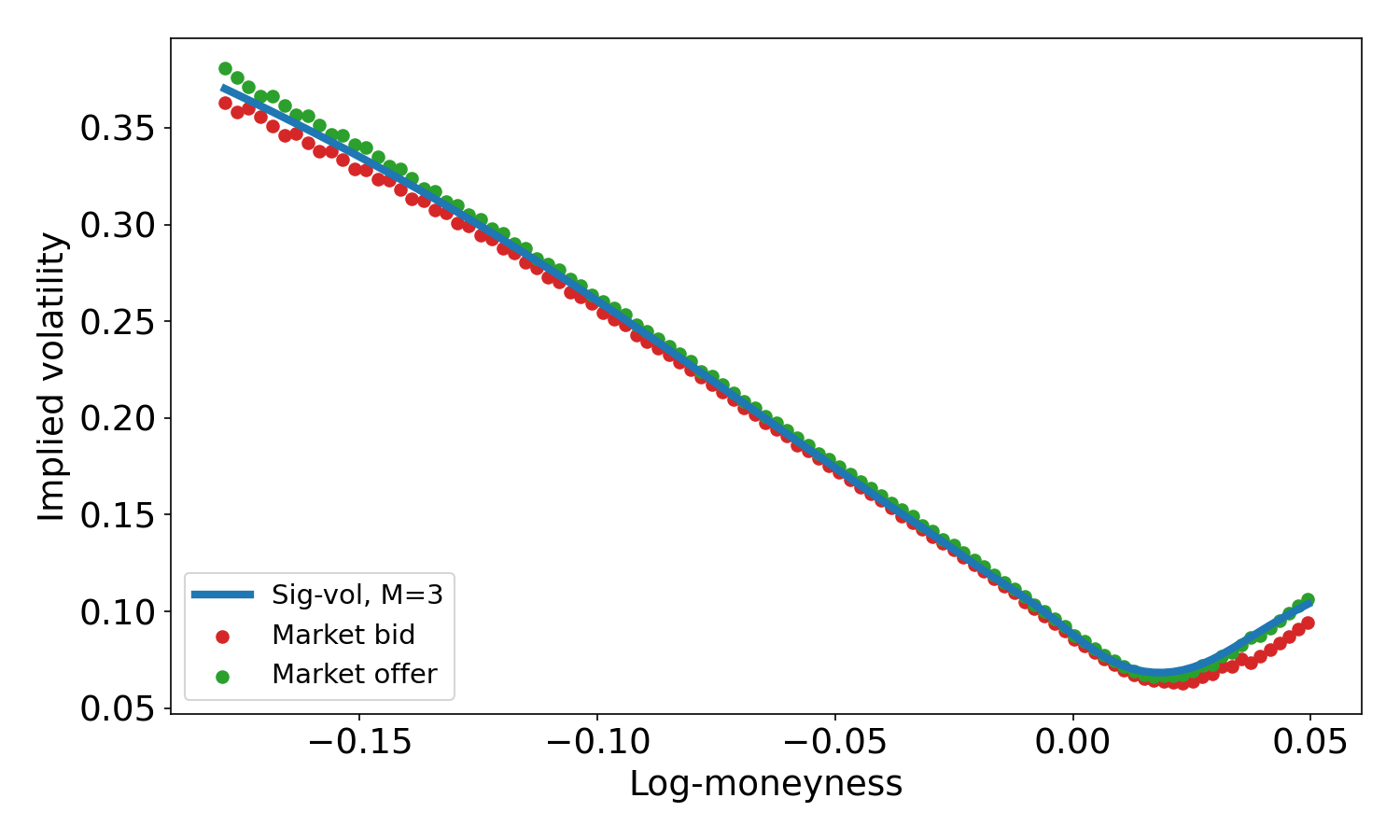}}}
        \\
        \subfloat[\centering Maturity 35 days.]{{\includegraphics[width=\twoplotswidth]{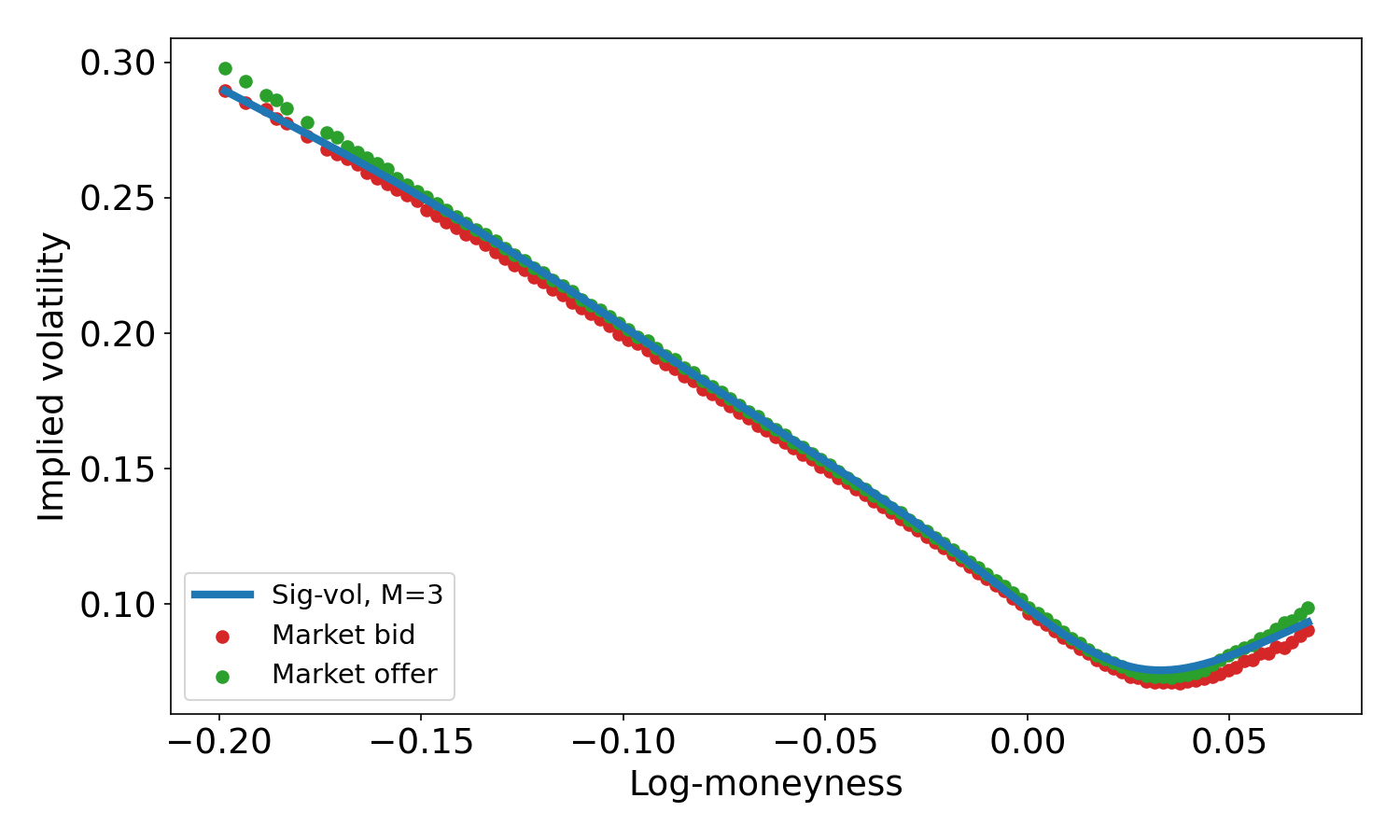}}}
        \quad
        \subfloat[\centering Maturity 56 days.]{{\includegraphics[width=\twoplotswidth]{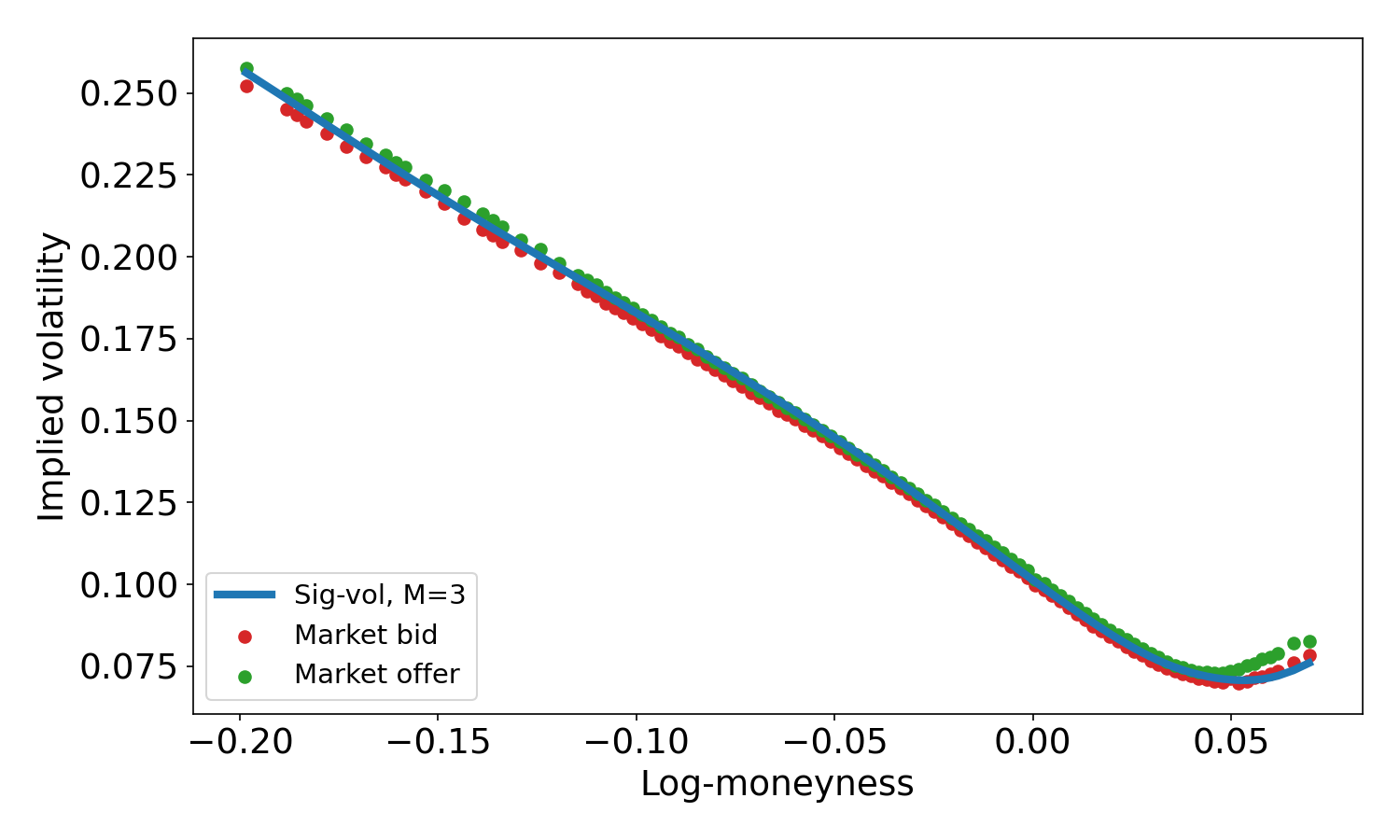}}}
        \caption{Implied volatility of signature volatility model calibrated against SPX data 2017-05-19 on several maturities.}
        \label{fig:calib-spx}
    \end{figure}

\section{Quadratic hedging by Fourier methods} \label{sec:hedging}

    The signature volatility model \eqref{eq:sigmodel1}-\eqref{eq:sigmodel2} generates an incomplete market in general (unless the correlation $\rho$ between the two Brownian motions $B$ and $W$ is $\pm 1$). Therefore contingent claims on the stock $S$ cannot be perfectly hedged. We will consider quadratic hedging methods instead, we refer to \citet{martinquadhedging} for a detailed overview of quadratic hedging approaches. We will show that in our setup, quadratic hedging remains highly tractable in signature volatility models using Fourier techniques on the conditional characteristic function \eqref{eq:charfun}. We do this in two steps:
    \begin{enumerate}
        \item We first solve in Section~\ref{S:hedging1} the quadratic hedging problem for a generic contingent claim $\xi$, in a general stochastic volatility model driven by the two dimensional Brownian motion $(W, W^\perp)$. The hedging strategy depends on the quantities that appear in the martingale representation theorem of $\E[\xi|\F_t]$. We provide a concise proof in this setting.
        
        \item We then show in Section~\ref{subsec:hedging-practical} that in the setting of a signature volatility model and for contingent claims that admit Fourier representation, such as European and Asian call and put options, the hedging strategy can be recovered from the conditional characteristic function \eqref{eq:charfun}. 
    \end{enumerate}

\subsection{A generic solution} \label{S:hedging1}

    Let $\xi$ be an $\F_T$-measurable non-negative random variable such that $\E[\xi^2] < \infty$ that we are looking to hedge using a self-financing portfolio. We recall that $(\F_t)_{t \geq 0}$ is the filtration generated by $(W, W^\perp)$. A self-financing hedging portfolio $X$ consists of an initial wealth $X_0 \in \R$ and a progressively measurable strategy $(\vartheta_u)_{u \leq T}$ of the amount of shares invested in asset $S$ given in \eqref{eq:sigmodel1} at time $u \leq T$. It has the following dynamics 
    \begin{align}
        X_t^\alpha = X_0 + \int_0^t \vartheta_u \d S_u = X_0 + \int_0^t \alpha_u \d B_u,
    \end{align}
    with $\alpha_u := \vartheta_u S_u \Sigma_u$.
    The set of admissible hedging strategies $\alpha$ is defined by
    \begin{align}
        \mathcal{H} = \left \{ \alpha \text{ progressively measurable such that } \int_0^T \E \left[ \alpha_s^2 \right] \d s < \infty \right\}.
    \end{align}
    We stress that in this section we do not impose specific dynamics for the stochastic volatility $\Sigma$, i.e.~\eqref{eq:sigmodel2}, $\Sigma$ is only assumed to be adapted to the Brownian motion $(W, W^\perp)$. 

    A quadratic hedging strategy aims at minimizing the following objective function
    \begin{align} \label{eq:Jquad}
        J(X_0, \alpha) = \E \left[ \left( X_T^\alpha - \xi \right)^2 \right] 
    \end{align}
    over $X_0 \in \R$ and $\alpha \in \mathcal{H}$. \\ 
    
    Following \citet{hedging_foellmer} it is easy to derive a solution of the quadratic hedging problem using the martingale representation theorem. 
    Note that $(\E[\xi | \F_t])_{t \leq T}$ is a square integrable martingale with terminal value $\xi$ at $T$. An application of the martingale representation theorem \cite[Theorem 4.15]{martingalerep} yields the existence of two progressively measurable and square integrable processes $Z$ and $Z^{\perp}$ such that 
    \begin{align} \label{eq:martingalerep}
        \E \left[\xi | \F_t \right ] = \xi - \int_t^T Z_s \d W_s - \int_t^T Z^{\perp}_s \d W_s^{\perp}.
    \end{align}
    
    The value of the quadratic hedging problem is consequently given by 
    \begin{align} \label{eq:hedgeopti}
        \inf_{X_0 \in \R, \alpha \in \mathcal H} J(X_0, \alpha) = \E \left[ \int_0^T \left(Z_t^2 + (Z_t^\perp)^2 \right) \d t - \int_0^T \left( \rho Z_t + \sqrt{1 - \rho^2} Z_t^\perp \right)^2 \d t \right],
    \end{align}
    where the optimum is attained for $(X_0^*, \alpha^*)$ given by
    \begin{align} \label{eq:hedgingoptimal}
        X_0^* = \E \left[ \xi \right] \quad \text{and} \quad \alpha^*_t = \rho Z_t + \sqrt{1 - \rho^2} Z_t^\perp, \quad t \leq T.
    \end{align}

\subsection{Fourier implementation in the signature volatility model} \label{subsec:hedging-practical}

    We now illustrate how the optimal hedging strategy $X_0^*$ and $\alpha^*$ given in \eqref{eq:hedgingoptimal} can be recovered numerically from the knowledge of the conditional characteristic functional \eqref{eq:charfun} in the specific case of a signature volatility model, i.e.~when $\Sigma$ is of the form \eqref{eq:sigmodel2} and for contingent claims that admit a Fourier representation. In this section, we only consider European options, Asian options have been included in Appendix~\ref{subsec:hedging-asiat} for the interested reader. For more general payoffs, a similar approach can be adapted, following the representation formulas in \citet*{di2019semi}. \\
    
    In order to implement the quadratic hedging, the idea is to re-express the Fourier inversion formula \eqref{eq:control-variate} in terms of the process $M_t(u)$ in \eqref{eq:M} with $f(t) = iu$ and $g = 0$ and apply Itô. Since in this case
    $$ M_t(u) = \phi_t(u) e^{iu \log S_t}, $$
    the representation \eqref{eq:control-variate} directly leads to
    \begin{equation} \label{eq:control-variate-M}
        C_t(S_t; T, K) = C_t^{\textnormal{BS}}(S_t; T, K) - \frac{K}{\pi} \int_0^\infty \Re \left[ e^{-i (u - \frac{i}{2}) \log K} \left( M_t \left( u - \tfrac{i}{2} \right) - M_t^{\textnormal{BS}} \left( u - \tfrac{i}{2} \right) \right) \right] \frac{\d u}{\left( u^2 + \tfrac{1}{4} \right)},
    \end{equation}
    where
    $$ M_t^{\textnormal{BS}}(u) = \phi_{t}^{\textnormal{BS}}(u) e^{iu \log S_t}, $$
    and $\phi_{t}^{\textnormal{BS}}$ as defined in \eqref{cf-black-scholes}.
   
    Set $w(u) := \frac{K}{\pi} \frac{e^{-i (u - \frac{i}{2}) \log K}}{u^2 + \frac{1}{4}}$ and $\tilde{u} := (u - \frac{i}{2})$ and assume that $\sup_{u: \Im(u)=-1/2}\mathbb E[|M_t(u)|^2] <\infty$, $t\leq T$, then an application of Fubini's theorem combined with Itô's Lemma, see for instance \cite[Proposition 4.1]{di2019semi}, yields
    \begin{align*}
        \d C_t &
        = \d C_t^{\textnormal{BS}} + \int_0^\infty \Re \left[ w(u) \left( \d M_t(\tilde{u}) - \d M_t^{\textnormal{BS}}(\tilde{u}) \right) \right] \d u
        \\ &
        = \Delta_t^{\textnormal{BS}} \d S_t + \left( \Theta_t^{\textnormal{BS}} + \tfrac{1}{2} \Gamma_t^{\textnormal{BS}} (S_t \Sigma_t)^2 \right) \d t
        \\ & \qquad + \int_0^\infty \Re \Big[ w(u) \Big( M_t(\tilde{u}) \left( \d U_t + \tfrac{1}{2} \d [U]_t \right) - M_t^{\textnormal{BS}}(\tilde{u}) \left( \d \log M_t^{\textnormal{BS}}(\tilde{u}) + \tfrac{1}{2} \d [\log M^{\textnormal{BS}}(\tilde{u})]_t \right) \Big) \Big] \d u,
    \end{align*}
    where $\Delta_t^{\textnormal{BS}} = \frac{\partial}{\partial S_t} C_t^{\textnormal{BS}}$, $\Theta_t^{\textnormal{BS}} = \frac{\partial}{\partial t} C_t^{\textnormal{BS}}$ and $\Gamma_t^{\textnormal{BS}} = \frac{\partial^2}{\partial S_t^2} C_t^{\textnormal{BS}}$. Furthermore, using equalities between Black-Scholes \textit{Greeks}, one can show that 
    \begin{align*}
        \d C_t &
        = \Delta_t^{\textnormal{BS}} S_t \Sigma_t \d B_t + \frac{\Theta_t^{\textnormal{BS}}}{\sigma_{\textnormal{BS}}^2} \left(\sigma_{\textnormal{BS}}^2 - \Sigma_t^2 \right) \d t
        \\ & \qquad + \int_0^\infty \Re \Big[ w(u) \Big( M_t(\tilde{u}) \left( \bracketsig{\bpsi_t(\tilde{u}) \proj{2}} \d W_t + i \tilde{u} \Sigma_t \d B_t \right)
        \\ & \qquad \qquad \qquad \qquad \qquad - M_t^{\textnormal{BS}}(\tilde{u}) \left[ \tfrac{1}{2} \left( \tilde{u}^2 + i \tilde{u} \right) \left( \sigma_{\textnormal{BS}}^2 - \Sigma_t^2 \right) \d t + i \tilde{u} \Sigma_t \d B_t \right] \Big) \Big] \d u
        \\ &
        = \Delta_t^{\textnormal{BS}} S_t \Sigma_t \d B_t + \int_0^\infty \Re \Big[ w(u) \Big( M_t(\tilde{u}) \left( \bracketsig{\bpsi_t(\tilde{u}) \proj{2}} \d W_t + i \tilde{u} \Sigma_t \d B_t \right) - M_t^{\textnormal{BS}}(\tilde{u}) i \tilde{u} \Sigma_t \d B_t \Big) \Big] \d u.
    \end{align*}
    The final equality comes from the fact that $\Theta_t^{\textnormal{BS}}$ can also be written as a Fourier integral. This gives option price dynamics of the following form
    \begin{align} \label{eq:valuehedging}
        \d C_t &
        = Z_t \d W_t + Z_t^\perp \d W_t^\perp,
    \end{align}
    where $Z$ and $Z^\perp$ are defined as follows
    \begin{align*}
        Z_t &
        = \Sigma_t S_t \Delta_t^{\textnormal{BS}} \rho + \int_0^\infty \Re \left[ \zeta_t(\tilde{u}) w(u) \right] \d u, \qquad Z_t^\perp 
        = \Sigma_t S_t \Delta_t^{\textnormal{BS}} \sqrt{1 - \rho^2} + \int_0^\infty \Re \left[ \zeta_t^{\perp}(\tilde{u}) w(u) \right] \d u,
    \end{align*}
    with
    \begin{align*}
        \zeta_t(u) :&
        = iu \Sigma_t \left( M_t(u) - M_t^{\textnormal{BS}}(u) \right) \rho + M_t(u) \bracketsig{\bpsi_t(u) \proj{2}}
        \\ \zeta_t^\perp(u) :&
        = iu \Sigma_t \left( M_t(u) - M_t^{\textnormal{BS}}(u) \right) \sqrt{1 - \rho^2}.
    \end{align*}

    This allows us to solve the quadratic hedging problem in \eqref{eq:hedgingoptimal} numerically for European call options in the framework of signature volatility models. Moreover, applying the put-call parity allows us to easily extend it to European put options. \\

    In Figure \ref{fig:hedging-euro-OU}, we simulate price trajectories under the Stein-Stein model \cite{stein-stein} and compare the performance of the explicit hedging strategy to the Fourier hedging of the signature Ornstein-Uhlenbeck volatility model for a European put option with multiple strikes and two horizons.
    
    \begin{figure}[H]
        \centering
        \subfloat[\centering $T=$ 1 month, $K=0.9$.]{{\includegraphics[width=\threeplotswidth]{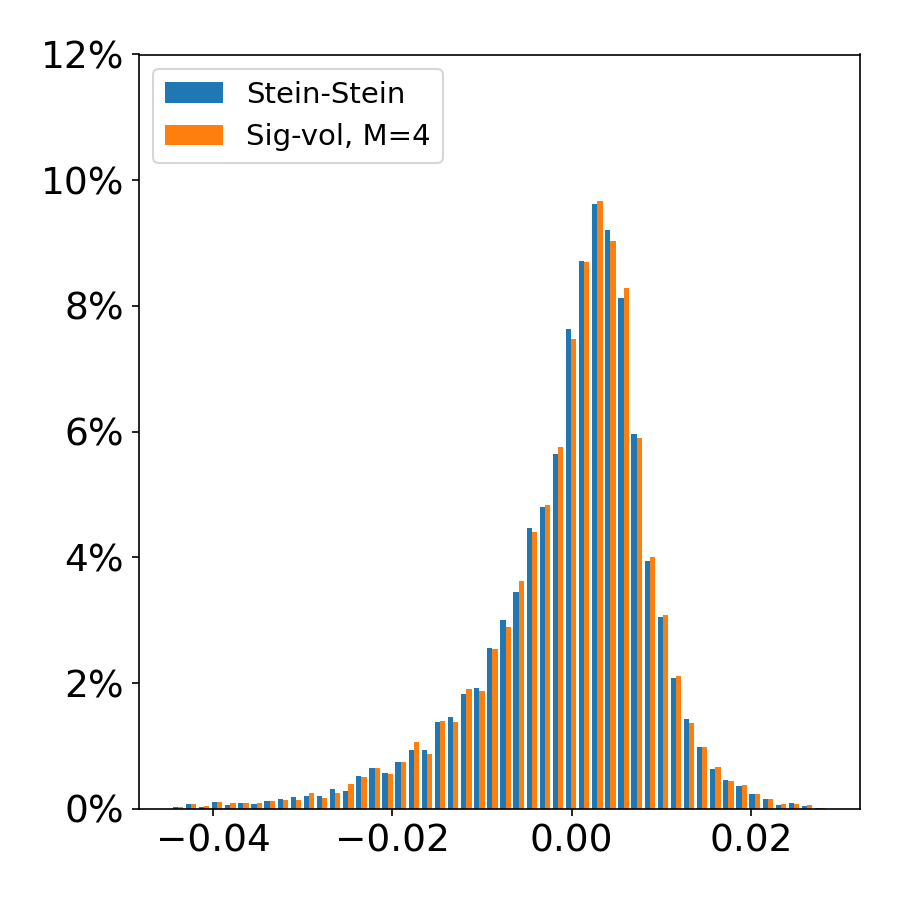}}}
        \quad
        \subfloat[\centering $T=$ 1 month, $K=1.0$.]{{\includegraphics[width=\threeplotswidth]{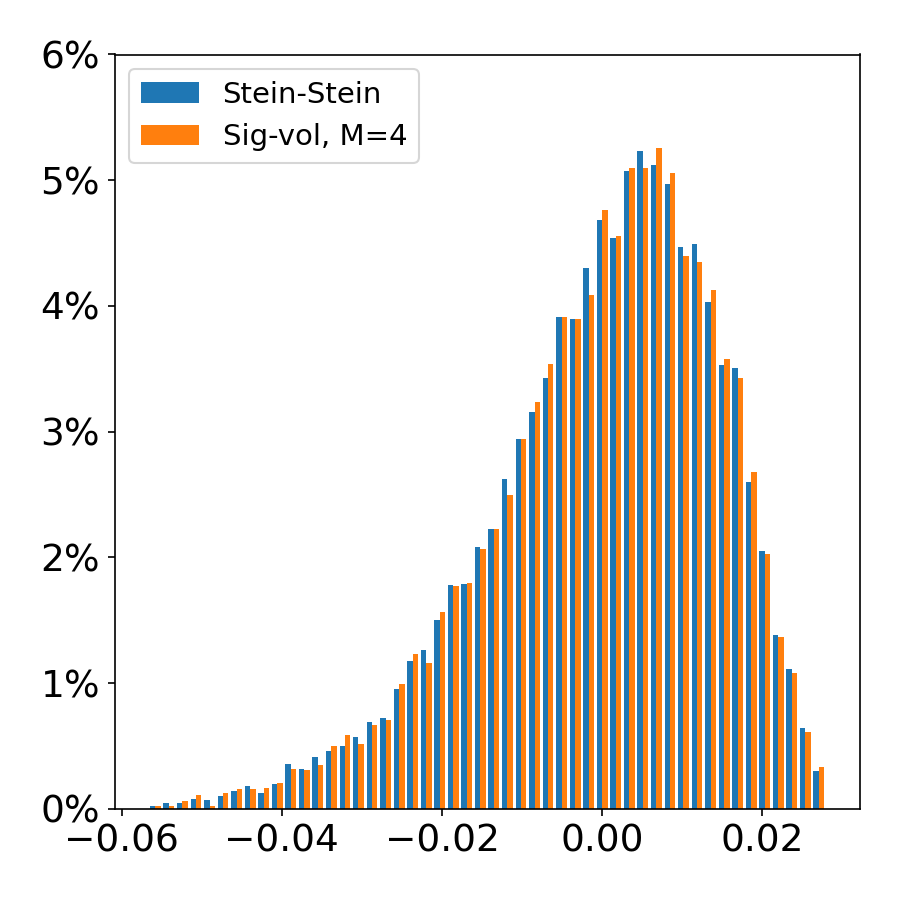}}}
        \quad
        \subfloat[\centering $T=$ 1 month, $K=1.1$.]{{\includegraphics[width=\threeplotswidth]{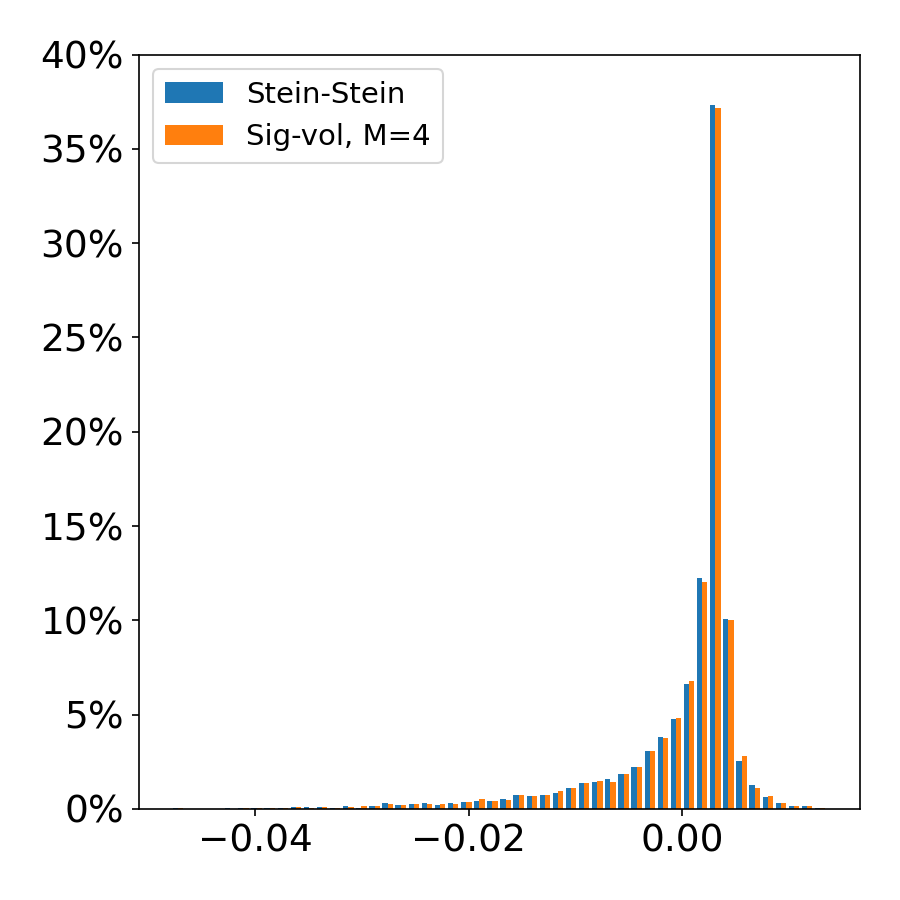}}}
        \\
        \subfloat[\centering $T=$ 6 months, $K=0.75$.]{{\includegraphics[width=\threeplotswidth]{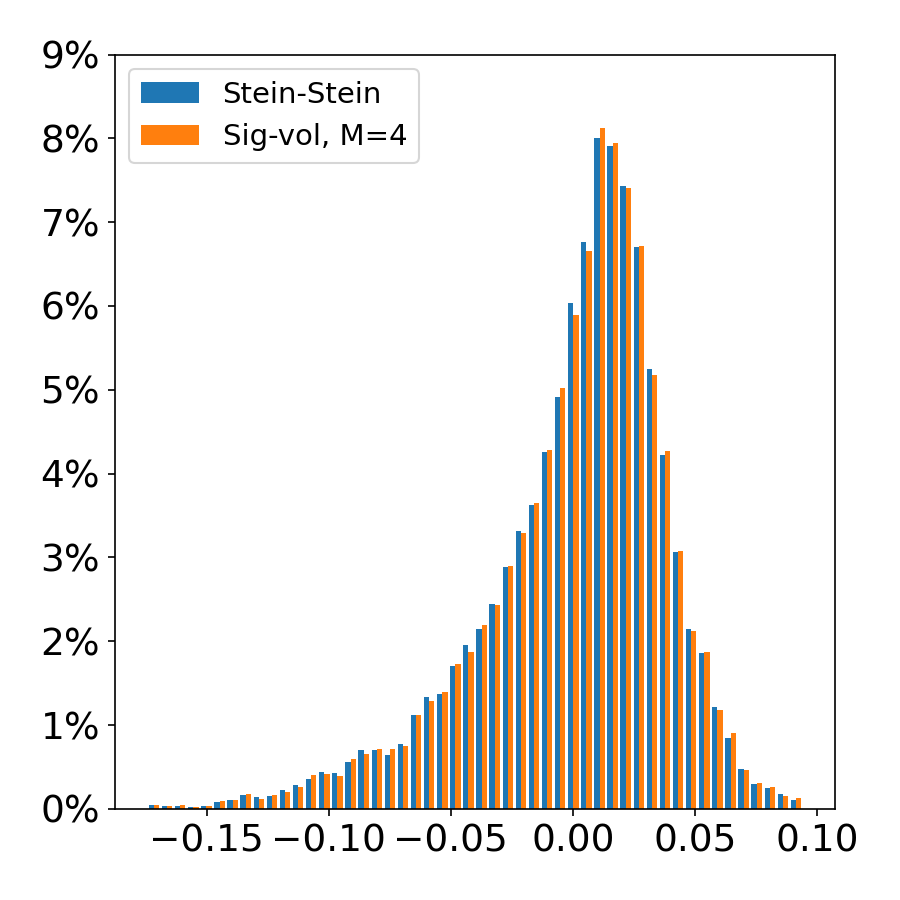}}}
        \quad
        \subfloat[\centering $T=$ 6 months, $K=1.0$.]{{\includegraphics[width=\threeplotswidth]{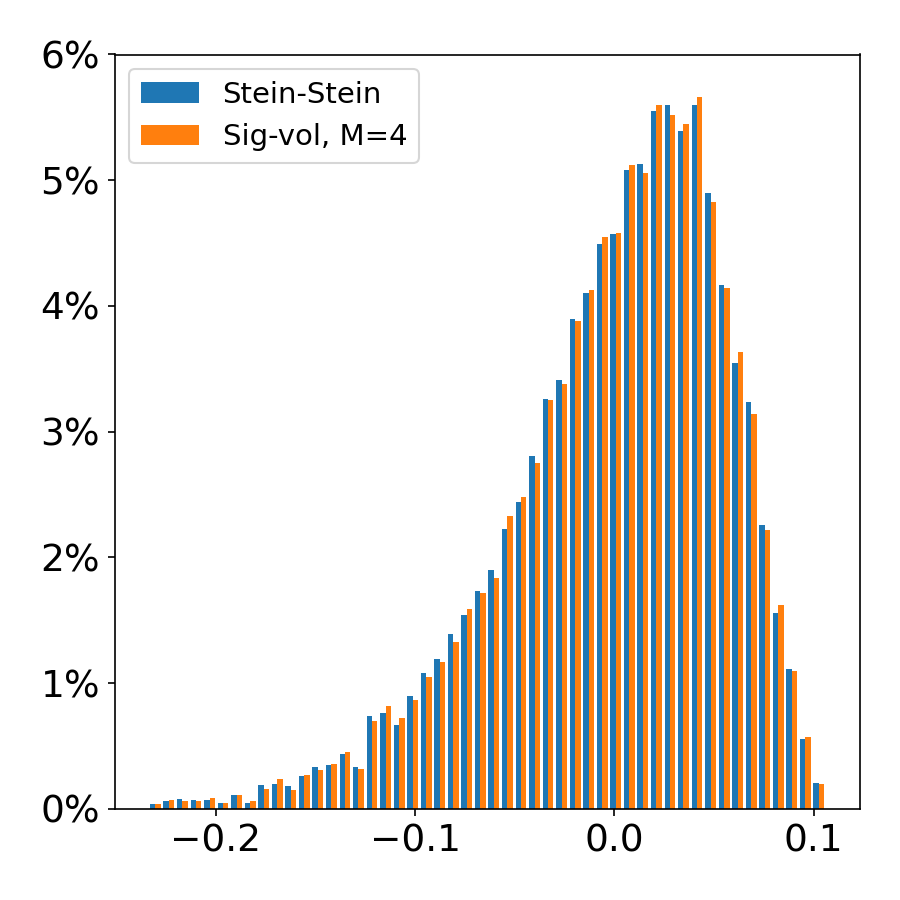}}}
        \quad
        \subfloat[\centering $T=$ 6 months, $K=1.3$.]{{\includegraphics[width=\threeplotswidth]{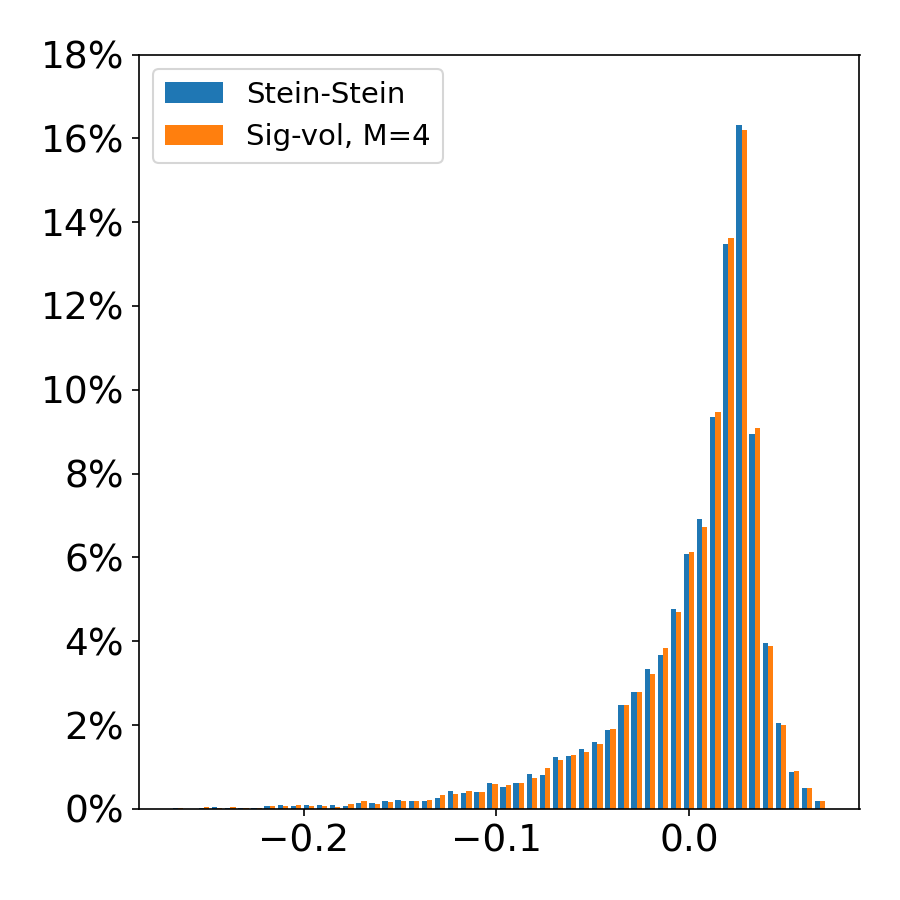}}}
        \caption{P\&L of Stein-Stein model (blue) vs signature Ornstein-Uhlenbeck volatility model (orange) quadratic hedging strategies for several maturities and strikes. $\kappa=1, \theta=0.25, \eta=1.2$ and $\rho=-0.6$.}
        \label{fig:hedging-euro-OU}
    \end{figure}
    
    \begin{table}[H]
        \centering
        \begin{tabular}{|c|c|c|c|c|c|c|}
            \hline
            \multirow{2}{*}{$J(X_0^*, \alpha^*)$}
            & \multicolumn{3}{|c|}{1 months} & \multicolumn{3}{|c|}{6 months} \\
            \cline{2-7}
                         & $K=0.9$     & $K=1.0$     & $K=1.1$     & $K=0.75$    & $K=1.0$     & $K=1.3$ \\
            \hline
            Stein-Stein  & 8.575e$-$05 & 2.067e$-$04 & 5.000e$-$05 & 1.484e$-$03 & 3.110e$-$03 & 1.859e$-$03 \\
            \hline
            Sig-vol, M=4 & 8.578e$-$05 & 2.067e$-$04 & 4.997e$-$05 & 1.489e$-$03 & 3.115e$-$03 & 1.900e$-$03 \\
            \hline
        \end{tabular}
        \caption{P\&L of Stein-Stein model (blue) vs signature Ornstein-Uhlenbeck volatility model (orange) quadratic hedging strategies for several maturities and strikes. $\kappa=1, \theta=0.25, \eta=1.2$ and $\rho=-0.6$.}
        \label{tab:hedging-euro}
    \end{table}

    Remark that both strategies mostly coincide. It illustrates that both the method in Section \ref{subsec:hedging-practical} works well within our framework and that the truncated signature representation of the volatility process is a good approximation when horizons are short enough, relative to the stiffness of the model's parameters.

\appendix

\section{Exact representations} \label{apn:linrep}
    
    \begin{proposition} \label{prop:coninv-shuexp}
        Whenever $\bell$ is of the form $\sum_{i \in \alphabet} \bell^{\word{i}} \word{i}$, with $\bell^{\word{i}} \in \R$, we have that 
        $$ \coninv{\emptyword - \bell} = \shuexp{\bell}. $$
        In particular, this implies
        $$ \shuexp{\bell} = \emptyword + \shuexp{\bell} \bell = \emptyword + \bell \shuexp{\bell}. $$
    \end{proposition}
    
    \begin{proof}
        First, it is easy to see by induction that $\bell \shupow{n} = n \bell \shupow{n-1} \bell$ for $n > 0$, i.e. $\bell \shupow{1} = \bell$ and
        \begin{align}
            \bell \shupow{n+1} &
            = \bell \shupow{n} \shuprod \bell 
            = n (\bell \shupow{n-1} \bell) \shuprod \bell
            = n (\bell \shupow{n-1} \bell + \bell \shupow{n-1} \shuprod \bell) \bell
            \\ &
            = (\bell \shupow{n} + n \bell \shupow{n}) \bell
            = (n+1) \bell \shupow{n} \bell.
        \end{align}
        It is then clear that $\frac{1}{n!} \bell \shupow{n} = \bell \conpow{n}$, hence proving the proposition.
    \end{proof}

\subsection{Proof of Proposition~\ref{prop:rep-OU}} \label{apn:OU}

    \begin{proposition}
        $\bsigma^\textnormal{OU} = \left( x \emptyword + \kappa \theta \word{1} + \eta \word{2} \right) \shuexp{-\kappa \word{1}}$ solves the equation 
        \begin{equation} \label{eq:OU-iterated-form}
            \bsigma^\textnormal{OU} = (x \emptyword + \kappa \theta \word{1} + \eta \word{2}) - \kappa \bsigma^\textnormal{OU} \word{1}.
        \end{equation}
    \end{proposition}

    \begin{proof}
        Straightforward by applying Proposition~\ref{prop:coninv-shuexp}.
    \end{proof}
    
    Using \cite[Theorem 4.2]{linearfbm}, one has
    $$ \bsigma^\textnormal{OU} = \left( x \emptyword + \kappa \theta \word{1} + \eta \word{2} \right) \shuexp{-\kappa \word{1}} \in \A. $$
    
    We can therefore define the process 
    $$ X_t = \bracketsig{\bsigma^\textnormal{OU}}. $$
    
    Now, showing
    \begin{align} \label{eq:sigma-OU-proj}
        \left( \bsigma^\textnormal{OU} \right)^\emptyword &
        = x \emptyword, \quad \bsigma^\textnormal{OU} \proj{1}
        = \kappa \theta \emptyword - \kappa \bsigma^\textnormal{OU},
        \\ \bsigma^\textnormal{OU} \proj{2} &
        = \eta \emptyword, \quad \bsigma^\textnormal{OU} \proj{22} 
        = 0.
    \end{align}
    makes it clear that $||\bsigma^\textnormal{OU}||_t^\I < \infty \text{ a.s.}$.
    Moreover, remarking that $\bsigma^\textnormal{OU} \proj{1} \word{1} = \bsigma^\textnormal{OU} - x \emptyword - \eta \word{2}$ by using \eqref{eq:OU-iterated-form}, allows us to write
    
    $$ \int_0^t\bracketsig[s]{\bsigma^\textnormal{OU} \proj{1}} \d s = \bracketsig{\bsigma^\textnormal{OU} \proj{1} \word{1}} < \infty, $$
    
    and thus have $\int_0^t ||\bsigma^\textnormal{OU}||_s^\I \d s < \infty \text{ a.s.}$ and $\bsigma^\textnormal{OU} \in \I$.

    We are now ready to apply Lemma~\ref{lem:sig-ito}
    \begin{align}
        \d Y_t &
        = d \bracketsig{\bsigma^\textnormal{OU}}
        \\ &
        = \bracketsig{ \bsigma^\textnormal{OU} \proj{1} + \tfrac{1}{2} \bsigma^\textnormal{OU} \proj{22}} \d t + \bracketsig{\bsigma^\textnormal{OU} \proj{2}} \d W_t.
        \\ &
        = \bracketsig{\kappa \theta \emptyword - \kappa \bsigma^\textnormal{OU}} \d t + \bracketsig{\eta \emptyword} \d W_t.
        \\ &
        = \left( \kappa \theta - \kappa \bracketsig{\bsigma^\textnormal{OU}} \right) \d t + \eta \d W_t
        \\ &
        = \kappa \left( \theta - Y_t \right) \d t + \eta \d W_t.
    \end{align}

    By uniqueness of the solution of the Ornstein-Uhlenbeck, the representation~\eqref{eq:linear-OU} follows. \\
    
    To get to the time-dependent representation~\eqref{eq:linear-OUtime}, we first observe that $\shuexp{\kappa \word{1}} \in \A$ for all $\kappa \in \R$, recall~\eqref{nota:shuexp}, i.e.
    
    $$ \norm{\shuexp{\kappa \word{1}}}_t^\A = \sum_{n=0}^\infty \left| \frac{(\kappa t)^n}{n!} \right| = e^{|\kappa t|} < \infty, $$
    and consequently that
    
    $$ e^{\kappa t} = \sum_{n \geq 0} \frac{(\kappa t)^n}{n!} = \bracketsig{\shuexp{\kappa \word{1}}}. $$

    Then, applying Proposition~\ref{prop:resolvent} together with~\eqref{eq:OU-iterated-form}, we can derive
    $$ \bsigma^\textnormal{OU} - \theta \emptyword = (x - \theta) \shuexp{ -\kappa \word{1}} + \eta \word{2} \shuexp{-\kappa \word{1}}. $$

    The final tool needed to complete this proof comes from the following Lemma.

    \begin{lemma} \label{lem:elli-exp}
        Let $\bell \in \eTA$, $\word{i} \in \alphabet$ and $\bm{b} := \sum_{\word{j} \in \alphabet} b^\word{j} \word{j}$ for $b^\word{j} \in \R$, then
        $$ \bell \word{i} \shuexp{\bm{b}} = \shuexp{\bm{b}} \shuprod \left( (\shuexp{-\bm{b}} \shuprod \bell) \word{i} \right). $$
    \end{lemma}
    
    \begin{proof}
        Let $\bgamma := \bell \word{i} \shuexp{\bm{b}} - \shuexp{\bm{b}} \shuprod \left( (\shuexp{-\bm{b}} \shuprod \bell) \word{i} \right)$
        \begin{align*}
            \bgamma &
            = \bell \word{i} \left( \emptyword + \shuexp{\bm{b}} \bm{b} \right) - \left( \emptyword + \shuexp{\bm{b}} \bm{b} \right) \shuprod \left( (\shuexp{-\bm{b}} \shuprod \bell) \word{i} \right)
            \\ &
            = \bell \word{i} + \bell \word{i} \shuexp{\bm{b}} \bm{b} - (\shuexp{-\bm{b}} \shuprod \bell) \word{i}
            - \left[ (\shuexp{\bm{b}} \bm{b}) \shuprod (\shuexp{-\bm{b}} \shuprod \bell) \right] \word{i}
            - \left[ \shuexp{\bm{b}} \shuprod \left( (\shuexp{-\bm{b}} \shuprod \bell) \word{i} \right) \right] \bm{b}
            \\ &
            = \bell \word{i} - (\shuexp{\bm{b}} \shuprod \shuexp{-\bm{b}} \shuprod \bell) \word{i} + \bgamma \bm{b}
            \\ &
            = \bgamma \bm{b}
            = 0.
        \end{align*}
    \end{proof}
    It is then easy to see that
    $$ (\bsigma^\textnormal{OU} - \theta \emptyword) \shuprod \shuexp{ \kappa \word{1}} = (x - \theta) \emptyword + \eta \shuexp{\kappa \word{1}} \word{2} $$
    
    Moreover, thanks to Proposition~\ref{prop:shufflepropertyextended}, and the fact that $\theta \emptyword \in \A$ for all $\theta \in \R$, $(\bsigma^\textnormal{OU} - \theta \emptyword) \shuprod \shuexp{\kappa \word{1}} \in \A$. This allows us to write
    $$ X_t = \theta + e^{-\kappa t} \bracketsig{\shuexp{\kappa \word{1}}} \left( \bracketsig{\bsigma^\textnormal{OU}} - \theta \right) = \bracketsig{\theta \emptyword + e^{-\kappa t} \left( (x - \theta) \emptyword + \eta \shuexp{\kappa \word{1}} \word{2} \right)}, $$
    which concludes the proof of representation~\eqref{eq:linear-OUtime}.

\subsection{Proof of Proposition~\ref{prop:rep-mGBM}} \label{apn:mGBM}

    \begin{proposition}
        $\bsigma^\textnormal{mGBM} = \left( y \emptyword + \left( \kappa \theta - \frac{\alpha \eta}{2} \right) \word{1} + \eta \word{2} \right) \shuexp{-\left( \kappa + \frac{\alpha^2}{2} \right) \word{1} + \alpha \word{2}}$ solves the equation 
        \begin{align}
            \bsigma^\textnormal{mGBM} = \left( y \emptyword + \left( \kappa \theta - \frac{\alpha \eta}{2} \right) \word{1} + \eta \word{2} \right) + \bsigma^\textnormal{mGBM} \left( - \left( \kappa + \frac{\alpha^2}{2} \right) \word{1} + \alpha \word{2} \right).
        \end{align}
    \end{proposition}

    \begin{proof}
        Straightforward by applying Proposition~\ref{prop:coninv-shuexp}.
    \end{proof}
    
    Using \cite[Theorem 4.2]{linearfbm}, one has
    $$ \bsigma^\textnormal{mGBM} = \left( y \emptyword + \left( \kappa \theta - \frac{\alpha \eta}{2} \right) \word{1} + \eta \word{2} \right) \shuexp{-\left( \kappa + \frac{\alpha^2}{2} \right) \word{1} + \alpha \word{2}} \in \A. $$
    
    We can therefore define the process 
    $ Y_t = \bracketsig{\bsigma^\textnormal{mGBM}}. $ Moreover, showing
    \begin{align} \label{eq:sigma-mGBM-proj}
         \left( \bsigma^\textnormal{mGBM} \right)^\emptyword &
         = y \emptyword, \quad \bsigma^\textnormal{mGBM} \proj{1} 
         = \left( \kappa \theta - \tfrac{\alpha \eta}{2} \right) \emptyword - \left( \kappa + \tfrac{\alpha^2}{2} \right) \bsigma^\textnormal{mGBM},
         \\ \bsigma^\textnormal{mGBM} \proj{2} &
         = \eta \emptyword + \alpha \bsigma^\textnormal{mGBM},
         \quad \bsigma^\textnormal{mGBM} \proj{22} 
         = \alpha \eta \emptyword + \alpha^2 \bsigma^\textnormal{mGBM},
    \end{align}
    makes it clear that $\norm{\bsigma^\textnormal{mGBM}}_t^\I < \infty$. As it requires much more care to prove $\int_0^t \norm{\bsigma^\textnormal{mGBM}}_s^\I \d s < \infty$, it is left for the interested reader to refer to \cite[Section 5]{linearfbm} for a detailed proof. \\
    
    We can thus finally use Lemma~\ref{lem:sig-ito}
    \begin{align}
        \d Y_t &
        = d \bracketsig{\bsigma^\textnormal{mGBM}}
        \\ &
        = \bracketsig{\bsigma^\textnormal{mGBM} \proj{1} + \tfrac{1}{2} \bsigma^\textnormal{mGBM} \proj{22}} \d t + \bracketsig{\bsigma^\textnormal{mGBM} \proj{2}} \d W_t.
        \\ &
        = \bracketsig{\left( \kappa \theta - \tfrac{\alpha \eta}{2} \right) \emptyword - \left( \kappa + \tfrac{\alpha^2}{2} \right) \bsigma^\textnormal{mGBM} + \tfrac{1}{2} \left( \alpha \eta \emptyword + \alpha^2 \bsigma^\textnormal{mGBM} \right)} \d t
        + \bracketsig{\eta \emptyword + \alpha \bsigma^\textnormal{mGBM}} \d W_t
        \\ &
        = \left( \kappa \theta - \kappa \bracketsig{\bsigma^\textnormal{mGBM}} \right) \d t + \left( \eta + \alpha \bracketsig{ \bsigma^\textnormal{mGBM}} \right) \d W_t
        \\ &
        = \kappa \left( \theta - Y_t \right) \d t + \left( \eta + \alpha Y_t \right) \d W_t 
    \end{align}

    By uniqueness of the solution of the mean-reverting geometric Brownian motion, the representation~\eqref{eq:linear-mGBM} follows. \\

    The time-dependent representation~\eqref{eq:linear-mGBMtime} follows from the fact that $\shuexp{\lambda \word{1}} \in \A$ for all $\lambda \in \R$ and an application of Proposition~\ref{prop:shufflepropertyextended}.

\section{Numerical implementation} \label{apn:algo}
    
    \begin{algorithm}[H] \label{alg:pricing}
        \caption{Signature-volatility Fourier pricing}
        \textbf{Input}:
        \begin{itemize}
            \item Fix $J > 0$ the number of points in the discretization of $[0, T]$,
            \item Fix $L > 0$ the number of points in the Gauss-Laguerre quadrature of the Fourier integral,
            \item Fix $M \geq 0$ the truncation order of $\bsigma$ and consequently $\tilde{M} = 2M$,
            \item Fix $\bsigma \in \R^J \times \R^N$ where $N = 2^{M+1} - 1$.
        \end{itemize}

        \textbf{Define}:
        \begin{itemize}
            \item $\Delta_j = t_{j-1} - t_j$,
            \item $F(t, \bpsi_t, u) := (\bpsi_t \proj{2}) \widetilde{\shuprod} (\frac{1}{2} \bpsi_t \proj{2} + f(t, u) \rho \bsigma_t) + \frac{1}{2} \bpsi_t \proj{22} + \bpsi_t \proj{1} + \left( \frac{f(t, u)^2 - f(t, u)}{2} + g(t, u) \right) \bsigma_t \shupow{2}$.
        \end{itemize}
        
        \textbf{Offline}:
        \begin{itemize}
            \item Compute the projection operators $\proj{\word{v}}: \tTA{\tilde{M}} \mapsto \tTA{\tilde{M} - |\word{v}|}$ for $\word{v} \in \{\word{1}, \word{2}, \word{22} \}$ as sequences of indices,
            \item Compute the truncated shuffle product operator $\widetilde{\shuprod}: \tTA{\tilde{M}} \times \tTA{\tilde{M}} \mapsto \tTA{\tilde{M}}$ as a sequence of 4-tuples, three coordinates and one count,
            \item Compute $\bsigma \shupow{2} \in \R^J \times \R^{\tilde{N}}$ where $\tilde{N} = 2^{\tilde{M}+1} - 1$,
            \item Compute the weights and points in the Gauss-Laguerre quadrature $(w_i, u_i)_{1 \leq i \leq L}$,
            \item Compute $C_0^{\textnormal{BS}}$ and $\phi_0^\textnormal{BS}(u_i)$ the price and characteristic function of the control variate for all $u_i$, $1 \leq i \leq L$.
        \end{itemize}

        \textbf{Online}:
        \begin{enumerate}
            \item Initialize $\bpsi \in \R^L \times \R^J \times \R^{\tilde{N}}$ where $\bpsi_{t_J} = \bpsi_T = 0$,
            \item Run the Runge-Kutta method of order 4 on $F$ for all $u_i$, $1 \leq i \leq L$, i.e.
            \begin{enumerate}
                \item $k_1 \leftarrow F(t_j, \bpsi_{t_j}, u)$,
                \item $k_2 \leftarrow F(t_j + \frac{\Delta_j}{2}, \bpsi_{t_j} + \frac{\Delta_j}{2} k_1, u)$,
                \item $k_3 \leftarrow F(t_j + \frac{\Delta_j}{2}, \bpsi_{t_j} + \frac{\Delta_j}{2} k_2, u)$,
                \item $k_4 \leftarrow F(t_{j-1}, \bpsi_{t_j} + \Delta_j k_3, u)$,
                \item $\bpsi_{t_{j-1}} \leftarrow \bpsi_{t_j} + \frac{\Delta_j}{6} (k_1 + 2 k_2 + 2 k_3 + k_4)$.
            \end{enumerate}
            \item Compute the Fourier/Laplace integral in the rhs of \eqref{eq:control-variate} or \eqref{eq:laplace_qvol} as a weighted sum on all $(w_i, u_i)_{1 \leq i \leq L}$.
        \end{enumerate}
    \end{algorithm}

\section{Asian options} \label{apn:asian}

    In this section, we illustrate Sections~\ref{sec:pricing} and \ref{sec:hedging} with Asian options on the geometric average $\frac{1}{T} \int_0^T \log S_s \d s$. It can be easily recovered by setting $f(s) := iu \frac{T-s}{T}$ in \eqref{eq:charfun} since 
    \begin{align} \label{ex:geom-average}
        \int_t^T f(s) \d \log S_s = iu \frac{1}{T} \int_t^T \log S_s \d s - iu \frac{T-t}{T} \log S_t.
    \end{align}

    We will denote the geometric average by $\bar{S}_t = \exp \left( \frac{1}{T} \int_0^t \log S_s \d s \right)$ and $\bar{M}_t(u)$ for \eqref{eq:charfun} with $f(s) := iu \frac{T-s}{T}$ and $g=0$. Finally, the price of the geometric Asian call option will be denoted by $\bar{C}_t(T, K) = \E [(\bar{S}_T - K)^+ | \F_t]$.

\subsection{Pricing} \label{subsec:pricing-asiat}

    Similarly to \eqref{eq:control-variate}, we obtain the Fourier representation
    \begin{equation} \label{eq:control-variate-asian}
        \bar C_t(\bar{S}_t; T, K) = \bar C_t^{\textnormal{BS}}(\bar{S}_t; T, K) - \frac{K}{\pi} \int_0^\infty \Re \left[ e^{i (u - \frac{i}{2}) \bar{k}_t} \left( \bar \phi_t \left( u - \tfrac{i}{2} \right) - \bar \phi_t^{\textnormal{BS}} \left( u - \tfrac{i}{2} \right) \right) \right] \frac{\d u}{\left( u^2 + \tfrac{1}{4} \right)},
    \end{equation}
    where $\bar{\phi}$ is the characteristic function $\bar{\phi}_t(u) = \E \Big[ e^{iu \log \frac{\bar{S}_T}{\bar{S}_t}} \Big| \F_t \Big]$ and $\bar{k}_t = \log \frac{\bar S_t}{K}$
    with
    \begin{equation} \label{cf-black-scholes-asiat}
        \bar{\phi}_t^{\textnormal{BS}} = \exp \left( - \frac{\sigma_{\textnormal{BS}}^2}{2} \left( u^2 \frac{(T-t)^3}{3T^2} + iu \frac{(T-t)^2}{2T} \right) \right),
    \end{equation}
    and where
    \begin{equation} \label{eq:bar-C-BS}
        \bar{C}_t^{\textnormal{BS}}(\bar{S}_t; T, K) := \mathcal{N}(\bar{d}_1) \bar S_t \exp \left( \frac{T-t}{T} \log S_t - \frac{\sigma^2_{\textnormal{BS}}}{2} \left( \frac{(T-t)^2}{2 T} - \frac{(T-t)^3}{3 T^2} \right) \right) + \mathcal{N}(\bar{d}_2) K,
    \end{equation}
    and
    \begin{equation}
        \bar{d}_2 := \frac{1}{\sigma_{\textnormal{BS}} \sqrt{\frac{(T-t)^3}{3 T^2}}} \left( \log \frac{\bar{S}_t}{K} + \frac{T-t}{T} \log S_t - \frac{\sigma^2_{\textnormal{BS}}}{2} \frac{(T-t)^2}{2T} \right), \quad \bar{d}_1 := \bar{d}_2 - \sigma_{\textnormal{BS}} \sqrt{\frac{(T-t)^3}{3T^2}}.
    \end{equation}

    In Figure \ref{fig:pricing-asian-OU}, we compare our signature volatility Fourier pricing for Asian options using the truncated, at order $M=4$, representation of the linear Ornstein-Uhlenbeck \eqref{eq:linear-OU} to Monte Carlo simulations. Lewis' approach together with Black-Scholes control variate was also used.
    
    \begin{figure}[H]
        \centering
        \subfloat[\centering Maturity 1 week.]{{\includegraphics[width=\twoplotswidth]{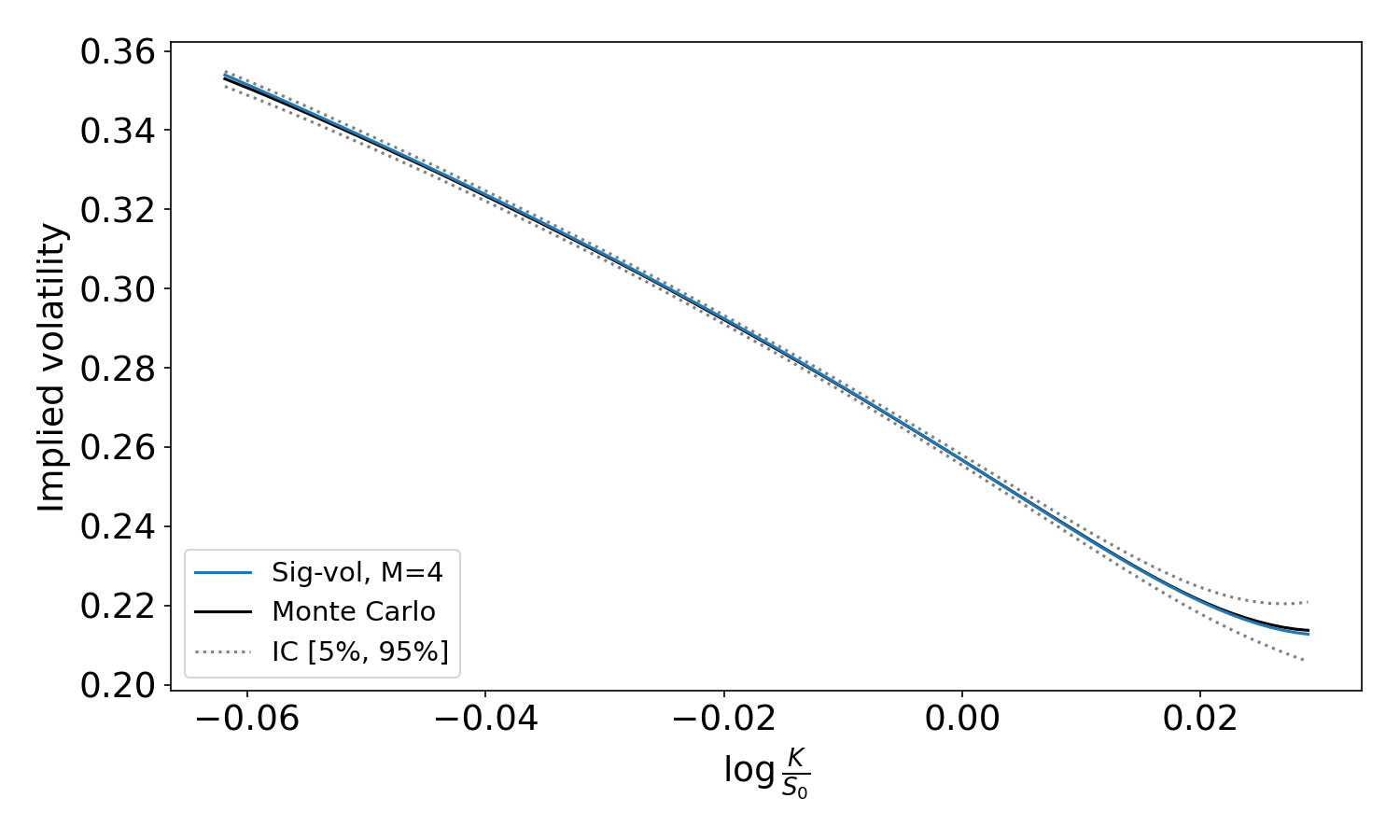}}}
        \quad
        \subfloat[\centering Maturity 1 year.]{{\includegraphics[width=\twoplotswidth]{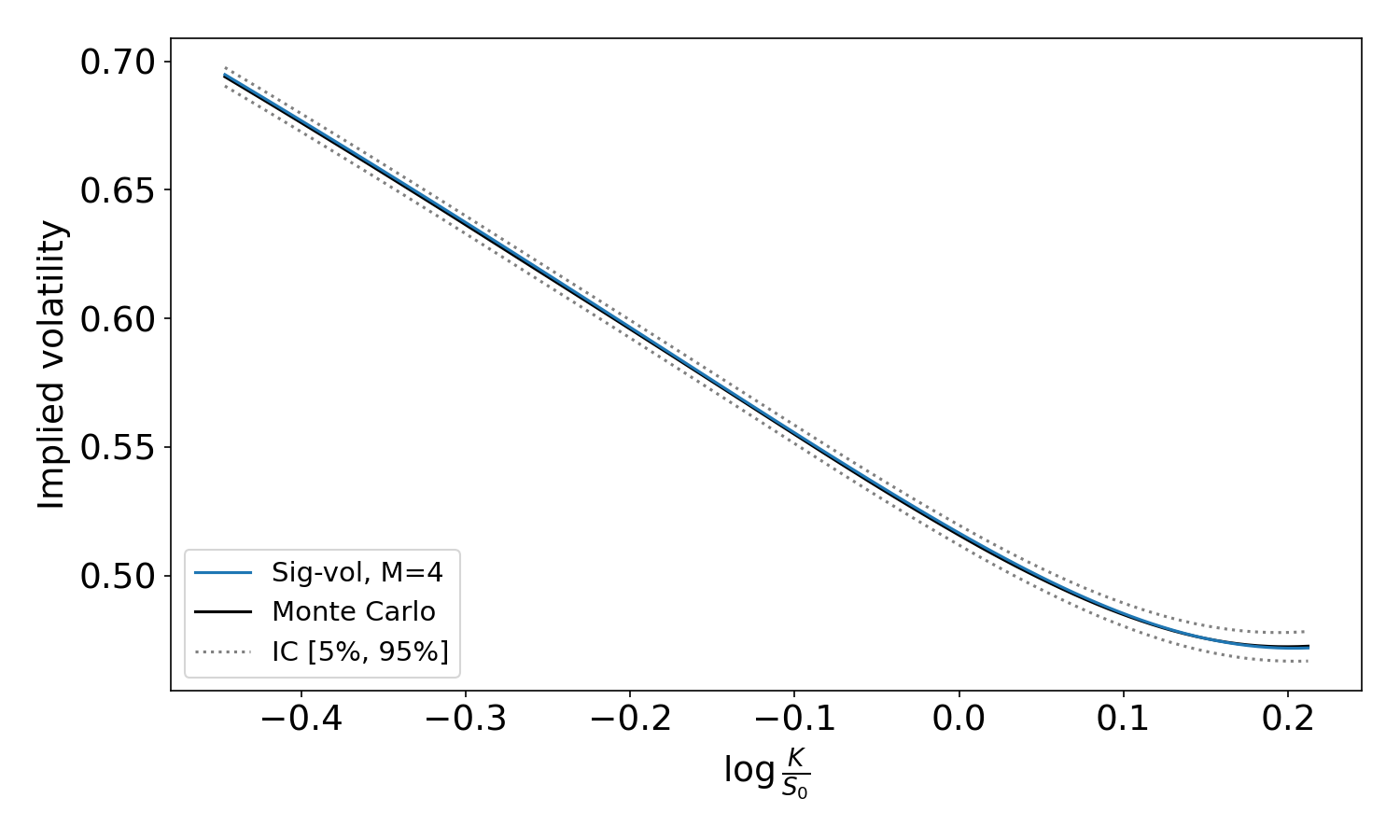}}}
        \caption{Asian put Monte Carlo pricing of Stein-Stein model vs signature Ornstein-Uhlenbeck volatility model. $\kappa=1, \theta=0.25, \eta=1.2$ and $\rho=-0.9$.}
        \label{fig:pricing-asian-OU}
    \end{figure}

\subsection{Hedging} \label{subsec:hedging-asiat}

    In the same spirit as for the European call option, we will express the Fourier inversion formula \eqref{eq:control-variate} in terms of the process $\bar{M}_t(u)$ and apply Itô. Using the same notations as in Section \ref{subsec:pricing-asiat}, for all $t \leq T$
    \begin{align}
        \exp \left( \bracketsig{\bpsi_t} \right) &
        = \E \left[ \left. \exp \left( \int_t^T f(s) \d \log S_s \right) \right| \F_t \right]
        \\ &
        = \E \left[ \left. \exp \left( iu \log \frac{\bar{S}_T}{\bar{S}_t} - iu \frac{T-t}{T} \log S_t \right) \right| \F_t \right],
    \end{align}
    so that, recall \eqref{eq:U},
    \begin{align*}
        \bar{\phi}_t(u) &
        = \exp \left( \bracketsig{\bpsi_t} + iu \frac{T-t}{T} \log S_t \right)
        = \bar{M}_t(u) \exp \left( - \int_0^t f(s) \d \log S_s + iu \frac{T-t}{T} \log S_t \right)
        = \bar{M}_t(u) e^{-iu \log \bar{S}_t}.
    \end{align*}
    
    Thus one can now write $\bar{C}$ as a Fourier integral on $\bar{M}$
    \begin{align*}
        \bar{C}_t(\bar{S}_t; T, K) &
        = \bar{S}_t - \frac{K}{\pi} \int_0^\infty \Re \left[ e^{-i (u - \frac{i}{2}) \log K} M_t \left( u - \tfrac{i}{2} \right) \right] \frac{\d u}{\left( u^2 + \tfrac{1}{4} \right)},
    \end{align*}
    together with its Black-Scholes control variate version
    \begin{equation} \label{eq:control-variate-M-asiat}
        \bar{C}_t(\bar{S}_t; T, K) = \bar{C}_t^{\textnormal{BS}}(\bar{S}_t; T, K) - \frac{K}{\pi} \int_0^\infty \Re \left[ e^{i (u - \frac{i}{2}) \log K} \left( \bar{M}_t \left( u - \tfrac{i}{2} \right) - \bar{M}_t^{\textnormal{BS}} \left( u - \tfrac{i}{2} \right) \right) \right] \frac{\d u}{\left( u^2 + \tfrac{1}{4} \right)}
    \end{equation}
    with $\bar{M}_t^{\textnormal{BS}}(u) := \bar{\phi}_t^{\textnormal{BS}}(u) e^{iu \log \bar{S}_t}$, where $\bar{\phi}^{\textnormal{BS}}$ is defined in \eqref{cf-black-scholes-asiat} and $\bar{C}^{\textnormal{BS}}$ in \eqref{eq:bar-C-BS}. Finally, $w(u)$ and $\tilde{u}$ are defined as in Section~\ref{subsec:hedging-practical} and with very similar computations, one can get $\bar{Z}$ and $\bar{Z}^\perp$ as follows
    \begin{align*}
        \bar{Z}_t &
        = \Sigma_t S_t \bar{\Delta}_t^{\textnormal{BS}} \rho + \int_0^\infty \Re \left[ \bar{\zeta}_t(\tilde{u}) w(u) \right] \d u
        \quad \bar{Z}_t^\perp 
        = \Sigma_t S_t \bar{\Delta}_t^{\textnormal{BS}} \sqrt{1 - \rho^2} + \int_0^\infty \Re \left[ \bar{\zeta}_t^\perp(\tilde{u}) w(u) \right] \d u,
        \end{align*}
        with
            \begin{align*}
        \bar{\zeta}_t(u) :&
        = \frac{T-t}{T} iu \Sigma_t \left( \bar{M}_t(u) - \bar{M}_t^{\textnormal{BS}}(u) \right) \rho + \bar{M}_t(u) \bracketsig{\bpsi_t(u) \proj{2}}
        \\ \bar{\zeta}_t^\perp(u) :&
        = \frac{T-t}{T} iu \Sigma_t \left( \bar{M}_t(u) - \bar{M}_t^{\textnormal{BS}}(u) \right) \sqrt{1 - \rho^2},
    \end{align*}
    where $\bar{\Delta}_t^{\textnormal{BS}} = \frac{\partial}{\partial \bar{S}_t} \bar{C}_t^{\textnormal{BS}}$.
    
    This allows us to solve the quadratic hedging problem in \eqref{eq:hedgingoptimal} numerically for Asian call and put options in the framework of sig-volatility models. \\

    In Figure \ref{fig:hedging-asiat-mGBM}, we simulate price trajectories under mean-reverting geometric Brownian motion volatility model and compare the performance of the Black-Scholes hedging strategy, simply as a point of reference with $\sigma_{\textnormal{BS}} = \theta$, to the Fourier hedging of the signature mean-reverting geometric Brownian motion volatility model for an Asian put option.

    Remark that the signature volatility systematically outperforms, in terms of minimized squared P\&L, the naive Black-Scholes quadratic hedging strategy by 10 to 25\%, see Table~\ref{tab:hedging-asiat}. This suggests that our framework might specifically be relevant for path-dependent options and more complex volatility dynamics where hedging strategies are not  known explicitly or tractable.
    
    \begin{figure}[H]
        \centering
        \subfloat[\centering $T=$ 1 month, $K=0.95$.]{{\includegraphics[width=\threeplotswidth]{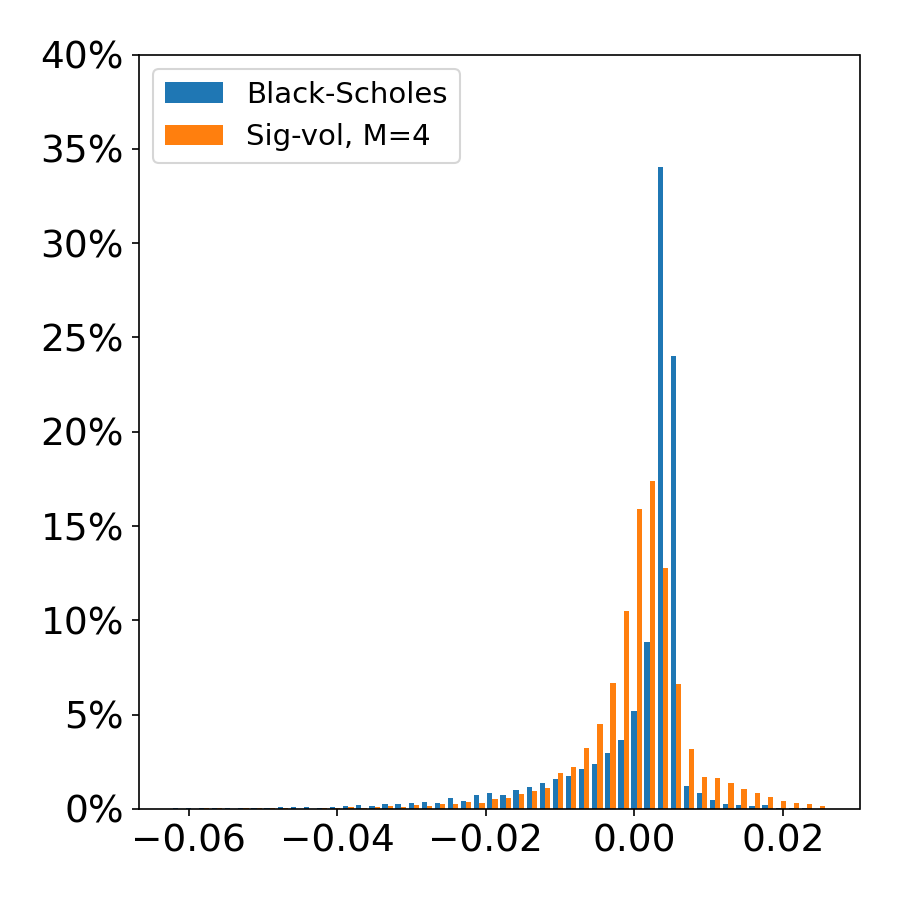}}}
        \quad
        \subfloat[\centering $T=$ 1 month, $K=1.0$.]{{\includegraphics[width=\threeplotswidth]{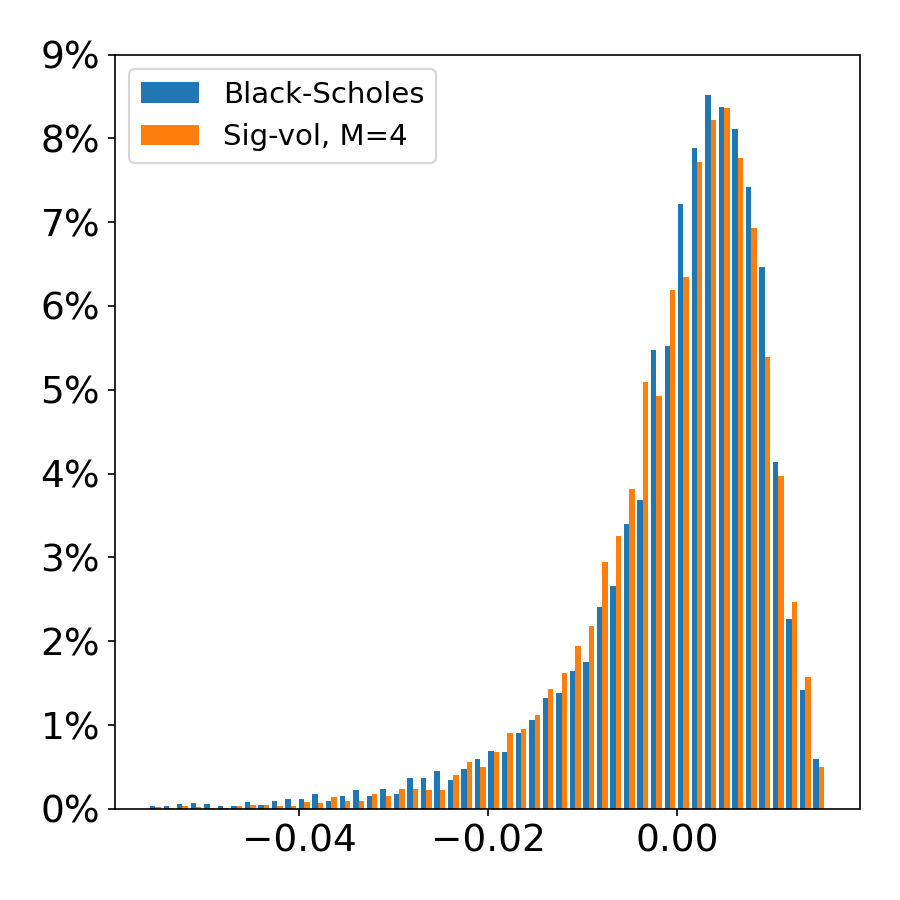}}}
        \quad
        \subfloat[\centering $T=$ 1 month, $K=1.05$.]{{\includegraphics[width=\threeplotswidth]{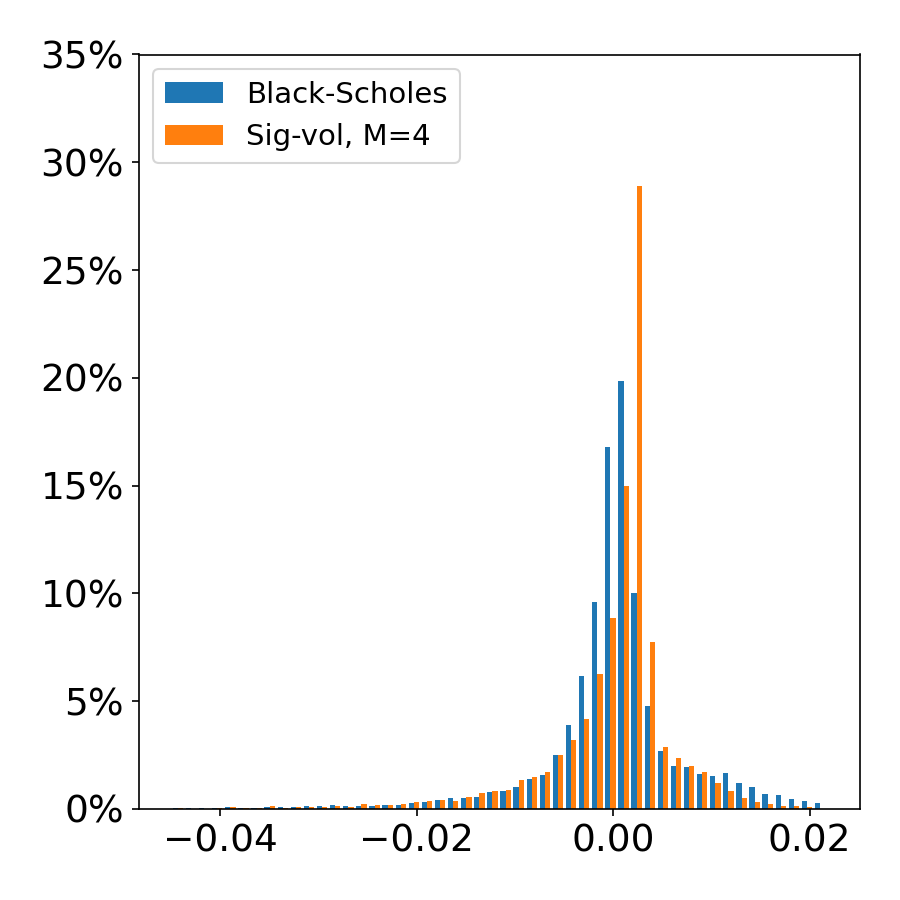}}}
        \\
        \subfloat[\centering $T=$ 6 months, $K=0.85$.]{{\includegraphics[width=\threeplotswidth]{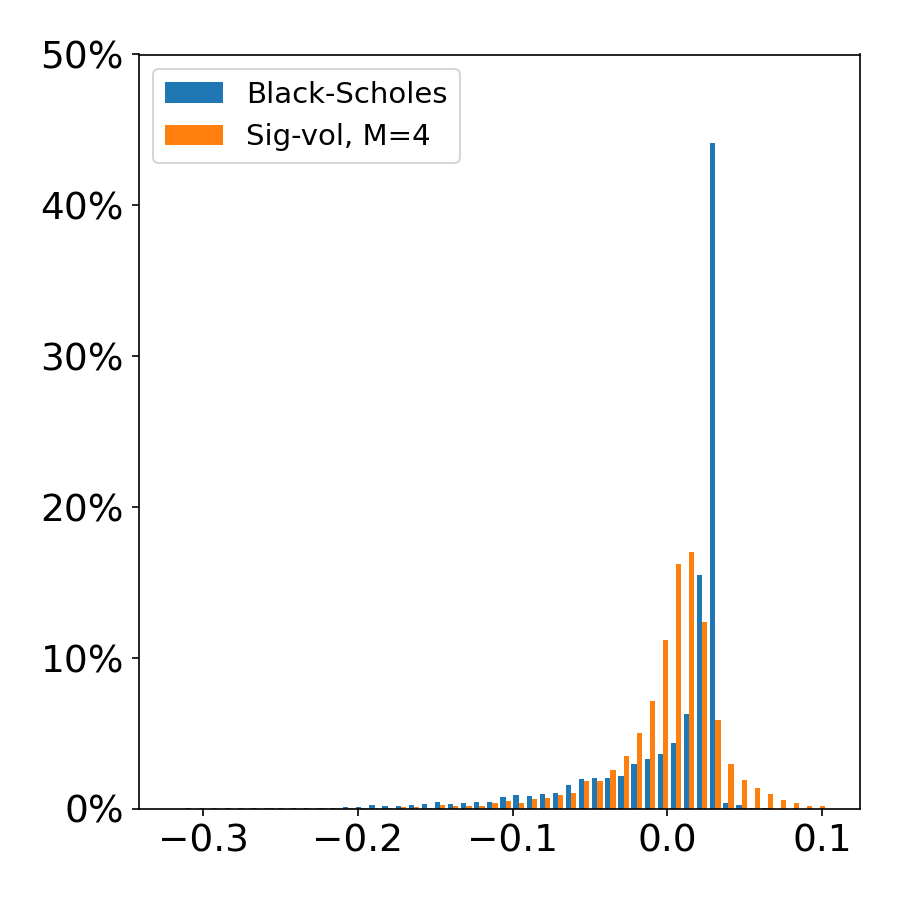}}}
        \quad
        \subfloat[\centering $T=$ 6 months, $K=1.0$.]{{\includegraphics[width=\threeplotswidth]{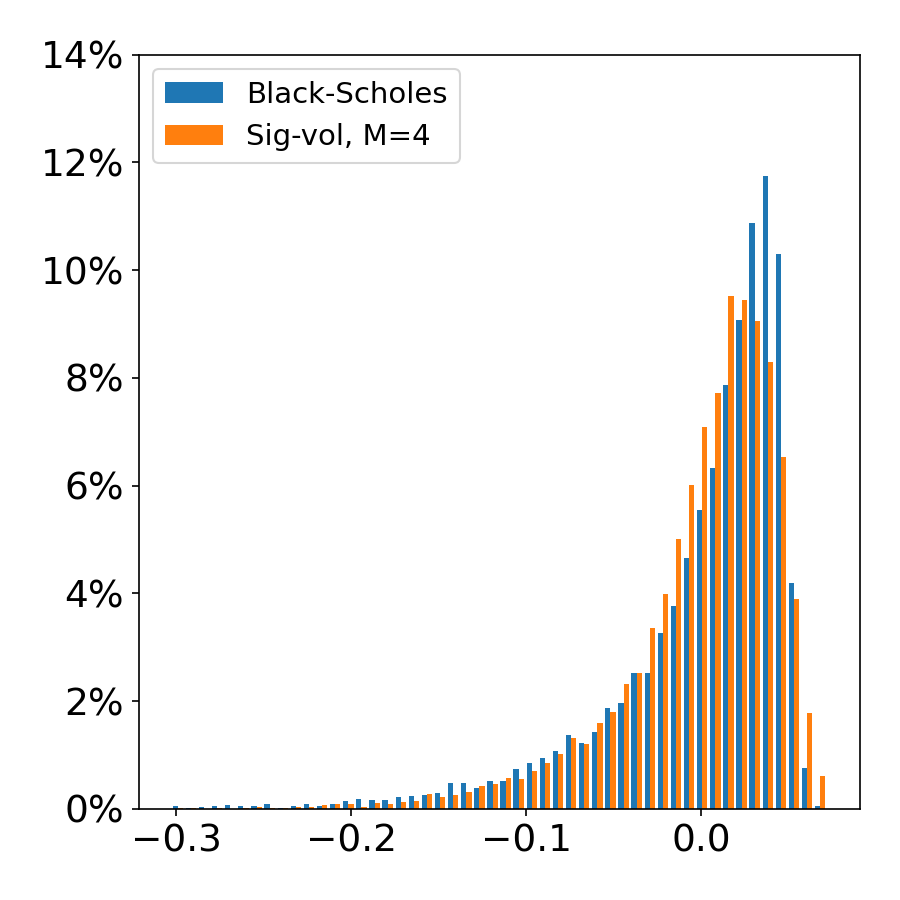}}}
        \quad
        \subfloat[\centering $T=$ 6 months, $K=1.15$.]{{\includegraphics[width=\threeplotswidth]{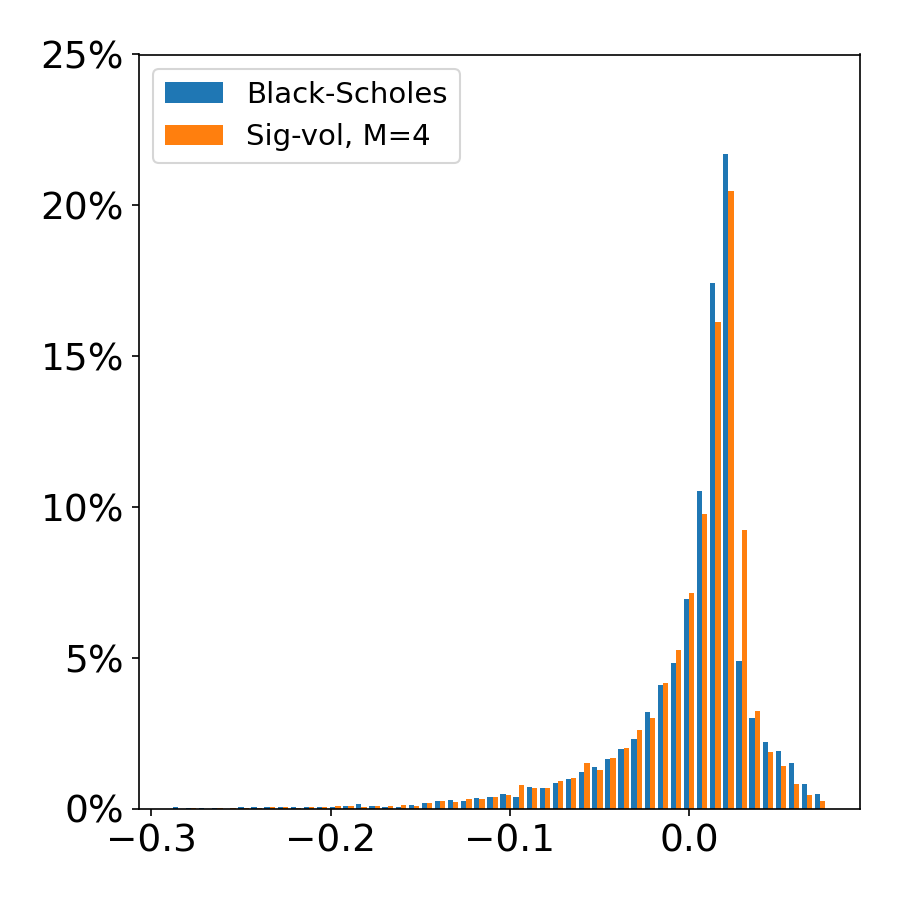}}}
        \caption{P\&L of Black-Scholes model (blue) vs signature mean-reverting geometric Brownian motion volatility model (orange) quadratic hedging strategies for several maturities and strikes. $\kappa=1, \theta=0.25, \eta=1.2, \alpha=0.6$ and $\rho=-0.6$.}
        \label{fig:hedging-asiat-mGBM}
    \end{figure}
    
    \begin{table}[H]
        \centering
        \begin{tabular}{|c|c|c|c|c|c|c|}
            \hline
            \multirow{2}{*}{$J(X_0^*, \alpha^*)$}
            & \multicolumn{3}{|c|}{1 months} & \multicolumn{3}{|c|}{6 months} \\
            \cline{2-7}
                          & $K=0.95$    & $K=1.0$     & $K=1.05$    & $K=0.85$    & $K=1.0$     & $K=1.15$    \\
            \hline
            Black-Scholes & 8.727e$-$05 & 1.000e$-$04 & 5.456e$-$05 & 2.446e$-$03 & 2.846e$-$03 & 1.867e$-$03 \\
            \hline
            Sig-vol, M=4  & 6.710e$-$05 & 8.472e$-$05 & 4.510e$-$05 & 1.453e$-$03 & 2.182e$-$03 & 1.657e$-$03 \\
            \hline
        \end{tabular}
        \caption{P\&L of Black-Scholes model vs signature mean-reverting geometric Brownian motion volatility model quadratic hedging strategies for several maturities and strikes. $\kappa=1, \theta=0.25, \eta=1.2, \alpha=0.6$ and $\rho=-0.6$.}
        \label{tab:hedging-asiat}
    \end{table}

\bibliographystyle{plainnat}
\bibliography{sigvol.bib}

\end{document}